\documentclass[preprint,11pt,authoryear]{elsarticle}
\usepackage{enumitem} %control over list
\usepackage{ulem} % to underline
\usepackage{hyperref}
\makeatletter
\let\c@author\relax
\makeatother

\usepackage[a4paper, top=2cm, bottom=2.5cm, left=2.5cm, right=2.5cm]{geometry}

\usepackage{float}
\usepackage{adjustbox} 

\usepackage[onehalfspacing]{setspace}
\graphicspath{{figures/}}
\usepackage{amsthm}
\usepackage{xcolor}

\newtheorem{example}{Example}

\newtheorem{definition}{Definition}
\newtheorem{remark}{Remark}

\newtheorem{theorem}{Theorem}
\newtheorem{corollary}{Corollary}
\newtheorem{proposition}{Proposition}
\newtheorem{lemma}{Lemma}

\usepackage{amsmath,amssymb,amsfonts}
\usepackage[linesnumbered]{algorithm2e}

\usepackage{appendix}

\newcommand{\R}{\mathbb{R}}

\usepackage{url}

\newcommand{\quasistar}{cash-quasiconcavity} %at $k \in \mathbb{R}$

\newcommand{\viola}[1]{\textcolor{violet}{#1}}

\usepackage{cancel}
\usepackage[final]{changes} %
\usepackage[normalem]{ulem}

\usepackage{booktabs}

\begin{document}
\begin{frontmatter}
\title{{Ranking Metrics: Extending Acceptability and Performance Indices}}
%Metrics for Ranking Financial and Insurance Performance

\author[1]{Asmerilda Hitaj}
\ead{asmerilda.hitaj@uninsubria.it}

\author[2]{Elisa Mastrogiacomo}
\ead{elisa.mastrogiacomo@uninsubria.it}

\author[3]{Ilaria Peri}
\ead{I.Peri@bbk.ac.uk}

\author[4]{Marcelo Righi}
\ead{marcelo.righi@ufrgs.br}

%\affiliation[1]{organization={Birkbeck Business School, Birkbeck University of London},%Department and Organization
%            addressline={Malet St, Bloomsbury}, 
%            city={London},
%            postcode={WC1E 7HX}, 
 %           country={United Kingdom}}

 %\author[2]{xxx\corref{cor1}}
% \cortext[cor1]{Corresponding author}
%\ead{author}
% \affiliation[2]{organization={organization},
% addressline = {addressline},
%city = {city},
%postcode={postcode},
%country = {country}} 
            
\begin{abstract}
This paper develops an axiomatic framework for ranking metrics, a general class of functionals for evaluating and ordering financial or insurance positions. Unlike traditional risk-adjusted performance measures, such as the Sharpe ratio, RAROC, or Omega, that express reward per unit of risk, ranking metrics assign each position a performance level rather than a normalized return. Relying on monotonicity and a new property called cash-quasiconcavity, we derive representation results linking ranking metrics to families of acceptance sets and risk measures, extending the theory of acceptability indices. Classical ratios arise as special cases, while new examples -- based on expected-loss, $\Lambda$-quantile, and bibliometric indices -- illustrate the framework's flexibility. Empirical applications to portfolio ranking and climate-risk insurance demonstrate its practical relevance.
\end{abstract}


\end{frontmatter}

\section{Introduction}

Comparing the performance of alternative positions or institutions is a fundamental task in finance, insurance, and related fields.
Classical risk-adjusted performance measures, including the Risk-Adjusted Return on Capital (RAROC) \citep{zaik1996raroc}, the Omega ratio, \cite{keating2002introduction}, 
and the Gain–Loss Ratio (GLR), \cite{Bernardo2000}, 
provide alternative ways of assessing performance by expressing returns relative to a specific notion of risk, thereby enabling consistent comparison of investment alternatives. Despite their practical appeal, these indices remain conceptually fragmented, having arisen as heuristic tools reflecting specific interpretations of risk and reward.

A significant step toward theoretical unification was achieved by \cite{ChernyMadan2009} 
through the notion of acceptability indices, which assign each position an acceptability level corresponding to the smallest performance threshold under which the position is considered acceptable. This axiomatic framework embeds classical ratios within a broader and unified structure and shows how properties, such as monotonicity, follow naturally from a reward-to-risk interpretation. This approach was further extended by \cite{RosazzaSgarra2013}, who relaxed  the scale invariance property and developed a dynamic framework based on $g$-expectations. Subsequently, \cite{PeriThesis} and \cite{Frittelli2014SRM} introduced a minimal definition for acceptability indices as monotone and quasi-concave maps, providing a dual representation analogous to the risk measure framework in \cite{DrapeauKupper2013}.  Recent extensions 
relax the quasiconcavity property in favour of the more general star-shapedness \citep{Righi2024}. 

Although the acceptability index framework represents a major theoretical advance, its broader potential remains largely underexplored. Most studies continue to interpret performance as a normalized measure of reward or a monetary value, which limits applicability in settings such as regulatory stress testing, climate-risk assessment, or ESG benchmarking -- where evaluation focuses on meeting specific thresholds or ranking levels rather than monetary or risk-adjusted performance.

The main contribution of this paper is the introduction of a broader concept which we refer to as \textit{ranking metric}. A ranking metric is a functional that assigns to each random outcome a positive real number representing its rank, score, or degree of acceptability, without requiring a reward-to-risk interpretation. It is characterized by two properties: monotonicity, ensuring that better outcomes are never penalized, and cash–quasiconcavity (or quasiconcavity at constants), which requires that mixing a random outcome with a certain one cannot lower its rank below the minimum of its components -- an idea that resembles the concave counterpart of quasi-star-shapedness by \cite{Han2024}. This minimal and economically meaningful mixing condition generalizes the traditional quasiconcavity and star-shapedness assumptions, allowing for a broader class of ranking rules that extend beyond classical performance measures and acceptability indices. While mathematically these metrics correspond to \viola{quasi-concave-star-shaped} acceptability indices, we propose this shift in terminology to enhance interpretability. Under this framework, performance can be assessed through ordinal rankings or qualitative levels rather than being restricted to normalized rewards or monetary values, thereby accommodating policy-driven objectives that extend beyond purely financial magnitudes.
 %\footnote{Mathematically, this class of metrics generalizes star-shaped acceptability indices. We intentionally avoid the term quasi–star-shaped acceptability indices to highlight a conceptual shift — separating ranking from the traditional reward-to-risk interpretation. The ranking metric framework accommodates a diverse range of normalized reword, such as RAROC and GLR, monetary measures, such as certainty equivalent and quantiles, but including also  bibliometric-inspired metrics.}
%This normative framework allows performance rankings to be based on indicators that serve as ordinal measures reflecting also qualitative or policy-driven objectives rather than financial magnitudes.

The remainder of the paper is structured as follows. In Section \ref{Sec: Rank-metric}, we introduce the definition of ranking metric, study the property of cash-quasiconcavity, its decision theoretic justification  and explore its relationship with other established properties in the performance measurement literature. Our aim is also to clarify the interpretation of quasi-star-shapedness and star-shapedness, properties recently introduced into the theory of risk measures and acceptability indices \citep{Castagnoli2022,Han2024,Righi2024} and their link to cash-quasiconcavity. Then, in Section \ref{sec:rank-metric-charact}, we provide a unified axiomatic framework for ranking metrics. We show that any monotone, cash-quasiconcave ranking metric admits representations in terms of decreasing families of cash-convex acceptance sets and increasing families of cash-quasiconvex risk measures. This enables the construction of new performance measures that allow for deeper comparisons across portfolios, institutions, and contexts. In Section \ref{sec:Rank_Metr_exam}, we present several examples of ranking metrics, encompassing both classical reward-to-risk performance ratios (such as RAROC, GLR, and Omega) and newly constructed metrics generated by families of risk measures and representing monetary values -- including expected-loss and certainty-equivalent functionals and $\Lambda$-quantiles. 
%and eligible-asset risk measures. 
We also identify the conditions under which these ranking metrics satisfy cash-(sub or sup)-additivity properties. Particular attention is devoted to a novel class of ranking metrics inspired by bibliometric indices (e.g., the $h$-index and related measures), discussed in Subsection \ref{sub_biblio_index}. Finally, Section \ref{sec: emp_app} illustrates the framework through empirical analyses based on real data. Although the concept of a ranking metric is general and applicable across various domains, our focus is on financial performance and climate-risk resilience. The first application compares portfolio rankings obtained from different metrics, examining their interpretation and implications for practitioners. We show that identical datasets can yield markedly different orderings depending on the chosen metric, highlighting the economic relevance of selecting a ranking rule consistent with specific performance objectives or constraints. A complementary analysis compares European regions’ climate-risk-related loss, while a final study explores optimal portfolio allocation under alternative ranking metrics.

\section{Ranking metrics: basic definitions and related properties} \label{Sec: Rank-metric}
%\textcolor{blue}{Decide: assume if we want to take the monotonicity increasing or decreasing because from this depends how we consider. In the empirical part we need to consider two different $\alpha$. We need to highlight that $\lambda$VaR is very selective for low $\alpha$.}

Let $(\Omega, \mathcal{F}, \mathbb{P})$ be a probability space, and let $L^{\infty} := L^{\infty}(\Omega, \mathcal{F}, \mathbb{P})$
denote the space of $\mathcal{F}$-measurable random variables that are $\mathbb{P}$-almost surely ($\mathbb{P}$-a.s.) bounded. Each $X \in L^{\infty}$ represents the profit-and-loss of a risky financial position or the (negative) loss of an insurance contract. Unless stated otherwise, all equalities and inequalities are understood $\mathbb{P}$-almost surely. Throughout this paper,
%we will denote by $L^\infty$ a subset of $L^\infty$ and by $r$ a functional from $L^\infty$ with values in $[0,+\infty]$. We 
we will indicate  by $\text{supp}(r)$ the support of the functional $r: L^\infty \to [0, +\infty]$, which, given the non-negativity of $r$, is defined by
$\text{supp}(r) := \{ X \in L^\infty : r(X) > 0 \}.$ We will also adopt the convention $\sup \emptyset = 0$ and $\sup  \mathbb{R}_+=+\infty$. 

We propose a definition of a \textit{performance ranking metric} or simply \textit{ranking metric} that meets essential criteria for effectively ranking financial and insurance performances.
\begin{definition} \label{def:RM}
A map $r:  L^{\infty} \to [0,+\infty]$ is called a ranking metric if it satisfies:
\begin{itemize}
    \item[i)] (increasing) monotonicity: $X \leq Y  \implies r(X) \leq r(Y)$, $\, X,Y \in L^\infty$, 
    \item[ii)] cash-quasiconcavity: $ r(\lambda X + (1-\lambda) k) \geq \min \{ r(X),  r(k) \}, \ X \in L^\infty ,\ k\in \R, \ \lambda \in [0,1]$.
\end{itemize}
\end{definition}
Unlike traditional performance measures or acceptability indices, ranking metrics assign to a risky financial profit-and-loss or an insurance loss a value that does not necessarily correspond to a monetary amount; instead, it may represent a relative level of performance.
%Substantially, ranking metric assigns a performance rank to a risky financial profit-and-loss or an insurance loss. Unlike a traditional performance measure or acceptability index, the ranks need not correspond to monetary values and may instead represent relative levels of performance. 
The choice of ranking metric may reflect investor preferences, but it must satisfy monotonicity and cash-quasiconcavity to ensure meaningful performance comparisons between different financial positions. Monotonicity is the most natural property as it guarantees that higher ranks reflect more attractive portfolios. 
On the other hand, cash-quasiconcavity captures the structural preference of asset managers for integrating risk-free positions within a portfolio. This property generalizes standard quasiconcavity stating that any linear convex combination of a risky asset and a deterministic payoff must yield a rank at least as high as the minimum rank of the two.   %to the space of risky and risk-free assets or more generally, assets with deterministic payoff $k \in \R$ with performance $r(k)$ - for instance fixed-income instruments. It implies that any linear convex combination of a risky asset and a deterministic payoff must yield a rank at least as high as the minimum rank of the two.
 %This condition captures the intuition that diversification with a risk-free component should not reduce
%\blue{steady convexity, weighted stability???}
%\hyperlink{item:(QCV)}{\customlabel{(QCV)}}.

\begin{remark}
Cash-quasiconcavity is a weaker requirement than its classical counterpart, as standard quasiconcavity (see Definition \ref{def:RaM_other_prop}) directly implies cash-quasiconcavity. Consequently, ranking metrics generalize the definition of acceptability indices in \cite{PeriThesis} and \citep{frittelli2016scientific}, which are only required to be monotone and quasiconcave. Hence, they also generalize the original framework of \cite{ChernyMadan2009}  which further imposes scale invariance, i.e., $r(\lambda X) = r(X)$ for all $\lambda > 0$, and the Fatou property for dual representation. 
\end{remark}

From a behavioral perspective, we proved in the following proposition that cash-quasiconcavity has also decision-theoretic justification. This property characterizes managers who exhibit a preference for diversification across risky and deterministic positions. % or risk-free positions.
Specifically, when indifferent between a certain and an uncertain outcome with respect to the ranking metric, % —anticipating equivalent performance from both — 
managers tend to prefer a convex mixture of the two rather than a full allocation to the uncertain position.   
%\viola{[è simile a Prop 3.1 in Wang]}
Formally, we assume that asset managers evaluate their investments using a ranking metric $r: L^\infty \to [0,+\infty]$. % such that an investment $X$ is preferred to $Y$ if and only if $r(X) \ge r(Y)$. Mathematically, this ranking metric $r$ induces a 
The preference relation $\succeq$ induced by $r$ is then defined as: for any  $X,Y \in L^\infty$
\[
X \succeq Y \quad \Longleftrightarrow \quad r(X) \ge r(Y).
\]
Consequently, $\succeq$ represents the ordering of an agent who evaluates positions according to $r$ and prioritizes those with higher ranks. Under this framework, we demonstrate below that requiring the ranking metric to satisfy cash-quasiconcavity is equivalent to assuming that agents prefer combining risky investments with certain ones, even if they are indifferent to them in isolation. %Mathematically, if a random payoff $X$ is indifferent to a deterministic one $k > 0$ (denoted $X \simeq k$), then any linear convex mixture $\lambda X + (1-\lambda)k$, for $\lambda \in [0,1]$, is weakly preferred to $X$.
In other words, integrating a risky position (for instance, a stock) with a certain payoff (such as a zero coupon bond) results in a rank at least as high as that of the risky asset in isolation.

Throughout this proposition we restrict attention to finite-valued ranking 
metrics, i.e.\ maps $r\colon L^\infty\to[0,+\infty)$; the extension to the case 
$r\colon L^\infty\to[0,+\infty]$ is discussed in Remark~\ref{rem:extended-codomain}.

\begin{proposition} \label{prop:cash-quasiconc-decision-therory}
Let $r: L^\infty \to [0,+\infty)$ be a monotone increasing, $L^\infty$-continuous 
map satisfying the normalization condition
\begin{equation}\label{eq:norm-cond}
\lim_{m\to-\infty} r(X+m) = 0 \qquad \text{for all } X\in L^\infty.
\end{equation}
Then $r$ satisfies cash-quasiconcavity if and only if
\[
X\simeq k \implies \lambda X+(1-\lambda)k\succeq X
\qquad \text{for all } k\in\mathbb{R},\ \lambda\in[0,1].
\]
\end{proposition}

\begin{proof}
$(\Rightarrow)$ Suppose $r$ is cash-quasiconcave and $X\simeq k$, i.e. $r(X)=r(k)$. 
Then by cash-quasiconcavity, for every $\lambda\in[0,1]$,
\[
r(\lambda X+(1-\lambda)k)\geq \min\{r(X),r(k)\}=r(X),
\]
that is $\lambda X+(1-\lambda)k\succeq X$.

$(\Leftarrow)$ Fix $X\in L^\infty$, $k\in\mathbb{R}$ and $\lambda\in[0,1]$. We must show 
$r(\lambda X+(1-\lambda)k)\geq \min\{r(X),r(k)\}$.
{Case $r(k)=0$.} In this case cash-quasiconcavity is trivial since $\min\{r(X),r(k)\}=0$.

  {Case $r(X)\geq r(k) >0$.} Consider $g(m):=r(X+m)$ for $m\leq 0$. By monotonicity 
$g$ is nondecreasing, by continuity it is continuous, and by \eqref{eq:norm-cond} 
$\lim_{m\to-\infty} g(m)=0\leq r(k)$. Since $g(0)=r(X)\geq r(k)$, the intermediate 
value theorem yields some $m\leq 0$ with $r(X+m)=r(k)$, i.e. $X+m\simeq k$. 
Applying the hypothesis to $X+m$ and $k$, together with monotonicity (note 
$\lambda X+(1-\lambda)k\geq \lambda(X+m)+(1-\lambda)k$ since $m\leq 0$), we obtain
\[
r(\lambda X+(1-\lambda)k)\geq r(\lambda(X+m)+(1-\lambda)k)\geq r(X+m)=r(k)
=\min\{r(X),r(k)\}.
\]

  {Case $r(X)<r(k)$.} Consider $h(m):=r(k+m)$ for $m\leq 0$. As above, $h$ is 
continuous, nondecreasing, with $\lim_{m\to-\infty} h(m)=0\leq r(X)$ and 
$h(0)=r(k)>r(X)$. By the intermediate value theorem there is $m\leq 0$ with 
$r(k+m)=r(X)$, i.e. $X\simeq k+m$. Applying the hypothesis to $X$ and $k+m$, 
together with monotonicity (note $\lambda X+(1-\lambda)k\geq \lambda X+(1-\lambda)(k+m)$ 
since $m\leq 0$), we obtain
\[
r(\lambda X+(1-\lambda)k)\geq r(\lambda X+(1-\lambda)(k+m))\geq r(X)
=\min\{r(X),r(k)\}.
\]

In both cases the cash-quasiconcavity inequality holds, which completes the proof.
\end{proof}

\begin{remark}\label{rem:extended-codomain}
The result above can be
extended to maps $r:L^\infty\to[0,+\infty]$ provided the continuity of 
$m\mapsto r(X+m)$ is understood with values in the extended half-line $[0,+\infty]$ 
equipped with its order topology, under which $[0,+\infty]$ is compact and connected. 
The only point requiring care is the existence of the shift. When the target level is 
finite, the argument above is unchanged. When the target level equals $+\infty$, say $r(k)=+\infty$ in the case 
$r(X)\geq r(k)$, then $r(X)=r(k)=+\infty$ forces $X\simeq k$ directly, so the 
hypothesis applies with no shift and the conclusion follows immediately.
\end{remark}

\begin{remark}
Mathematically, cash-quasiconcavity can be viewed as \textit{quasiconcavity at constants}. This means $r$ satisfies quasiconcavity at $k$ for all $k \in \mathbb{R}$, where:
$$\text{quasiconcavity at }k \in \mathbb{R}:
 r(\lambda X + (1-\lambda) k) \geq \min\left\{r(X), r(k)\right\}, \; X \in L^\infty, \ \lambda \in [0,1].$$ 
 Alternatively, this property can be referred to as quasi-star-concavity or concave-quasi-star-shapedness, being the concave analog of quasi-star-convexity, simply called quasi-star-shapedness by \cite{Han2024}. Here, we suggest the term cash-quasiconcavity to enhance its interpretation within a financial context.
\end{remark}

Investors may adopt different ranking metrics to evaluate the performance of their positions according to specific criteria, which can be formalized through the following properties.

\begin{definition}\label{def:RaM_other_prop}
Let $r: L^\infty  \to [0,+\infty]$ be a map.    
The map $r$ may satisfy:
\begin{itemize}
%\item[i)] weak expectation consistency: $k \in \mathbb{R}_{-}\implies\:r(k)=0$; 
\item[i)] concavity at $k\in \mathbb{R}$: %if $r$ satisfies concavity at $k, \; \; \forall k \in \mathbb{R}$, where: $\text{concavity at }k \in \mathbb{R}:$
$ r(\lambda X + (1-\lambda) k) \geq \lambda r(X) + (1-\lambda) r(k) , \; \; X \in L^\infty , \; \lambda \in [0,1];$
%\item[ii)] cash-concavity: %if $r$ satisfies concavity at $k, \; \; \forall k \in \mathbb{R}$, where: $\text{concavity at }k \in \mathbb{R}:$
%$ r(\lambda X + (1-\lambda) k) \geq \lambda r(X) + (1-\lambda) r(k) , \; \; X \in L^\infty , \; k \in \mathbb{R}, \; \lambda \in [0,1];$
\item[ii)] concavity:
$ r(\lambda X + (1-\lambda) Y) \geq \lambda r(X) + (1-\lambda) r(Y) , \; \; X, Y \in L^\infty , \; \lambda \in [0,1]$;
\item[iii)] quasiconcavity at $k\in \mathbb{R}$:
$ r(\lambda X + (1-\lambda) k) \geq \min\{ r(X), r(k)\} , \; \; X \in L^\infty , \; \lambda \in [0,1]$;
\item[iv)] quasiconcavity:
$ r(\lambda X + (1-\lambda) Y) \geq \min\left\{r(X), r(Y)\right\}, \; X,Y \in L^\infty, \ \lambda \in [0,1]$;
\item[v)] $[0,1]$-super-linearity: 
   $r(\lambda X) \geq \lambda r(X)$, $X \in L^\infty$, $\lambda \in [0,1]$.

%\item \deleted{positive homogeneity:  
 %  $r(\lambda X) = \lambda r(X)$, $X \in L^\infty$, $\lambda \in \mathbb{R}_+$;}
   \end{itemize} 
    If $r$ satisfies concavity at $k$ for all $k\in \R$, $r$ is called cash-concave. It could also happen that $r$ is concave at $k$ only on a subset of $L^\infty$. For this reason, we will also introduce the following notion: $r$ is concave  at $k>0$ on $A \subset L^\infty$ if for any $X\in A$ 
    $$
     r(\lambda X + (1-\lambda )k) \geq \lambda r(X) + (1-\lambda)r(k), \quad \text{for all }\lambda \in [0,1].
    $$
    If $r$ satisfies quasiconcavity at $k$ for all $k \in \mathbb{R}$, $r$ satisfies cash-quasiconcavity. 
%\textit{cash-superadditive} for any $X \in L^\infty$ and $k \in \mathbb{R}_+$ such that $X+k \in \mathrm{supp}(r)$, it holds that $r(X+k) \ge r(X) + k$. 
\end{definition}

\begin{remark}
We suggest the term \textit{cash-concavity} to enhance interpretability in the financial context. From a mathematical standpoint, cash-concavity, or \textit{concavity at constants}, corresponds to the notion of \textit{concave-star-shapedness}. This property is the concave counterpart to the star-shapedness (or convex-star-shapedness) defined by \cite{Han2024}, which generalizes the approach in \cite{Castagnoli2022}, where star-shapedness is defined as
convexity at $0$ with $r (0) = 0$. 
This specific term, concave-star-shapedness, is largely absent from the existing literature related to acceptability indices. Historically, star-shapedness originated in geometry as a property of sets or polygons \citep{hansen2020starshaped}, a context in which a concave counterpart lacks a meaningful definition. However, such a distinction becomes essential when the property characterizes functionals; here, the direction of the inequality is no longer purely geometric, but a fundamental characterization of the agent's behavioral properties.
\end{remark} 

 \begin{remark} 
 What we refer to as $[0,1]$-super-linearity is related to the concept of positive-subhomogeneity in the risk measures theory literature. Specifically, for $\lambda \geq 1$, the property can be rewritten as $r(\lambda X) \leq \lambda r(X)$, for all $X \in L^\infty$  (see \citealt{frittelli2002putting}). 
\end{remark} 
%While what we call 1-super-homogeneity is known in literature simply as positive-super-homogeneity and it can be equivalently rewritten as $r(\lambda X) \leq \lambda r(X)$, $X \in L^\infty$, $\lambda \in [0,1] $  (see \cite{Han2024}).]

\begin{remark}
Cash-concavity is a weaker requirement than the standard concavity, as the latter directly implies cash-concavity. Naturally, cash-concavity itself implies cash-quasiconcavity.
\end{remark}

In addition to the above properties ranking metric can be chosen to satisfy cash additivity or weaker related properties.

\begin{definition}\label{def:cash-sub-sup-add}
Let $r: L^\infty  \to [0,+\infty]$ be a map. 
The map $r$ may satisfy:
\begin{itemize}
\item[i)] sub-additivity at $k \in \mathbb{R}_+$: $r(X+k) \leq r(X)+k$, $X \in L^\infty$; 
\item[ii)] super-additivity at $k \in \mathbb{R}_+$: there exists $A \subseteq L^\infty$ such that $r(X+k) \geq r(X)+k$, $X \in A$ %{\rm supp}(r)$.   
\item[iii)] additivity at $k \in \mathbb{R}_+$: $r$ is both sub-additive at $k$ and super-additive at $k$;%\eli{controllare sa la usiamo ancora}
\end{itemize}
$r$ is \textit{cash-superadditive} (respectively, cash-subadditive) if it is super-additive (respectively, sub-additive) at $k$ for all $k \in \mathbb{R}_+$. $r$ is cash-additive if it is both cash-subadditive and cash-superadditive.  In particular, there exists $A\subseteq L^\infty$ such that 
$r(X+k) = r(X) + k$
for all $k\in \mathbb{R_+}$ and $X \in A$.%{\rm supp}(r)$.
\end{definition}

\begin{remark}\label{rem:supp-restriction}
Cash-superadditivity is required only on a subset of $L^\infty$, rather than on 
the whole space, because outside a suitable set the inequality may fail for 
trivial reasons. Indeed, suppose there exist $X\in L^\infty$ and 
$k\in\mathbb{R}_+\setminus\{0\}$ with $r(X+k)=0$. Since $r$ is nonnegative,
\[
r(X+k)=0<r(X)+k,
\]
so the cash-superadditivity inequality $r(X+k)\geq r(X)+k$ cannot hold. By 
contrast, if $X\in\operatorname{supp}(r)$, then $r(X)>0$ and, by monotonicity, 
$r(X+k)\geq r(X)>0$, so the situation $r(X+k)=0$ described above cannot occur, 
and the inequality is not ruled out a priori. This motivates requiring 
cash-superadditivity only on a subset excluding such points, such as 
$\operatorname{supp}(r)$.
\end{remark}

%If $r$ satisfies AD$_k$ (respectively SupAD$_k$) at any $k\in \R_+$, then $r$  satisfies the well known cash-additivity (respectively cash-superadditivity).

%Building on these definitions, we introduce the notions of cash-coherent (or simply coherent) and cash-concave ranking metrics. \blue{A ranking metric is called \textit{coherent} if it satisfies cash-additivity, positive homogeneity, and cash-super-additivity, in direct analogy with the classical definition of coherent risk measures \citep{Artzner1999}. [I do not think we use the expression coherent ranking metric, also I think this definition would not be correct]}.
%A ranking metric is called \textit{cash-concave} if it is cash-concave, paralleling the classical definition of concave risk measures \citep{FollmerSchied2002,frittelli2002putting}. 

%If a ranking metric $r$ satisfies cash-additivity, positive homogeneity and super-additivity, it is called a coherent ranking metric. 
%Furthermore, if $r$ is cash-superadditive and positive homogeneous, then $r$ is cash-concave, and we say that $r$ is a concave ranking metric. 

\color{black}
The following propositions  clarify the relationship among the properties discussed above and shed light on the lesser-known concave-star-shapedness.

% --- Convenzione sul dominio della cash-additività ---
% Da inserire vicino alla Definizione \ref{def:cash-sub-sup-add} o appena prima della proposizione.

Throughout, when $r$ is cash-additive on a set $A\subseteq L^\infty$, we assume 
$0\in A$, so that $A=\operatorname{supp}(r)\cup\{0\}$. This is consistent with 
the rationale of Remark~\ref{rem:supp-restriction}: the restriction to a subset 
of $L^\infty$ serves to exclude points $X$ at which $r(X+k)=0$ would force the 
additivity inequality to fail. The point $0$ is not of this kind, since under 
monotonicity $r(0+k)=r(k)\geq r(0)=0$ and, as shown in 
Lemma~\ref{lem:normalization}, in fact $r(k)=k>0$ for $k>0$. Including $0$ in 
$A$ therefore introduces no inconsistency.

% --- Lemma di normalizzazione ---

\begin{lemma}\label{lem:normalization}
Let $r\colon L^\infty\to[0,+\infty]$ satisfy $r(0)=0$ and be cash-additive on a 
set $A\subseteq L^\infty$ with $0\in A$. Then $r$ is normalized on the 
nonnegative constants, i.e.
\[
r(k)=k \qquad \text{for all } k\geq 0.
\]
\end{lemma}

\begin{proof}
For $k=0$ the claim is $r(0)=0$, which holds by assumption. For $k>0$, since 
$0\in A$, cash-additivity at $k$ applied to $X=0$ gives
\[
r(k)=r(0+k)=r(0)+k=0+k=k. \qedhere
\]
\end{proof}

% --- Proposizione (parte ii) riscritta ---
% Nota: ho rinominato le label duplicate in prop:rel-i e prop:rel-ii.

\begin{proposition} \label{prop: prop_relation_RM}
Let $r\colon L^\infty\to[0,+\infty]$.
\begin{itemize}
\item[i)] \label{prop:rel-i} If $r(0)=0$, then concavity at $0$ is equivalent 
to $[0,1]$-super-linearity.

\item[ii)] \label{prop:rel-ii} Suppose $r(0)=0$, $r$ is concave at $0$, and $r$ 
is cash-additive on $A=\operatorname{supp}(r)\cup\{0\}$. Then $r$ is concave at 
$k$ on $\operatorname{supp}(r)$ for every $k\in\mathbb{R}_+$, and therefore 
quasiconcave at $k$ on $\operatorname{supp}(r)$ for every $k\in\mathbb{R}_+$.
\end{itemize}
\end{proposition}

\begin{proof}
\textbf{i)} Using $[0,1]$-super-linearity, for any $\lambda\in[0,1]$,
\[
r(\lambda X+(1-\lambda)\cdot 0)=r(\lambda X)\geq \lambda r(X)
=\lambda r(X)+(1-\lambda)r(0),
\]
which is concavity at $0$. Conversely, if $r$ is concave at $0$, then
\[
r(\lambda X)=r(\lambda X+(1-\lambda)\cdot 0)\geq \lambda r(X)+(1-\lambda)r(0)
=\lambda r(X),
\]
which is $[0,1]$-super-linearity. In both directions we used $r(0)=0$.

\medskip
\textbf{ii)} By Lemma~\ref{lem:normalization}, $r(k)=k$ for every $k\geq 0$.

Fix $k>0$ and $X\in\operatorname{supp}(r)$, i.e. $r(X)>0$. The cases 
$\lambda\in\{0,1\}$ are trivial: for $\lambda=1$ the inequality reads 
$r(X)\geq r(X)$, and for $\lambda=0$ it reads $r(k)\geq r(k)$. Take 
$\lambda\in(0,1)$.

First, $\lambda X\in\operatorname{supp}(r)$: by concavity at $0$ and $r(0)=0$,
\[
r(\lambda X)\geq \lambda r(X)>0,
\]
since $\lambda>0$ and $r(X)>0$. Hence $\lambda X\in\operatorname{supp}(r)\subseteq A$, 
and cash-additivity applies to $\lambda X$ with increment $(1-\lambda)k>0$:
\[
r(\lambda X+(1-\lambda)k)=r(\lambda X)+(1-\lambda)k.
\]
Combining with $r(\lambda X)\geq \lambda r(X)$ and with $r(k)=k$ from 
Lemma~\ref{lem:normalization},
\[
r(\lambda X+(1-\lambda)k)\geq \lambda r(X)+(1-\lambda)k
=\lambda r(X)+(1-\lambda)r(k).
\]
This is concavity at $k$ on $\operatorname{supp}(r)$. Since concavity at $k$ 
implies quasiconcavity at $k$, the latter also holds on $\operatorname{supp}(r)$.

Finally, quasiconcavity at $k$ extends to the case $r(X)=0$ as well, since then
\[
r(\lambda X+(1-\lambda)k)\geq 0=\min\{r(X),r(k)\}. \qedhere
\]
\end{proof}

\subsection{Normalization properties}

% \begin{proposition}
%     If $r(0) \neq +\infty$ then positive super-homogeneity  (i.e. $r(\lambda X) \geq \lambda r(X)$, $X \in L^\infty$, $\lambda \geq 1)$) implies $r(0)=0$.\label{lem:pos-sup-norm}
%     \eli{mi sembra valga anche il viceversa: if $r(0)=0$ and $r$ concave at $0$ then $r$ super-homogeneous}
% \end{proposition} 
% \begin{proof}
%  If $r(0) \neq +\infty$ and $r$ satisfies positive super-homogeneity, then for any $\lambda \geq 1$ we have:
%     $$
%     r(\lambda \cdot 0)= r(0) \geq \lambda r(0)
%     \qquad \Rightarrow  r(0)(1-\lambda) \geq 0.
%     $$
%    In particular, for $\lambda >1$ we obtain $r(0)  \leq 0$. Recalling that $r: L^\infty \to [0,+\infty]$
%      we have
%     $r(0) = 0 $. 
% \end{proof}
Let us discuss the implications of different normalization choices.

One could require normalization at $k$ choosing $r(k)=k$, for all $k\in\R_+$, as 
originally suggested by \cite{ChernyMadan2009}. A sufficient condition for this 
normalization has already been established in Lemma~\ref{lem:normalization}: if 
$r(0)=0$ and $r$ is cash-additive on a set containing $0$, then $r(k)=k$ for all 
$k\geq 0$.
% --- Convenzione terminologica (se non già fissata altrove) ---
% Recall: r is cash-subadditive if r(X+h) <= r(X)+h for all X in L^infty, h>=0;
% r is cash-superadditive on a set B if r(X+h) >= r(X)+h for all X in B, h>=0;
% r is cash-additive (on B) if it is cash-subadditive on L^infty and 
% cash-superadditive on B, so that r(X+h)=r(X)+h for all X in B, h>=0.

\begin{proposition}\label{prop:qconc0+add-qconck}
Let $r\colon L^\infty\to[0,+\infty]$ be monotone increasing with $r(0)=0$. Then:
\begin{itemize}
    \item[i)] $r$ is quasiconcave at every level $k\leq 0$;
    \item[ii)] if $r$ is cash-subadditive on $L^\infty$ and cash-superadditive on 
    $\operatorname{supp}(r)$, then, for every $k>0$, $r$ is quasiconcave at $k$ on 
    the set
    \begin{equation}\label{eq:A_k}
        A_k:=\bigl\{X\in L^\infty:\ \lambda(X-k)\in\operatorname{supp}(r)
        \ \text{for all }\lambda\in(0,1]\bigr\}.
    \end{equation}
\end{itemize}
\end{proposition}

\begin{proof}
\textbf{i)} Let $k\leq 0$. By monotonicity, $k\leq 0$ gives $r(k)\leq r(0)=0$, and 
since $r\geq 0$ we conclude $r(k)=0$. Hence, for every $X\in L^\infty$ and every 
$\lambda\in[0,1]$,
\[
r(\lambda X+(1-\lambda)k)\geq 0=\min\{r(X),r(k)\},
\]
the inequality holding because $r$ is nonnegative and $r(k)=0$. Therefore $r$ is 
quasiconcave at every level $k\leq 0$.

\medskip
\textbf{ii)} Fix $k>0$ and $X\in A_k$. We show 
$r(\lambda X+(1-\lambda)k)\geq\min\{r(X),r(k)\}$ for all $\lambda\in[0,1]$.

  {Case $\lambda=0$.} The left-hand side is $r(k)$, so
\[
r(k)\geq\min\{r(X),r(k)\}
\]
holds trivially.

  {Case $\lambda\in(0,1]$.} First note the algebraic identity
\[
\lambda X+(1-\lambda)k=\lambda(X-k)+k.
\]
By definition of $A_k$, $\lambda(X-k)\in\operatorname{supp}(r)$. Applying 
cash-superadditivity on $\operatorname{supp}(r)$ and cash-subadditivity on 
$L^\infty$ at the point $\lambda(X-k)$ with increment $k>0$ yields the two 
inequalities
\[
r(\lambda(X-k))+k\leq r\bigl(\lambda(X-k)+k\bigr)\leq r(\lambda(X-k))+k,
\]
hence the equality
\[
r(\lambda X+(1-\lambda)k)=r\bigl(\lambda(X-k)+k\bigr)=r(\lambda(X-k))+k.
\]
Since $r\geq 0$, this gives
\[
r(\lambda X+(1-\lambda)k)=r(\lambda(X-k))+k\geq k.
\]
On the other hand, cash-subadditivity applied at $0$ with increment $k$ gives
\[
r(k)=r(0+k)\leq r(0)+k=k.
\]
Combining the two displays,
\[
r(\lambda X+(1-\lambda)k)\geq k\geq r(k)\geq\min\{r(X),r(k)\},
\]
where the last inequality is the definition of the minimum. Therefore $r$ is 
quasiconcave at the level $k$ on $A_k$.
\end{proof}

\begin{remark}
As an alternative to the normalization property $r(k)=k$, one may consider the 
case in which $r(0)=+\infty$. By monotonicity, this implies
\[
r(k)=+\infty \qquad \forall\, k\in\mathbb{R}_+.
\]
From a financial perspective, this assumption characterizes investors for whom 
risk-free allocations are always fully acceptable and constitute a preferred 
component of the portfolio.

Under this assumption, quasiconcavity at a non-negative cash level $k\in\mathbb{R}_+$,
\[
r(\lambda X+(1-\lambda)k)\geq\min\{r(X),r(k)\},
\]
reduces to
\begin{equation}\label{def:cash-quasiconc-r(k)infty}
r(\lambda X+(1-\lambda)k)\geq r(X),
\qquad
\forall\,\lambda\in[0,1],\; k\in\mathbb{R}_+,
\end{equation}
since $\min\{r(X),r(k)\}=r(X)$ whenever $r(k)=+\infty$.

Condition~\eqref{def:cash-quasiconc-r(k)infty} can be viewed as a generalized 
version of the star-shapedness property introduced in \cite{Righi2024}, which 
requires both $r(0)=+\infty$ and $r(\lambda X)\geq r(X)$ for all 
$\lambda\in[0,1]$. Indeed, under the assumption $r(0)=+\infty$, 
setting $k=0$ in~\eqref{def:cash-quasiconc-r(k)infty} gives
\[
r(\lambda X)\geq r(X) \qquad \forall\,\lambda\in[0,1],
\]
which, together with $r(0)=+\infty$, is precisely the star-shapedness of 
\cite{Righi2024}. Moreover, under $r(0)=+\infty$, star-shapedness is equivalent 
to quasiconcavity at $0$: indeed,
\[
r(\lambda X)=r(\lambda X+(1-\lambda)\cdot 0)\geq\min\{r(X),r(0)\}=r(X),
\]
and conversely.

The next proposition shows that, for monotone increasing maps, 
condition~\eqref{def:cash-quasiconc-r(k)infty} is sufficient for 
cash-quasiconcavity. The converse does not hold in general, as 
cash-quasiconcavity does not require $r(k)=+\infty$ for $k\in\mathbb{R}_+$.
\end{remark}

\begin{proposition}
Let $r\colon L^\infty\to[0,+\infty]$ be a monotone increasing map.
\begin{itemize}
    \item[i)] If $r(\lambda X+(1-\lambda)k)\geq r(X)$ for all $X\in L^\infty$, 
    $\lambda\in[0,1]$, $k\in\mathbb{R}$, then $r$ is cash-quasiconcave.

    \item[ii)] If $r$ is quasiconcave at $0$ and $r(0)=+\infty$, then 
    $r(\lambda X+(1-\lambda)k)\geq r(X)$ for all $\lambda\in[0,1]$ and 
    $k\in\mathbb{R}_+$.
\end{itemize}
\end{proposition}

\begin{proof}
\textbf{i)} For any $k\in\mathbb{R}$, using the hypothesis,
\[
r(\lambda X+(1-\lambda)k)\geq r(X)\geq\min\{r(X),r(k)\},
\]
which is cash-quasiconcavity.

\medskip
\textbf{ii)} Fix $k\in\mathbb{R}_+$ and $\lambda\in[0,1]$. Since $k\geq 0$, 
we have $\lambda X+(1-\lambda)k\geq\lambda X$ pointwise, so by monotonicity
\[
r(\lambda X+(1-\lambda)k)\geq r(\lambda X).
\]
By quasiconcavity at $0$ and $r(0)=+\infty$,
\[
r(\lambda X)=r(\lambda X+(1-\lambda)\cdot 0)\geq\min\{r(X),r(0)\}
=\min\{r(X),+\infty\}=r(X).
\]
Combining the two inequalities,
\[
r(\lambda X+(1-\lambda)k)\geq r(X). \qedhere
\]
\end{proof}

%  \iffalse
%  ***

% Let $r$ satisfy \quasistar. If $\lambda=0$, the claim holds trivially. 
% In fact we should prove that $r(0 \cdot X + k) =r(k) \geq r(X)$. 
% Taking into account quasiconcavity of $r$: 
% $r(0 \cdot X + k) \geq \min\left\{ r(X), r(k)\right\} = r(X)$ if $k>0$. Take $\lambda\in (0,1]$. Then we have 
% \[%\frac{1}{\lambda}
% r(\lambda X+(1-\lambda )k)\geq r(\lambda X)\geq \min \{r(X),r(0)\}=r(X).\] 

%  For the converse, if $k<0$, then the claim follows trivially since $\min\{r(X),r(k)\}=r(k)=0$.
%  In fact:
%  $$
%  r(\lambda X+(1-\lambda )k)  \geq \min\{r(X),r(k)\}=r(k)=0.
%  $$
%  If $k\geq 0$, then by monotonicity we have $r(k)\geq r(0)=\infty$. Hence, we obtain, \[r(\lambda X+(1-\lambda)k)\geq r(\lambda X)\geq r(X)=\min\{r(X),r(k)\}.\] 

%  \fi
%In particular, under the condition $r(0) = \infty$, we obtain the following characterizations of cash-quasiconcavity.

\begin{remark}
Note that part~ii) establishes condition~\eqref{def:cash-quasiconc-r(k)infty} 
only for $k\in\mathbb{R}_+$, whereas part~i) requires it for all $k\in\mathbb{R}$. 
Consequently, under the sole assumptions of quasiconcavity at $0$ and 
$r(0)=+\infty$, part~i) and part~ii) together imply cash-quasiconcavity only 
at $k\in\mathbb{R}_+$. Cash-quasiconcavity at $k<0$ requires additional 
assumptions on $r$, such as those in 
Proposition~\ref{prop:qconc0+add-qconck}.
\end{remark}
Another particular normalization choice is taking $r(0) = + \infty$ as in  \cite{RosazzaSgarra2013} and \cite{Righi2024}. %, meaning basically that a null position is always acceptable. However, this assumption would not incentivize building risky portoflios as doing nothing would be always preferable.
If this is assumed, the following remark further justifies relaxing the properties to cash-quasiconcavity alone to avoid the collapse to trivial cases.

\begin{remark}
Assume that $r(0)=+\infty$.

\begin{itemize}

\item[i)]
If $r$ is concave at $0$, then $r$ is trivial, namely
\[
r(X)=+\infty
\qquad\forall X\in L^\infty .
\]

Indeed, concavity at $0$ yields
\[
r(\lambda X+(1-\lambda)0)
   \ge \lambda r(X)+(1-\lambda)r(0)
\]
for every $X\in L^\infty$ and every $\lambda\in[0,1]$.
Since $r(0)=+\infty$, for every $\lambda\in[0,1)$ we obtain
\[
r(\lambda X)\ge
\lambda r(X)+(1-\lambda)(+\infty)
=+\infty,
\]
hence
\[
r(\lambda X)=+\infty .
\]
Given any $Y\in L^\infty$, choosing $\lambda\in(0,1)$ and
$X=Y/\lambda$, we obtain
\[
r(Y)=r(\lambda X)=+\infty.
\]
Therefore $r$ is identically equal to $+\infty$.

\item[ii)]
If $r$ is monotone increasing, then every almost surely nonnegative
position is assigned the value $+\infty$. More precisely, for every
$X\in L^\infty$ such that $X\ge 0$ $\mathbb P$-a.s.,
\[
r(X)\ge r(0)=+\infty,
\]
and therefore
\[
r(X)=+\infty.
\]

Hence every nonnegative payoff is automatically acceptable.

\end{itemize}
\end{remark}

% \begin{remark}
% We observe that in case $r(0)=+\infty$:
% \begin{itemize}
%     \item[i)] Concavity at $0$ implies that $r$ is trivial, i.e., $r(X)=+\infty$ for all $X\in L^\infty$. In fact, we would have  
%     $$
%       r(\lambda X + (1-\lambda)0) \geq \lambda r(X) + (1-\lambda) r(0) = +\infty,
%     $$
%     that is $r(\lambda X)=+\infty$ for every $X \in L^\infty$, for every $\lambda \in [0,1]$. 
    
%     % \viola{This is not correct! $r$ cannot be trivial if $r(0)=+\infty$ and $r$ is concave at $0$} In addition, if $r$ is not trivial, i.e. $r(X) \neq +\infty$ for some $X \in L^\infty$, and $r(0)=+\infty$ then super-homogeneity is ``not compatible'' with concavity at $0$, since 
%     % for any $X$ and $\lambda \geq 1$ such that
%     % $r(\lambda X) \in \mathbb{R}$ we would have 
%     % $$
%     % r(\lambda X) < r(0)= +\infty = r(\lambda X) +(1-\lambda ) r(0), $$
%     % which contradicts concavity at $0$.
% \item[ii)] If $r$ is quasiconcave at $0$, then for every $X \geq 0$, $\mathbb{P}$-a.s., we have $r(X)=+\infty$, meaning that any random payoff almost surely positive is always acceptable.
% \end{itemize}
% \end{remark}
From a financial perspective, the normalization $r(0) = +\infty$ characterizes passive investors and accumulators, such as those who leave capital idle or follow a static ``buy-and-hold'' strategy. In this framework, the status quo serves as an infinitely acceptable benchmark, implying that a position with a zero payoff is always inherently acceptable. Consequently, these agents consistently prioritize their current holdings over any active reallocation, as ``doing nothing'' is viewed as an optimal baseline.

\section{Characterization of ranking metrics} \label{sec:rank-metric-charact}
In this section, we characterize ranking metrics with respect to a family of acceptance sets and a family of risk measures. 
Acceptance sets $\mathcal{A}$ are widely used in the theory of risk measures and acceptability indices to characterize acceptable positions (in terms of risk or performance). When acceptability is indexed by a level $x \in \mathbb{R}_+$, this leads to a family of acceptance sets at different levels $\{\mathcal{A}_x\}_{x \in \mathbb{R}_+}$ where $\mathcal{A}_x$ consists of all positions acceptable at level $x$. A position is said to be acceptable at level $x$ if and only belongs to $\mathcal{A}_x$. Let us provide some preliminary definitions.
\begin{definition}
A subset $\mathcal{A}\subseteq L^\infty$ is:
\begin{itemize}
    \item[i)] monotone: if $X \leq Y, X\in\mathcal{A} \implies Y \in\mathcal{A}$; 
    \item [ii)] cash-convex: if $ X\in\mathcal{A} , \, k\in\mathcal{A}\cap \R \implies\lambda X+(1-\lambda)k\in\mathcal{A}$ for any $\lambda\in[0,1]$.
    %\item [iii)] monetary: if $\inf\{c\in\mathbb{R}\colon c\in\mathcal{A}\}=0$.\blue{Do we use this?}
\end{itemize}
A family of sets $\{\mathcal{A}_x\}_{x\in \R_+}$ is said  
\begin{itemize}
    \item[i)] increasing (decreasing, resp.) monotone: if $\mathcal{A}_x\subseteq\mathcal{A}_y$ when $x\leq y $ ($x\geq y$, resp); \item [ii)] cash-convex: if $\mathcal{A}_x$ is cash-convex for any $x\in \R_+$.
\end{itemize}
\end{definition}

Given a map $r: L^\infty  \to [0,+\infty]$, for instance a ranking metric, a natural choice of acceptance sets is the family of upper level sets
\[
\mathcal{A}^r_x := \{ X \in L^\infty : r(X) \ge x \}, \quad x\in \R_+,
\]
thus, $\mathcal{A}^r_x$ represents the set of positions whose performance is at least $x$. In the following result, we characterize \quasistar\ in terms of this family.

\begin{lemma}\label{lem:r-acceptance}
Consider $r\colon L^\infty\to[0,+\infty]$.
\begin{itemize}
    \item[i)] $r$ is cash-quasiconcave if and only if $\{\mathcal{A}^r_x\}_{x\in\R_+}$ 
    is cash-convex.
    \item[ii)] $r$ is (increasing) monotone if and only if $\mathcal{A}^r_x$ is 
    monotone for every $x>0$. Moreover, the family $\{\mathcal{A}^r_x\}_{x\in\R_+}$ 
    is decreasing monotone for every map $r$.
\end{itemize}
\end{lemma}

\begin{proof}
\textbf{i)} $(\Leftarrow)$ Let $X\in L^\infty$, $k\in\R$, and $\lambda\in[0,1]$. Set 
$\bar x:=\min\{r(X),r(k)\}$.

If $\bar x<+\infty$, then $X\in\mathcal{A}^r_{\bar x}$ and 
$k\in\mathcal{A}^r_{\bar x}\cap\R$, with $\bar x\in\R_+$; by cash-convexity of 
$\mathcal{A}^r_{\bar x}$,
\[
\lambda X+(1-\lambda)k\in\mathcal{A}^r_{\bar x},
\quad\text{i.e.}\quad
r(\lambda X+(1-\lambda)k)\geq\bar x=\min\{r(X),r(k)\}.
\]
If $\bar x=+\infty$, then $r(X)=r(k)=+\infty$, so $X,k\in\mathcal{A}^r_x$ for every 
$x\in\R_+$; by cash-convexity of each $\mathcal{A}^r_x$, 
$\lambda X+(1-\lambda)k\in\mathcal{A}^r_x$ for every $x\in\R_+$, whence 
$r(\lambda X+(1-\lambda)k)=+\infty=\min\{r(X),r(k)\}$.

In both cases $r$ is cash-quasiconcave.

$(\Rightarrow)$ Assume $r$ is cash-quasiconcave and fix $x\in\R_+$. Let 
$X\in\mathcal{A}^r_x$ and $k\in\mathcal{A}^r_x\cap\R$, i.e. $r(X)\geq x$ and 
$r(k)\geq x$. For every $\lambda\in[0,1]$,
\[
r(\lambda X+(1-\lambda)k)\geq\min\{r(X),r(k)\}\geq x,
\]
so $\lambda X+(1-\lambda)k\in\mathcal{A}^r_x$. Therefore $\mathcal{A}^r_x$ is 
cash-convex for every $x\in\R_+$.

\medskip
\textbf{ii)} $(\Rightarrow)$ Suppose $r$ is increasing monotone and fix $x>0$. If 
$X\in\mathcal{A}^r_x$ and $Y\geq X$, then $r(Y)\geq r(X)\geq x$, so 
$Y\in\mathcal{A}^r_x$; thus $\mathcal{A}^r_x$ is monotone.

$(\Leftarrow)$ Suppose $\mathcal{A}^r_x$ is monotone for every $x>0$, and let 
$X\leq Y$. If $r(X)=0$, then $r(Y)\geq 0=r(X)$ trivially. If $r(X)>0$, then for 
every $x\in\R_+$ with $0<x\leq r(X)$ we have $X\in\mathcal{A}^r_x$, hence by 
monotonicity $Y\in\mathcal{A}^r_x$, i.e. $r(Y)\geq x$. Taking the supremum over 
all such $x$ gives $r(Y)\geq r(X)$ (with $r(Y)=+\infty$ when $r(X)=+\infty$). 
Therefore $r$ is increasing monotone.

Finally, the family $\{\mathcal{A}^r_x\}_{x\in\R_+}$ is decreasing monotone for 
every map $r$: if $x\geq y$ and $X\in\mathcal{A}^r_x$, then $r(X)\geq x\geq y$, so 
$X\in\mathcal{A}^r_y$, i.e. $\mathcal{A}^r_x\subseteq\mathcal{A}^r_y$.
\end{proof}

% \textcolor{blue}{Shall we put a remark or proposition on the relation between properties of $r$ and of acceptance sets? }

% \textbf{MR: It is interesting to put the relationships in the blue remark below into the Lemma above, even if we do not prove them.}

% \blue{\begin{remark}TO DO MAYBE
% We recall here some well-known results. Take $r: L^\infty\subseteq L^{\infty} \to [0,+\infty]$    \begin{itemize}
%         \item $r$ is cash-concave if and only if $\{\mathcal{A}^r_x\}_{x \in \R_+}$ is xxx;
%         \item $r$ is cash-additive if and only if $\{\mathcal{A}^r_x\}_{x \in \R_+}$ is a xxx    
%     \end{itemize}
% \end{remark}}

%\textbf{MR: I believe these properties on risk measures must to be put into a Definition environment.}

\begin{definition} 
 A map $\rho:  L^{\infty} \to \R \cup \left\{ \pm\infty\right\}$ is a risk measure if it satisfies:
\begin{itemize}
    \item[i)] (decreasing) monotonicity: $X \leq Y, \quad X,Y \in L^\infty \implies \rho(X) \geq \rho(Y)$, 
    \item[ii)] cash-quasiconvexity: $ \rho(\lambda X + (1-\lambda) k) \leq \max \{ \rho(X),  \rho(k) \}, \; X \in L^\infty , k \in \R$, $\lambda \in [0,1]$.
\end{itemize}   
\end{definition}

Monotonicity is the most universally accepted property of a risk measure: if a position is expected to incur greater losses, its risk must be higher.  Throughout, when not specified, monotonicity is understood as decreasing in the case of risk measures and increasing in the case of ranking metrics. Cash-quasiconvexity naturally generalizes classical quasiconvexity by requiring that the risk of a mixture between a risky position and a deterministic payoff never exceeds the maximum risk of the two components evaluated in isolation. This property offers a minimal and intuitive criterion for risk measure selection while extending the \cite{DrapeauKupper2013} framework. 

\begin{definition}
 A map $\rho: L^{\infty} \to \mathbb{R}\cup \left\{ \pm\infty\right\}$ can also satisfy:
\begin{itemize}
     \item[i)]convexity at $k\in \mathbb{R}$: $\rho(\lambda X+(1-\lambda)k)\leq \lambda \rho(X)+(1-\lambda)\rho(k)$, $X \in L^\infty,\lambda\in[0,1]$;
         \item[ii)]convexity: $\rho(\lambda X+(1-\lambda)Y)\leq \lambda \rho(X)+(1-\lambda)\rho(Y)$, $X,Y \in L^\infty,\lambda\in[0,1]$;
     \item[iii)] quasiconvexity at $k\in \mathbb{R}$: $\rho(\lambda X+(1-\lambda)k)\leq \max\left\{ \rho(X),\rho(k)\right\}$, $X \in L^\infty,\lambda\in[0,1]$;
     \item[iv)] quasiconvexity: $\rho(\lambda X+(1-\lambda)Y)\leq \max\left\{ \rho(X),\rho(Y)\right\}$, $X,Y \in L^\infty,\lambda\in[0,1]$;
     \item[v)] $[0,1]$-sublinearity:  
   $\rho(\lambda X) \leq \lambda \rho(X)$, $X \in L^\infty$, $\lambda \in [0,1]$; 
    \item[vi)] additivity at $k\in \mathbb{R}_+$: $\rho(X-k) = \rho(X)+k$, $X \in L^\infty$; equivalently, $\rho(X+k) = \rho(X)-k$, $X \in L^\infty$;
    \item[vii)] sub-additivity at $k\in \mathbb{R}_+$: $\rho(X-k) \leq \rho(X)+k$, $X \in L^\infty$; equivalently, $\rho(X+k) \geq \rho(X)-k$, $X \in L^\infty$;
    \item[viii)] super-additivity at $k\in \mathbb{R}_+$: $\rho(X-k) \geq \rho(X)+k$, $X \in L^\infty$; equivalently, $\rho(X+k) \leq \rho(X)-k$, $X \in L^\infty$;
    %\blue{$\rho(X+k)+k$ is increasing in $k$}.
    
\end{itemize}
 If $\rho$ satisfies convexity at $k$ for all $k\in \R$, $\rho$ is called cash-convex. If $\rho$ satisfies quasiconvexity at $k$ for all $k \in \mathbb{R}$, $\rho$ satisfies cash-quasiconvexity. If $\rho$ satisfies additivity at $k$ (respectively sub/super-additivity at $k$) for all $k\in \R_+$, then $\rho$ satisfies the well known cash-additivity (respectively cash-sub/super-additivity). 

\end{definition}

\begin{remark}
Notice that convexity at $0$ 
corresponds to the star-shapedness property in \cite{Han2024} and, in case $\rho$ is normalized (i.e., $\rho(0)=0)$ it also corresponds to star-shapedness in \cite{Castagnoli2022}; cash-quasiconvexity corresponds to the quasi-star-shaped property in \cite{Han2024}.
\end{remark}

  We recall that a risk measure is called \textit{coherent} \citep{Artzner1999}  if it is monotone, subadditive (i.e. $\rho(X+Y) \leq \rho(X)+\rho(Y), $ for all $ X, Y \in L^\infty$), cash-additive, and normalized at zero. It is called \textit{convex} if it is monotone, cash-additive and convex \citep{FollmerSchied2002, frittelli2002R}.

As in the case of ranking metrics, the following proposition establishes the relationships among these properties.

%\textbf{MR: the proposition below is more a Lemma. 

\begin{proposition}
Let $\rho:  L^{\infty} \to \R\cup \left\{ \pm \infty\right\}$.
\begin{itemize}
%\item[i)] \ If $\rho$ is $[0,1]$-sublinear, then $\rho(0)=0$.
\item[i)] \ If $\rho$ is $[0,1]$-sublinear and $\rho(0)=0$, then $\rho$ is convex (therefore quasiconvex) at 0.    

\item[ii)] \ If $\rho$ cash-additive and convex at 0. then $\rho$ is cash-convex, therefore cash-quasiconvex.

\end{itemize}   
\end{proposition}

\begin{proof}
i) The reasoning is analogous to that in Proposition~\ref{prop: prop_relation_RM}.i).  
 ii) The proof parallels that of Proposition~\ref{prop: prop_relation_RM}.ii). 
 %with straightforward adaptations for any $k \in \mathbb{R}$.    
\end{proof}

The following result provides a characterization of a ranking metric in terms of its acceptance sets and an associated family of risk measures. Let $\{\rho_x\}_{x\in \mathbb{R}}$ be a family of risk measures. We say that this family is increasing if $\rho_x\geq\rho_y$ when $x\geq y$.

\begin{theorem}\label{thm:main}
The following statements are equivalent:
\begin{enumerate}
    \item $r$ is a ranking metric.
    \item There is a decreasing (in $x$) family $\{\mathcal{A}_x\}_{x\in \mathbb{R}_+}$ 
    of monotone and cash-convex acceptance sets such that
    \begin{equation}\label{eq:rhoTh1.3-Ax}
        r(X) = \sup\{ x > 0 \colon X \in \mathcal{A}_x\}.
    \end{equation}
    \item There is an increasing family of risk measures $\{\rho_x\}_{x \geq 0}$ such that
    \begin{equation}\label{eq:rhoTh1.3}
        r(X) = \sup \{x > 0 \colon \rho_x(X) \leq 0\}.
    \end{equation}
\end{enumerate}
\end{theorem}

\begin{proof}
We prove the equivalences by showing $1.\Rightarrow 2.$, $2.\Rightarrow 1.$, 
$1.\Rightarrow 3.$, and $3.\Rightarrow 1.$.

\medskip

\textbf{$1.\Rightarrow 2.$} Suppose $r$ is a ranking metric. Define the upper level sets
\[
    \mathcal{A}_x := \{X \in L^\infty \colon r(X) \geq x\}, \qquad x \in \mathbb{R}_+.
\]
We verify that the family $\{\mathcal{A}_x\}_{x \in \mathbb{R}_+}$ satisfies all the
required properties.

  {Monotonicity of each $\mathcal{A}_x$.} If $X \in \mathcal{A}_x$ and $Y \geq X$,
then $r(Y) \geq r(X) \geq x$ by monotonicity of $r$, so $Y \in \mathcal{A}_x$.

  {Cash-convexity of each $\mathcal{A}_x$.} If $X \in \mathcal{A}_x$ and 
$k \in \mathcal{A}_x \cap \mathbb{R}$, i.e.\ $r(X) \geq x$ and $r(k) \geq x$, then
by cash-quasiconcavity of $r$,
\[
    r\bigl(\lambda X + (1-\lambda)k\bigr) \geq \min\{r(X), r(k)\} \geq x,
\]
so $\lambda X + (1-\lambda)k \in \mathcal{A}_x$ for every $\lambda \in [0,1]$.

  {Decreasing monotonicity of the family.} If $x \leq y$ and $X \in \mathcal{A}_y$,
then $r(X) \geq y \geq x$, so $X \in \mathcal{A}_x$. Hence $\mathcal{A}_y \subseteq
\mathcal{A}_x$.

  {Representation.} By definition of $\mathcal{A}_x$,
\[
    \sup\{x > 0 \colon X \in \mathcal{A}_x\} 
    = \sup\{x > 0 \colon r(X) \geq x\} 
    = r(X),
\]
where the last equality follows from the convention $\sup \emptyset = 0$.

\medskip

\textbf{$2.\Rightarrow 1.$} Suppose $\{\mathcal{A}_x\}_{x \in \mathbb{R}_+}$ is a
decreasing family of monotone and cash-convex sets and $r$ is defined by
\eqref{eq:rhoTh1.3-Ax}.

  {Monotonicity of $r$.} If $X \leq Y$ and $X \in \mathcal{A}_x$, then
$Y \in \mathcal{A}_x$ by monotonicity of $\mathcal{A}_x$. Hence
$\{x > 0 \colon X \in \mathcal{A}_x\} \subseteq \{x > 0 \colon Y \in \mathcal{A}_x\}$,
which gives $r(X) \leq r(Y)$.

  {Cash-quasiconcavity of $r$.} Fix $X \in L^\infty$, $k \in \mathbb{R}$,
$\lambda \in [0,1]$, and set $\bar{x} := \min\{r(X), r(k)\}$. For every
$x < \bar{x}$ we have $r(X) \geq x$ and $r(k) \geq x$, i.e.\
$X \in \mathcal{A}_x$ and $k \in \mathcal{A}_x \cap \mathbb{R}$. By
cash-convexity of $\mathcal{A}_x$,
\[
    \lambda X + (1-\lambda)k \in \mathcal{A}_x.
\]
Since this holds for all $x < \bar{x}$, taking the supremum gives
$r(\lambda X + (1-\lambda)k) \geq \bar{x} = \min\{r(X), r(k)\}$.

\medskip

\textbf{$1.\Rightarrow 3.$} Suppose $r$ is a ranking metric. Define
\[
    \rho_x(X) := x - r(X), \qquad x \geq 0.
\]

  {Each $\rho_x$ is a risk measure.} Monotonicity: if $X \leq Y$, then
$r(X) \leq r(Y)$, so $\rho_x(X) = x - r(X) \geq x - r(Y) = \rho_x(Y)$, i.e.\
$\rho_x$ is decreasing monotone. Cash-quasiconvexity: for every $X \in L^\infty$,
$k \in \mathbb{R}$, $\lambda \in [0,1]$,
\[
    \rho_x\bigl(\lambda X + (1-\lambda)k\bigr) 
    = x - r\bigl(\lambda X + (1-\lambda)k\bigr)
    \leq x - \min\{r(X), r(k)\}
    = \max\{x - r(X),\, x - r(k)\}
    = \max\{\rho_x(X), \rho_x(k)\},
\]
where we used cash-quasiconcavity of $r$ in the inequality
and the identity
$
-\min\{a,b\} = \max\{-a,-b\}$ in the last step.

  {The family $\{\rho_x\}_{x \geq 0}$ is increasing.} For $x \leq y$,
$\rho_y(X) - \rho_x(X) = y - x \geq 0$, so $\rho_y \geq \rho_x$.

  {Representation.} We have $\rho_x(X) \leq 0$ if and only if $x \leq r(X)$.
Hence
\[
    \sup\{x > 0 \colon \rho_x(X) \leq 0\} 
    = \sup\{x > 0 \colon x \leq r(X)\} 
    = r(X).
\]

\medskip

\textbf{$3.\Rightarrow 1.$} Suppose $\{\rho_x\}_{x \geq 0}$ is an increasing family
of risk measures and $r$ is defined by \eqref{eq:rhoTh1.3}.

  {Monotonicity of $r$.} If $X \leq Y$, then by decreasing monotonicity of each
$\rho_x$, we have $\rho_x(X) \geq \rho_x(Y)$. Therefore $\rho_x(X) \leq 0$ implies
$\rho_x(Y) \leq \rho_x(X) \leq 0$, so
\[
    \{x > 0 \colon \rho_x(X) \leq 0\} \subseteq \{x > 0 \colon \rho_x(Y) \leq 0\},
\]
which gives $r(X) \leq r(Y)$.

  {Cash-quasiconcavity of $r$.} Fix $X \in L^\infty$, $k \in \mathbb{R}$,
$\lambda \in [0,1]$, and set $\bar{x} := \min\{r(X), r(k)\}$. For every $x < \bar{x}$,
we have $\rho_x(X) \leq 0$ and $\rho_x(k) \leq 0$, i.e.\
$\max\{\rho_x(X), \rho_x(k)\} \leq 0$. By cash-quasiconvexity of $\rho_x$,
\[
    \rho_x\bigl(\lambda X + (1-\lambda)k\bigr) 
    \leq \max\{\rho_x(X), \rho_x(k)\} \leq 0.
\]
Since this holds for all $x < \bar{x}$, taking the supremum gives
$r(\lambda X + (1-\lambda)k) \geq \bar{x} = \min\{r(X), r(k)\}$.
\end{proof}

% \begin{definition}
%     A ranking metric $r\colon L^\infty \to [0,+\infty]$ is 
%     quasi-normalized if $r(k)=k$ for every $k \in \mathrm{Im}(r) \cap \mathbb{R}$,
%     where $\mathrm{Im}(r) = \{r(X) : X \in L^\infty\}$.
% \end{definition}

\begin{proposition}\label{prop:4iff1}
Let $\mathcal{I}$ be a family of monotone increasing, cash-quasiconcave maps
$\alpha\colon L^\infty\to[0,+\infty]$ satisfying:
\begin{itemize}
    \item[\rm i)] $\alpha \leq r$ for every $\alpha \in \mathcal{I}$;
    \item[\rm ii)] $\alpha(k) = k$ for every $k \in [d_\alpha, u_\alpha] \cap \mathbb{R}$
    and every $\alpha \in \mathcal{I}$, where $d_\alpha = \inf_{X \in L^\infty}\alpha(X)$
    and $u_\alpha = \sup_{X \in L^\infty}\alpha(X)$, with the convention
    $\alpha(k) = k$ for every $k \geq d_\alpha$ when $u_\alpha = +\infty$;
    \item[\rm iii)] $r(X) = \sup\{\alpha(X) : \alpha \in \mathcal{I}\}$ is attained
    for every $X \in L^\infty$;
    \item[\rm iv)] $r(0) = 0$.
\end{itemize}
Then $r$ is a ranking metric.

% Conversely, if $r\colon L^\infty \to [0,+\infty]$ is a ranking metric satisfying
% $r(0)=0$ and $r(k)=k$ for every $k \in [0,u] \cap \mathbb{R}$, where
% $u = \sup_{X \in L^\infty} r(X)$, then the family
% \[
%     \mathcal{I} = \bigl\{\alpha \colon L^\infty \to [0,+\infty] :
%     \alpha \text{ monotone, cash-quasiconcave},\
%     \alpha(k) = k\ \forall\, k \in [d_\alpha, u_\alpha]\cap\mathbb{R},\
%     \alpha \leq r\bigr\}
% \]
% satisfies $r(X) = \max\{\alpha(X) : \alpha \in \mathcal{I}\}$ for every
% $X \in L^\infty$, with the maximum attained by $r$ itself.
\end{proposition}

\begin{proof}
We use the following notation:
\[
    d = \inf_{X \in L^\infty} r(X), \qquad
    u = \sup_{X \in L^\infty} r(X), \qquad
    d_\alpha = \inf_{X \in L^\infty} \alpha(X), \qquad
    u_\alpha = \sup_{X \in L^\infty} \alpha(X).
\]
Since $r(0) = 0$ and $r \geq 0$, monotonicity gives $r(k) = 0$ for all $k \leq 0$,
so $d = 0$. Since $\alpha \leq r$, we have $u_\alpha \leq u$ and $d_\alpha \leq d = 0$
for every $\alpha \in \mathcal{I}$; since $\alpha \geq 0$, we also have
$d_\alpha = 0$. We restrict to the subfamily
\[
    \mathcal{I}' = \bigl\{\alpha \in \mathcal{I} : d \leq u_\alpha\bigr\}
    = \bigl\{\alpha \in \mathcal{I} : u_\alpha > 0\bigr\},
\]
since any $\alpha$ with $u_\alpha = 0$ satisfies $\alpha \equiv 0 \leq r(X)$ for
all $X$, contributing nothing to the supremum.

\medskip

\textit{Monotonicity.} $r = \max_{\alpha \in \mathcal{I}'}\alpha$ is monotone
increasing as a pointwise maximum of monotone increasing maps.

\medskip

\textit{Cash-quasiconcavity.} Fix $X \in L^\infty$, $k \in \mathbb{R}$,
$\lambda \in [0,1]$  and let $\alpha^*_X \in \mathcal{I}'$
denote a maximiser, so that $\alpha^*_X(X) = r(X)$.

\textit{Case $k < 0$.}
Since $r(0) = 0$ and $r \geq 0$, monotonicity gives $r(k) = 0$. Hence
\[
    r\bigl(\lambda X + (1-\lambda)k\bigr) \geq 0
    = \min\{r(X), 0\} = \min\{r(X), r(k)\}. 
\]

\textit{Case $k \geq 0$.}
Since $d_\alpha = 0 \leq k$ for every $\alpha \in \mathcal{I}'$, hypothesis ii)
gives
\[
    \alpha(k) = k \wedge u_\alpha \qquad \forall\,\alpha \in \mathcal{I}'.
\]
Indeed, for $k \leq u_\alpha$: $\alpha(k) = k$ by ii). For $k > u_\alpha$:
$\alpha(k) = u_\alpha = k \wedge u_\alpha$ by monotonicity and the definition
of $u_\alpha$.

Using cash-quasiconcavity of each $\alpha \in \mathcal{I}'$:
\begin{equation*}
\begin{aligned}
    r\bigl(\lambda X + (1-\lambda)k\bigr)
    &= \max_{\alpha \in \mathcal{I}'} \alpha\bigl(\lambda X + (1-\lambda)k\bigr) \\
    &\geq \max_{\alpha \in \mathcal{I}'} \min\bigl\{\alpha(X),\, \alpha(k)\bigr\} \\
    &= \max_{\alpha \in \mathcal{I}'} \min\bigl\{\alpha(X),\, k \wedge u_\alpha\bigr\} \\
    &= \max_{\alpha \in \mathcal{I}'} \min\bigl\{\alpha(X),\, k\bigr\},
\end{aligned}
\end{equation*}
where the last equality uses $\alpha(X) \leq u_\alpha$, making $u_\alpha$
redundant in the minimum. Since $t \mapsto \min\{t, k\}$ is increasing and the
maximum is attained by $\alpha^*_X$:
\[
    \max_{\alpha \in \mathcal{I}'} \min\bigl\{\alpha(X),\, k\bigr\}
    = \min\bigl\{\alpha^*_X(X),\, k\bigr\}
    = \min\bigl\{r(X),\, k\bigr\}.
\]
It remains to show $\min\{r(X), k\} \geq \min\{r(X), r(k)\}$,
i.e.\ $r(k) \leq k$. Let $\alpha^*_k \in \mathcal{I}'$ be a maximiser on $k$,
so $\alpha^*_k(k) = r(k)$. By hypothesis ii):
\[
    r(k) = \alpha^*_k(k) = k \wedge u_{\alpha^*_k} \leq k. 
\]
Therefore
\[
    r\bigl(\lambda X + (1-\lambda)k\bigr)
    \geq \min\{r(X), k\}
    \geq \min\{r(X), r(k)\}.
\]
This completes the proof that $r$ is cash-quasiconcave, hence a ranking metric.

\medskip

% \textit{Converse.} Suppose $r$ is a ranking metric with $r(0)=0$ and
% $r(k)=k$ for every $k \in [0,u]\cap\mathbb{R}$. We verify $r \in \mathcal{I}$:
% $r$ is monotone and cash-quasiconcave by assumption; $r(k)=k$ for
% $k \in [0,u]\cap\mathbb{R} = [d_r, u_r]\cap\mathbb{R}$ by hypothesis; and
% $r \leq r$ trivially. Hence $r(X) \leq \max\{\alpha(X):\alpha \in \mathcal{I}\}
% \leq r(X)$ for every $X$, where the second inequality uses $\alpha \leq r$.
% Therefore the maximum equals $r(X)$ and is attained by $r$ itself.
\end{proof}

\begin{corollary}\label{cor:-rho}
Let $\rho\colon L^\infty \to \mathbb{R} \cup \{\pm\infty\}$ be a risk measure, and let $\{\rho_x\}_{x>0}$ be the family of maps defined by 
\begin{equation*}
    \rho_x(X) := \rho(X) + x, \quad \text{for any } x>0.
\end{equation*}
Then the representation of the ranking metric $r$ given in equation~\eqref{eq:rhoTh1.3} of Theorem~\ref{thm:main} reduces to
\begin{equation}
    r(X) = \max \{0, -\rho(X)\}.
\end{equation}
\end{corollary}

\begin{proof}
First, we must verify that the family $\{\rho_x\}_{x>0}$ satisfies the conditions of Theorem~\ref{thm:main}. Since $\rho$ is a risk measure, it is decreasing monotone and cash-quasiconvex. Adding a constant $x$ preserves both properties, meaning $\rho_x$ is a risk measure for each $x > 0$. Furthermore, for any $y \geq x > 0$, we have $\rho_y(X) - \rho_x(X) = y - x \geq 0$, ensuring the family $\{\rho_x\}_{x>0}$ is increasing.

Since the premises of Theorem~\ref{thm:main} hold, we can apply the representation formula:
\begin{equation*}
\begin{aligned}
r(X)& = \sup\left\{ x > 0 : \rho_x(X) \leq 0 \right\}\\
& = \sup\left\{ x > 0 : \rho(X) + x \leq 0 \right\}\\
& = \sup\left\{ x > 0 : x \leq -\rho(X) \right\}.
\end{aligned}
\end{equation*}
To evaluate this supremum, we consider two mutually exclusive cases. If $-\rho(X) \leq 0$, the set $\{ x > 0 : x \leq -\rho(X) \}$ is empty. Applying the standard convention $\sup \emptyset = 0$, we obtain $r(X) = 0$. If $-\rho(X) > 0$, the set is the interval $(0, -\rho(X)]$, and its supremum is exactly $-\rho(X)$.
Combining these two cases yields the compact form $r(X) = \max\{0, -\rho(X)\}$, completing the proof.
\end{proof}

\begin{corollary}\label{corr:general_RAROC}
Let $\rho\colon L^\infty \to \mathbb{R}$ be a cash-convex, decreasing monotone
risk measure, strictly decreasing on constants and such that $\rho(0)=0$. Let $\ell\colon L^\infty \to \mathbb{R}$ be a concave, increasing
map, strictly increasing on constants and such that $\ell(0)=0$.
%Assume that
% \begin{equation}\label{eq:compatibility}
%     \rho(X) \leq 0 \implies \ell(X) \geq 0 \qquad \forall\, X \in L^\infty.
%\end{equation}
Define \begin{equation}\label{eq:r-RAROC}
    r(X) =
    \begin{cases}
        \dfrac{\ell(X)}{\rho(X)},
        & \ell(X) > 0 \text{ and } \rho(X) > 0, \\[8pt]
        0,
        & \ell(X) \leq 0 \text{ and } \rho(X) > 0, \\[4pt]
        +\infty,
        & \rho(X) \leq 0.
    \end{cases}
\end{equation}
Then $r$ is a ranking metric and the corresponding family of risk measures is $\{\rho_x\}_{x>0}$ defined by
\begin{equation}\label{eq:rho_x-RAROC}
    \rho_x(X) =
    \begin{cases}
        x\rho(X) - \ell(X), & \rho(X) > 0, \\
        -\infty, & \rho(X) \leq 0, 
    \end{cases}
\end{equation}
with acceptance set representation given by
\begin{equation}\label{eq:A_x-RAROC}
    \mathcal{A}_x =
    \bigl\{X \in L^\infty \colon \rho(X) \leq 0\bigr\}
    \cup
    \left\{X \in L^\infty \colon \rho(X) > 0,\ \ell(X) > 0,\
    \frac{\ell(X)}{\rho(X)} \geq x\right\},
    \qquad x > 0.
\end{equation}
\end{corollary}

\begin{proof}
    Regarding monotonicity, for any $X \leq Y$ we have $\ell(X) \leq \ell(Y)$ (since $\ell$ is increasing) and $\rho(X) \geq \rho(Y)$. The result is trivially obtained if $\rho(X) \leq 0$, since in this case both $r(X)$ and $r(Y)$ are $+\infty$. For the case $\rho(X) > 0$, if $\ell(X) \leq 0$, the result is also trivial since $r(X) = 0 \leq r(Y)$, whatever the value of $r(Y)$ is. Otherwise, if $\rho(X) > 0$ and $\ell(X) > 0$, then $\ell(Y) \geq \ell(X) > 0$, while for $Y$ we could have $\rho(Y) > 0$ or $\rho(Y) \leq 0$. In case $\rho(Y) \leq 0$, we get $r(X) = \frac{\ell(X)}{\rho(X)} < r(Y) = +\infty$, while for $\rho(Y) > 0$, we have
    $$
    r(X) = \frac{\ell(X)}{\rho(X)} \leq \frac{\ell(Y)}{\rho(Y)} = r(Y).
    $$
    Therefore, monotonicity is verified in all cases.

    For cash-quasiconcavity, we divide the proof into three cases: $k=0$, $k>0$, $k<0$.
    When $k=0$, since $\rho(0)=0$, we have $r(0) = +\infty$. In order to verify cash-quasiconcavity, we need to prove that
    $$
    r(\lambda X + (1-\lambda) \cdot 0) = r(\lambda X) \geq \min\{r(X), r(0)\} = r(X).
    $$
    For $\lambda = 0$, the inequality reduces to $+\infty = r(0) \geq r(X)$, which holds trivially. For $\lambda \in (0, 1]$, we consider three sub-cases:
    \begin{itemize}
        \item[i)] If $\ell(\lambda X) > 0$ and $\rho(\lambda X) > 0$, we have for any $\lambda \in (0, 1]$:
        $$
        r(\lambda X) = \frac{\ell(\lambda X)}{\rho(\lambda X)} \geq \frac{\lambda \ell(X)}{\lambda \rho(X)} = r(X),
        $$
        where we used the fact that $\ell$ is concave with $\ell(0)=0$, which implies $\ell(\lambda X) \geq \lambda \ell(X) + (1-\lambda)\ell(0) = \lambda \ell(X)$, and $\rho$ is convex with $\rho(0)=0$, which implies $\rho(\lambda X) \leq \lambda \rho(X)$.
        \item[ii)] If $\rho(\lambda X) \leq 0$, we have immediately $r(\lambda X) = +\infty \geq r(X)$.
        \item[iii)] If $\ell(\lambda X) \leq 0$ and $\rho(\lambda X) > 0$, we have for any $\lambda \in (0,1]$:
        $$
        r(\lambda X) = 0 \qquad \text{and} \qquad r(X) = 0,
        $$
        since $\rho(\lambda X) \leq \lambda \rho(X)$ (by convexity at $0$) implies $\rho(X) > 0$, and $\lambda \ell(X) \leq \ell(\lambda X) \leq 0$ (by concavity at $0$) implies $\ell(X) \leq 0$.
    \end{itemize}

    Now we consider $k > 0$. The strict decreasing monotonicity of $\rho$ on constants implies $\rho(k) < \rho(0) = 0$, so that $r(k) = +\infty$. Hence, we need to verify:
    $$
    r(\lambda X + (1-\lambda)k) \geq \min\{r(X), r(k)\} = r(X).
    $$
    Since $k > 0$ and $1-\lambda \geq 0$, we have $\lambda X + (1-\lambda)k \geq \lambda X$. By the monotonicity of $r$ established above, it follows that $r(\lambda X + (1-\lambda)k) \geq r(\lambda X)$. Thus, the proof directly reduces to Cases a, b, and c discussed for $k=0$, satisfying $r(\lambda X) \geq r(X)$.

    Finally, for $k < 0$, the strict decreasing monotonicity of $\rho$ on constants implies $\rho(k) > \rho(0) = 0$. Since $\ell$ is strictly increasing on constants with $\ell(0)=0$, we have $\ell(k) < \ell(0) = 0$. Because $\rho(k) > 0$ and $\ell(k) < 0$, it follows by definition that $r(k) = 0$. Therefore, the cash-quasiconcavity inequality reduces to:
    $$
    r(\lambda X + (1-\lambda)k) \geq \min\{r(X), r(k)\} = 0,
    $$
    which holds trivially since the ranking metric $r$ maps to $[0, +\infty]$.

    By Theorem~\ref{thm:main}, $r_{\rho_x}(X) := \sup\{x > 0 \colon \rho_x(X) \leq 0\}$ is a ranking metric since $\{\rho_x\}_{x>0}$ is increasing in $x$. We show that $r_{\rho_x}(X) = r(X)$. Indeed, if $\rho(X) \leq 0$, then $\rho_x(X) = -\infty$ for all $x > 0$, yielding:
    $$
    r_{\rho_x}(X) = \sup \{x > 0 \colon \rho_x(X) \leq 0\} = \sup (0, +\infty) = +\infty = r(X).
    $$
    If $\rho(X) > 0$ and $\ell(X) \leq 0$, then due to the signs of $\rho(X) > 0$ and $-\ell(X) \geq 0$, we have $\rho_x(X) = x\rho(X) - \ell(X) > 0$ for all $x > 0$. Thus, $\{x > 0 \colon \rho_x(X) \leq 0\} = \emptyset$, so that:
    $$
    r_{\rho_x}(X) = \sup \emptyset = 0 = r(X).
    $$
    For $\ell(X) > 0$ and $\rho(X) > 0$, we have $\frac{\ell(X)}{\rho(X)} > 0$, and setting the inequality gives:
    $$
    \{x > 0 \colon \rho_x(X) \leq 0\} = \{x > 0 \colon x\rho(X) - \ell(X) \leq 0\} = \left\{x > 0 \colon \frac{\ell(X)}{\rho(X)} \geq x\right\},
    $$
    so that $r_{\rho_x}(X) = \sup \{x > 0 \colon \frac{\ell(X)}{\rho(X)} \geq x\} = \frac{\ell(X)}{\rho(X)} = r(X)$.
    The acceptance set representation $\mathcal{A}_x$ follows directly from Theorem~\ref{thm:main}.
\end{proof}

Characterizations of ranking metrics are important, as they provide tractable ways to construct such metrics. An immediate corollary of Theorem~\ref{thm:main} shows that any ranking metric can be obtained as the cash-quasiconcave envelope of a given monotone base map $f\colon L^\infty\to[0,+\infty]$.

{\begin{proposition}
Let $f : L^\infty \to [0,+\infty]$ be an increasing monotone map. Its cash-quasiconcave
envelope
\[
r_f(X)=\inf \{r(X): r\geq f,\ r\text{ is a ranking metric}\}
\]
is a ranking metric, and it can be represented as in equation \eqref{eq:rhoTh1.3-Ax} and equation \eqref{eq:rhoTh1.3} with, respectively,
\[
\begin{aligned}
A_x
&=
\bigcap\{B : B\supseteq \{X\in L^\infty : f(X)\geq x\},\ B\text{ is cash-convex and monotone}\},\\
\rho_x(X)
&=x-r_f(X).
\end{aligned}
\]
\end{proposition}
\begin{proof}
For each $x>0$, the set $A_x$ is cash-convex and monotone, since it is the intersection
of cash-convex and monotone sets. Moreover, the family $(A_x)_{x>0}$ is decreasing in
$x$. Define
\[
\widetilde r_f(X):=\sup\{x>0:X\in A_x\}.
\]
By equation \eqref{eq:rhoTh1.3-Ax}, $\widetilde r_f$ is a ranking metric. Also,
$\widetilde r_f\geq f$, because if $f(X)\geq x$, then
$X\in\{Y\in L^\infty:f(Y)\geq x\}\subseteq A_x$. Now let $r\geq f$ be any ranking metric. Then, for every $x>0$, the upper level set
$\{X\in L^\infty:r(X)\geq x\}$ is cash-convex and monotone and contains
$\{X\in L^\infty:f(X)\geq x\}$. Hence,
\[
A_x\subseteq \{X\in L^\infty:r(X)\geq x\}.
\]
It follows that $\widetilde r_f(X)\leq r(X)$ for every $X\in L^\infty$ and for every
ranking metric $r\geq f$. Therefore, $\widetilde r_f\leq r_f$. Since $\widetilde r_f$ is itself a ranking metric satisfying $\widetilde r_f\geq f$, the reverse inequality
$r_f\leq \widetilde r_f$ is immediate. Thus $r_f=\widetilde r_f$, and $r_f$ is a ranking metric. Finally, define $\rho_x(X):=x-r_f(X)$. Since $r_f$ is increasing monotone and
cash-quasiconcave, each $\rho_x$ is decreasing monotone and cash-quasiconvex. Moreover,
the family $(\rho_x)_{x>0}$ is increasing in $x$. Therefore,
\[
\sup\{x>0:\rho_x(X)\leq 0\}
=
\sup\{x>0:x\leq r_f(X)\}
=
r_f(X),
\]
which gives the representation in equation \eqref{eq:rhoTh1.3}.
\end{proof}
}

The following proposition identifies sufficient conditions for the family $\{\rho_x\}_x$ to satisfy various cash-additivity properties within the ranking metric. 
%We denote with $\mathcal{S}_{\rho_x}:=\{X \in \viola{L^\infty} \mid \rho_x(X) >0 \}$ the support of $\rho_x$. 

\begin{proposition}\label{Prop:rho_x}
Let $\{\rho_x\}_{x\geq 0}$ be an increasing family of risk measures and let $r$
be the ranking metric associated with $\{\rho_x\}_{x\geq 0}$ as defined in
equation~\eqref{eq:rhoTh1.3}.
\begin{itemize}
\item[i)] If $\rho_{x+c}(X+c)\geq\rho_x(X)$ for all $X\in L^\infty$, $x>0$,
$c\geq 0$ (equivalently, $\rho_x(X+c)\geq\rho_{x-c}(X)$ for all $X\in L^\infty$,
$x>c\geq 0$), then $r$ is cash-subadditive.
\item[ii)] If $\rho_x(X+c)=\rho_{x-c}(X)$ for all $X\in L^\infty$, $x>c\geq 0$,
then $r$ is cash-additive, with cash-superadditivity holding on
$A:=\mathrm{supp}(r)$.
\end{itemize}
\end{proposition}

\begin{proof}
\textit{i)} Fix $X\in L^\infty$ and $c\geq 0$. We estimate $r(X+c)$ by splitting
the supremum:
\[
r(X+c)
= \sup\{y>0:\rho_y(X+c)\leq 0\}
= \sup\bigl(\{y>c:\rho_y(X+c)\leq 0\}
            \cup\{0<y\leq c:\rho_y(X+c)\leq 0\}\bigr).
\]
For the first set: writing $y=x+c$ with $x>0$, if $\rho_{x+c}(X+c)\leq 0$
then by hypothesis $\rho_x(X)\leq\rho_{x+c}(X+c)\leq 0$, so
$x\in\{x>0:\rho_x(X)\leq 0\}$ and $y=x+c\leq r(X)+c$. Hence
\[
\sup\{y>c:\rho_y(X+c)\leq 0\}\leq r(X)+c.
\]
For the second set: $y\leq c$, so trivially $\sup\{0<y\leq c:\rho_y(X+c)\leq 0\}
\leq c\leq r(X)+c$.

Combining:
\[
r(X+c)\leq r(X)+c,
\]
which is cash-subadditivity.

\medskip

\textit{ii)} Let $X\in\mathrm{supp}(r)$, so $r(X)>0$, i.e.\
$\{x>0:\rho_x(X)\leq 0\}\neq\emptyset$. Pick any $\bar x>0$ with
$\rho_{\bar x}(X)\leq 0$. By the hypothesis with $x=\bar x+c$:
$\rho_{\bar x+c}(X+c)=\rho_{\bar x}(X)\leq 0$, so
$\{y>0:\rho_y(X+c)\leq 0\}\neq\emptyset$ and $r(X+c)>0$.

Using the hypothesis $\rho_x(X+c)=\rho_{x-c}(X)$ for $x>c\geq 0$:
\[
\begin{aligned}
r(X)+c
&= \sup\{x>0:\rho_x(X)\leq 0\}+c\\
&= \sup\{x+c>c:\rho_x(X)\leq 0\}\\
&= \sup\{y>c:\rho_{y-c}(X)\leq 0\}\\
&= \sup\{y>c:\rho_y(X+c)\leq 0\}\\
&\leq \sup\{y>0:\rho_y(X+c)\leq 0\}\\
&= r(X+c),
\end{aligned}
\]
which gives cash-superadditivity on $\mathrm{supp}(r)$.

Cash-subadditivity follows from i): the hypothesis $\rho_x(X+c)=\rho_{x-c}(X)$
implies, substituting $x\to x+c$, that $\rho_{x+c}(X+c)=\rho_x(X)$, i.e.\ the
condition of i) holds with equality, giving $r(X+c)\leq r(X)+c$.

Together, $r(X+c)=r(X)+c$ for all $X\in\mathrm{supp}(r)$ and $c\geq 0$, which
is cash-additivity on $\mathrm{supp}(r)$.
\end{proof}

\begin{corollary} \label{cor:cash-subad-transfer}
If $\rho_x(X) = \rho(X) + x$ and $\rho$ is cash-subadditive (resp. cash-additive), then the associated ranking metric $r$ is cash-subadditive in $L^\infty$ (resp. cash-superadditive in ${\rm supp}(r)$). %, i.e. $\forall X \in 
%\mathrm{Supp}(r) := \{\, X \in L^\infty \mid r(X) > 0 \,\}
%$. 
\end{corollary}
\begin{proof}
If $\rho$ is cash-subadditive, for any \(x > c \geq 0\)
\[\begin{aligned}
    \rho_x(X + c) = \rho(X + c) + x \geq \rho(X) - c + x = \rho(X) + (x - c) = \rho_{x - c}(X);
\end{aligned}\]
Hence, \(r\) cash-subadditivity follows from Proposition~\ref{Prop:rho_x}.i). If $\rho$ is cash-additive, Proposition~\ref{Prop:rho_x}.ii) yields cash-superadditivity on $\rm{supp}(r)$; together
with the first part, this gives cash-additivity in the sense of Definition \ref{def:cash-sub-sup-add}.
\end{proof}

\section{Ranking metrics: examples}
\label{sec:Rank_Metr_exam}
In this section, we present several performance metrics commonly used in finance and actuarial science, in particular those that also qualify as ranking metrics. We introduce two novel classes of performance measures that satisfy the ranking metric properties: the lambda-quantile-based measures \citep{Frittelli2014} %~\blue{CITA FrittelliEtAl}
 and the bibliometric-index-based performance measures~ \citep{frittelli2016scientific, PeriThesis}.%\blue{CITA ManciniFrittelliPeri, PeriThesis}).
 
 %Each ranking measure emphasizes different aspects of a portfolio performance. 

\subsection{Ranking metrics based on classical performance measures}

%\subsubsection{Ratio based performance metrics}

A very relevant type of ranking metric is a performance measure that is typically a ratio between some gain or return and a risk measure. A first example is given by the gain--loss ratio (GLR) introduced by
\citet{Bernardo2000}, which is the ratio between the expected return and the downside risk measure $E[X^-]$. The GLR is given by:
\begin{equation*}
	GLR(X)=\begin{cases}
		\dfrac{E[X]}{E[X^-]}&\:\text{if}\:E[X]>0, E[X^-]>0 .\\
        +\infty &\:\text{if}\:E[X]>0, E[X^-]=0\\
		0&\:\text{otherwise,}
\end{cases}\end{equation*}
% \begin{equation*}
% 	GLR(X)=\begin{cases}
% 		\dfrac{E[X]}{E[X^-]}&\:\text{if}\:E[X]>0\:\text{and}\:E[X^-]>0.\\
% 		0&\:\text{if}\:E[X]\leq 0\:\text{and}\:E[X^-]>0.\\
% 		\infty&\:\text{if}\:E[X^-]=0.
% \end{cases}\end{equation*}
where $X^-=(-X)\vee 0$ denotes the negative part of $X$. 
%This ratio offers a meaningful balance between expected gains and expected losses, providing a solid framework for evaluating the acceptability of financial positions.  %This makes it easier to interpret,  more practical to apply, and allows for straightforward comparisons across different funds. \\
By definition we have that $GLR(k)=+\infty$ for all $k\in \mathbb{R}_+ \setminus \{0\}$. 
The GLR is an acceptability index in the sense of \citet{ChernyMadan2009} - monotone, quasiconcave, scale invariant $GLR(\lambda X) = GLR(X), \; \lambda > 0$ and satisfying the Fatou property - therefore it is also a ranking metric.

Another widespread example is the Risk-Adjusted Return on Capital (RAROC), that measures the downside risk through a coherent risk measure and it is defined by:
\begin{equation*}
	RAROC(X)=\begin{cases}
		\dfrac{E[X]}{\rho(X)}&\text{if} \:E[X]>0\:\text{and}\:\rho(X)>0.\\
		0&\text{if} \:E[X]\leq0\:\text{and}\:\rho(X)>0.\\
		\infty&\text{if}\:\rho(X)\leq 0.
	\end{cases}
\end{equation*} 

RAROC is an acceptability index in the sense of \cite{ChernyMadan2009} when the risk measure \( \rho \) is coherent, and RAROC is quasi-concave when the risk measure is convex.

\begin{proposition}\label{prop:RAROC}
Let $\rho\colon L^\infty \to \mathbb{R}$ be a cash-convex, decreasing monotone
risk measure, strictly decreasing on constants and such that $\rho(0)=0$. % and let $\ell\colon L^\infty \to \mathbb{R}$ be a linear increasing
%map. 
Assume that
\begin{equation}\label{eq:compatibility}
    \rho(X) \leq 0 \implies E(X) \geq 0 \qquad \forall\, X \in L^\infty.
\end{equation}
Define the family $\{\rho_x\}_{x>0}$ by
\begin{equation}\label{eq:rho_x-RAROC}
    \rho_x(X) =
    \begin{cases}
        x\rho(X) - E(X), & \rho(X) > 0, \\
        -\infty, & \rho(X) \leq 0.
    \end{cases}
\end{equation}
Then $\{\rho_x\}_{x>0}$ is an increasing family of risk measures, and the ranking
metric $r$ given by the representations~\eqref{eq:rhoTh1.3}
and~\eqref{eq:rhoTh1.3-Ax} of Theorem~\ref{thm:main} takes the form
\begin{equation}\label{eq:r-RAROC}
    RAROC(X) =
    \begin{cases}
        \dfrac{E(X)}{\rho(X)},
        & E(X) > 0 \text{ and } \rho(X) > 0, \\[8pt]
        0,
        & E(X) \leq 0 \text{ and } \rho(X) > 0, \\[4pt]
        +\infty,
        & \rho(X) \leq 0,
    \end{cases}
\end{equation}
with acceptance set representation given by
\begin{equation}\label{eq:A_x-RAROC}
    \mathcal{A}_x =
    \bigl\{X \in L^\infty \colon \rho(X) \leq 0\bigr\}
    \cup
    \left\{X \in L^\infty \colon \rho(X) > 0,\ E(X) > 0,\
    \frac{E(X)}{\rho(X)} \geq x\right\},
    \qquad x > 0.
\end{equation}
% In particular, when $\ell(X) = E[X]$ and $\rho$ is a coherent risk measure,
% $r$ coincides with $\mathrm{RAROC}$.
\end{proposition}

\begin{proof}
	Regarding monotonicity, for any $X\leq Y$ we have $ E(X)\leq E(Y)$ and $\rho(X)\geq \rho(Y)$. The result is trivially obtained if $\rho(X)\leq 0$, since in this case both $RAROC(X)$ and $RAROC(Y)$ are $+\infty$. For the case $\rho(X)>0$, if $ E(X)\leq 0$, the result is also trivial, since $RAROC(X)=0\leq RAROC(Y)$, whatever the value of $RAROC(Y)$ is. Otherwise, i.e. $\rho(X)>0$ and $E[X]>0$, then $E(Y) \geq E(X) >0$ while for $Y$ we could have $\rho(Y)>0$ or $\rho(Y)\leq 0$. In case $\rho(Y)\leq 0$ we get 
    $$RAROC(X) =\frac{ E(X)}{\rho(X)} < RAROC(Y)=+\infty 
    $$
    while for $\rho(Y) >0$ we have
    \[RAROC(X)=\frac{ E(X)}{\rho(X)}\leq\frac{ E(Y)}{\rho(Y)}=RAROC(Y).\]
    Therefore  monotonicity is verified in all the possible cases. 
    
    For cash-quasiconcavity, we divide the proof in three cases: $k= 0$, $k>0$, $k<0$.
    When $k=0$ we have
    $\rho(0)=0$, thus $RAROC(0)=+\infty$. Hence
    % $$
    % RAROC(0)=\begin{cases}
    %     0, & \text{if }\rho(0)>0\\
    %     +\infty, & \text{if } \rho(0)\leq 0.
    % \end{cases}
    % $$
    in order to verify cash-quasiconcavity we need to prove that $$RAROC(\lambda X + (1-\lambda) \cdot 0)= RAROC(\lambda X) \geq  RAROC(X)=\min\left\{ RAROC(X),RAROC(0) \right\}.$$

    For $\lambda =0$ the above inequality reduces to: 
    $$
     +\infty= RAROC(0)\geq RAROC(X),
    $$
    which holds trivially.
    
    For $\lambda \in (0,1]$ 
    %we consider the following cases.    % Case 1. If $\rho(0)> 0$
    % the result is trivial since $RAROC(0)=0$ and 
    % $\min\left\{ RAROC(X),RAROC(0) \right\}=0$.
    % Caso 2. If $\rho(0)\leq 0$, then $RAROC(0)=+\infty$ and, consequently, $\min\left\{ RAROC(X),RAROC(0) \right\}=RAROC(X)$. Hence, in this situation 
    we need to prove that 
    $RAROC(\lambda X) \geq RAROC(X)$.

    Case a. If $E(\lambda X) >0,\rho(\lambda X)>0$ we have, for any $\lambda \in (0,1]$
    $$
    RAROC(\lambda X)=\frac{\lambda E(X)}{\rho(\lambda X)} \geq 
    \frac{\lambda E(X)}{\lambda\rho(X)} =RAROC(X),
    $$
    where the inequality comes from the convexity of $\rho$ at $0$, which means $\rho(\lambda X)\leq \lambda \rho(X)$ for all $\lambda \in (0,1]$.

    Case b. If $\rho(\lambda X)\leq 0$, we have immediately,
    $RAROC(\lambda X)=+\infty \geq  RAROC(X)$.

    Case c. If $E(\lambda X)\leq 0,\rho(\lambda X)>0$, we have, for any $\lambda \in (0,1]$,
    $$
    RAROC(\lambda X) = 0, \qquad RAROC(X)=0
    $$
    since convexity at $0$ implies $ 0 < \rho(\lambda X)\leq \lambda \rho (X)$ and $E(\lambda X) \leq 0$ implies $E(X)\leq 0$.

Now we consider $k>0$ which means $\rho(k)<0$ from strict decreasing monotonicity of $\rho$ on constants, so that 
$$
RAROC(k)=
+\infty.
$$
    Hence we need to verify
    $$
    RAROC(\lambda X +(1-\lambda )k) \geq \min\left\{ RAROC(X),RAROC(k) \right\}= RAROC(X).
    $$
    Since $RAROC(\lambda X +(1-\lambda )k)\geq RAROC(\lambda X)$ for every $k>0$ (by monotonicity of $RAROC$), cash-quasiconcavity holds if we have, once again 
    $RAROC(\lambda X )\geq RAROC( X)$.
    
    % Case 2. For $\rho(k)>0$ we need to verify that
    % $$
    % RAROC(\lambda X +(1-\lambda )k)\geq \min\left\{RAROC( X) , RAROC(k)\right\}.
    % $$
    % Since
    % it holds in general that 
    % $$
    % RAROC(\lambda X +(1-\lambda )k)\geq RAROC(\lambda X) \geq RAROC(X) 
    % $$
    % where the first inequality comes from monotonicity, while the second one comes from the previous argument, we obtain
    % $$
    % RAROC(\lambda X +(1-\lambda )k)\geq \min\left\{RAROC( X) , RAROC(k)\right\}.
    % $$
Finally, for $k<0$, we notice that strictly monotonicity of $\rho$ on constants %condition \eqref{eq:compatibility}
guarantees 
$\rho(k) >0$, %(also from strictly monotonicity of $\rho$ on constants), 
so that 
$RAROC(k)=0$ and therefore the inequality
$$
RAROC(\lambda X +(1-\lambda )k) \geq \min\left\{ RAROC(X),RAROC(k)\right\}
$$
reduces to 
$$
RAROC(\lambda X +(1-\lambda )k) \geq 0,
$$
which holds trivially. 

   %  For $k<0$ we have
   %  $$
   %  RAROC(k)= \begin{cases}
   %      +\infty, & \rho(k)
   %  \end{cases}
   %  $$
   %  is whenever $RAROC(X)$ or $RAROC(k)$ are equal to $0$, since in this case we would have
   %  \begin{equation*}
   %      \begin{aligned}
   %      &RAROC(\lambda X+ (1-\lambda)k) \geq 0=\min\left\{r(X), r(k)\right\}.
   %  \end{aligned}
   %  \end{equation*}
   %  For the case both $ RAROC(X)$ and $RAROC(k)$ are finite and such that $RAROC(X)\geq c$ and $RAROC(k)\geq c$ with $c>0$, which means that 
   %  $$
   %  RAROC(X)= \frac{E(X)}{\rho(X)} \qquad 
   %  RAROC(k)= \frac{k}{\rho(k)},
   %  $$
   %  we show that the upper level sets of $r$ are cash-convex. 
   %  %Thus, take $c\in\mathbb{R}_+$ with . 
   %  Fix any $\lambda\in[0,1]$. If $\rho(\lambda X+(1-\lambda)k) \leq 0$, the result is trivial. Otherwise, we have  \[E[\lambda X+(1-\lambda)k]\geq c(\lambda\rho(X)+(1-\lambda)\rho(k))\geq c\rho(\lambda X+(1-\lambda)k).\]

   %  Since $RAROC(X)$ is a ranking metric, by Theorem \ref{thm:main}.2 we can choose $\mathcal{A}_x := \{x>0\colon RAROC(X)\geq x\}$, increasing in $x$, such that $RAROC(X)=\sup\{x>0\colon x \in \mathcal{A}_x\}$. 
	
	  % Let $\{\rho_x\}_{x>0}$ be a family defined as \[	\rho_x(X):=
   %      \begin{cases}
   %      -E(X)+x \rho(X), & E(X)>0, \rho(X) >0\\
   %      0, & \text{otherwise}
   %          \end{cases}\] 
    By Theorem \ref{thm:main}, $r_{\rho_x}:=\sup\{x>0\colon\rho_x(X)\leq 0\}$ is a ranking metric since $\{\rho_x\}_{x>0}$ increasing in $x$.   
  We show that $r_{\rho_x}(X)=RAROC(X)$. Indeed, if $\rho(X) \leq  0$ we have $RAROC(X)=+\infty$; 
  moreover, $\rho_x(X)=-\infty$ for all $x>0$, and thus
  \begin{equation}
  \begin{aligned}
     r_{\rho_x}(X)& = \sup \left\{ x>0:\rho_x(X)\leq 0\right\}
     = \sup \mathbb{R}_+ \setminus \left\{ 0\right\}=+\infty
  \end{aligned}
  \end{equation}
  %where the last equality follows whatever it is the sign of
  %$\frac{\mathbb{E}[X]}{\rho(X) }$.
  %Hence 
  %$RAROC(X)=r_{\rho_x}(X)= +\infty$.
  % If $\rho(X)=0$, 
  % by definition $RAROC(X)=+\infty$; in addition, under the assumption that $E[X]\geq 0$ when $ \rho(X) =0$, 
  % we have that
  % %$$
  % %-E[X]\leq \rho(X) \leq 0 %\leq x(-%\rho(X))
  % %$$
  % %for all 
  % %$X\in L^\infty$ 
  % %and $x>0$.
  % %, since
  % %$-\rho(X) \geq 0$.
  % \begin{equation}
  % \begin{aligned}
  %    r_{\rho_x}(X)& = \sup \left\{ x>0: -\mathbb{E}[X]+x \rho(X) \leq 0\right\}\\
  %    & = \sup \left\{ x>0: 0 \leq \mathbb{E}[X]\right\}
  %    = \sup \left(\mathbb{R}_+ \setminus\left\{0\right\}\right) = +\infty,
  % \end{aligned}
  % \end{equation}
  % which means that when $\rho(X)= 0$,
  % $r_{\rho_x}(X)=RAROC(X)$.

  % When $\rho(X)>0$ we consider two cases. 
  If $\rho(X) >0$ and  $E(X) \leq 0$,
    due to the sign of $\rho(X)$ and $E(X)$, we have
   $$
   \left\{x>0: -E(X)+x\rho(X) \leq 0\right\} = \emptyset,
   $$
   so that
   $$
   r_{\rho_x}(X) = \sup \emptyset =0.$$ which corresponds to $RAROC(X)$ when $E(X)\leq 0$ and $\rho(X) >0$, 

    For $E(X)>0$ and $\rho(X) >0$ we have
   $\frac{E(X)}{\rho(X)}>0$, thus \begin{equation*}\begin{aligned}
      \left\{x>0: \rho_x(X)\leq 0 \right\} & = \left\{x>0: -E(X) + x\rho(X) \leq 0\right\} \\
       & =
   \left\{x>0:  -E(X) \leq -x\rho(X)\right\}\\
     & =  \left\{ x>0:\frac{E(X)}{\rho(X)} \geq x\right\}
   \end{aligned}\end{equation*}
    so that
    $$
    r_{\rho_x}(X) = \sup 
    \left\{ x>0:\frac{E(X)}{\rho(X)} \geq x\right\}=\frac{E(X)}{\rho(X)}= RAROC(X). 
    $$
   The representation of $\mathcal{A}_x$ comes from Theorem \ref{thm:main}.  
   %\viola{qui è con $r_{\rho_x}$, ma va almeno fatto un remark per $\mathcal{A}_x$. \blue{[TO DO]}}
 %  if $\rho(X)> 0$ for all $X$ and $E(X) \leq 0$, then $r_{\rho_x}=\sup\{ \emptyset \}=0$. If $-E(X)\leq \rho(X) \leq -x \rho(X)$ and $\rho(X)\leq 0$ then 
 %    $r_{\rho_x}=\sup\{\mathbb{R}_+ \}=\infty$, while, if $E(X) \geq 0$, then 
	% \[\mathcal{A}_{\rho_x}=\left\lbrace  X\in L^\infty\colon - E(X)+x\rho(X)\leq 0\right\rbrace=\left\lbrace X\in L^\infty\colon\dfrac{ E(X)}{\rho(X)}\geq x \right\rbrace = \mathcal{A}_x, \]
 %    which concludes the proof.

    %since we have the following for any $X\in L^\infty$: \[\begin{aligned}
	%	\frac{\partial\rho_x(X)}{\partial x}&=(1+x)^{-2}( E(X)-x\rho(X))+(1+x)^{-1}\rho(X)\\
	%	&\geq(1+x)^{-2}(-\rho(X)-x\rho(X))+(1+x)^{-1}\rho(X)\\
	%	&=(1+x)^{-1}\rho(X)-(1+x)^{-1}\rho(X)=0.
	%\end{aligned}\] 

%
    
    %If $\rho(X)\leq 0$, then $RAROC(X)=\infty$, which assures $RAROC(X)\geq x$ for any $x>0$. If $\rho(X)> 0$, \blue{Fix this part as we did} then \replaced{$- E(X)\leq -x\rho(X) \leq 0 \leq \rho(X)$}{ $ E(X)\geq-\rho(X) \geq 0\geq x\rho(X)$}. In this case we obtain that \[\rho_x(X)=-\frac{1}{1+x} E(X)+\frac{x}{1+x}\rho(X) \leq 0.\:\forall\:x>0.\] Thus, $RAROC(X)=\infty=\sup\{x>0\colon\rho_x(X)\leq 0\}$. Finally, we get that
	%\[\mathcal{A}_x=\mathcal{A}_{\rho_x}=\left\lbrace  X\in L^\infty\colon - E(X)+x\rho(X)\leq 0\right\rbrace=\left\lbrace X\in L^\infty\colon\dfrac{ E(X)}{\rho(X)}\geq x \right\rbrace.  \]
\end{proof}

 A related acceptability formulation for the GLR is the $\Omega$ ratio \citep{keating2002introduction}.
 %with target $\tau=0$. 
 Omega further decomposes the return distribution into positive and negative parts. Omega is defined as
\begin{equation*}
	\Omega(X)=\begin{cases}
        \frac{E[X^+]}{E[X^-]} &  \text{if }\:E[X^-]>0\\
        +\infty, & \text{if }\:E[X^-]=0.
    \end{cases}
\end{equation*}
% \begin{equation*}
% \Omega(X)=\begin{cases}
% 	\dfrac{E[X^+]}{E[X^-]}&\:\text{if}\:E[X^-]>0.\\
% 		0 &\:\text{if}\:E[X^-]=0.
% \end{cases}\end{equation*}
%Depending on the chosen normalization, for $E[X^-]=0$, we set  $\Omega(X)=0$ if we want $\Omega(0)=0$, and $\Omega(X)= +\infty$ if
%we want $\Omega(0)=\infty$. \red{The only normalization that make sense if we want quasiconcality is $\Omega(0)=0$.}
This formulation is not quasiconcave due to convexity of the numerator. However, Omega ratio is a star-shaped acceptability index as in \cite{Righi2024} (i.e. monotone, quasiconcave at 0 and $\Omega(0)=+\infty$),  which can be represented by a family of expectile-based risk measures $\rho_x=EVaR^{\frac{1}{1+x}}$, which is not convex for $x<1$, it is quasi-convex at $0$ for any $x$. We remind that an expectile-based risk measure $EVaR^{p}$ is the negative of a $p$-expectile of $X$, for $p\in[0,1]$, that is the unique solution $y$ of $pE[(X-y)^+]=(1-p)E[(X-y)^-]$ for $p\in[0,1]$. 
The following result establishes that Omega is a ranking metric since it satisfies \quasistar. 

\begin{proposition}
\begin{itemize}
    \item[i)] The Omega ratio is monotone and cash-quasiconcave, hence it is a ranking metric.
    \item[ii)] \citep{Righi2024} Given $\rho_x=EVaR^{\frac{1}{1+x}}$, family of expectile-based risk measures $EVaR^{p}$ with $p=\frac{1}{1+x}$, then $\Omega(X)$ can be represented as in equation \eqref{eq:rhoTh1.3-Ax} and \eqref{eq:rhoTh1.3}  of Theorem \ref{thm:main} with, respectively, $$\rho_x=EVaR^{\frac{1}{1+x}}$$
and
$$
\mathcal{A}_x=\left\lbrace X\in L^\infty\colon\frac{E[X^+]}{E[X^-]}\geq\dfrac{1-p}{p} \right\rbrace.
$$
\end{itemize}

\end{proposition}

\begin{proof}
To verify this, let $Y=aX+(1-a)k$. Using the representation
$E[X^+]=\int_0^{\infty}\!P(X>t)\,dt$ and $E[X^-]=\int_0^{\infty}\!P(-X>t)\,dt$,
we analyze two cases. If $k<0$, then $\Omega(k)=0$, while $\Omega(Y)\ge0$ by definition.
Thus $\Omega(Y)\ge\min\{\Omega(X),\Omega(k)\}$. If $k\ge0$ we distinguish two cases.
 If $E[Y^-]=0$ then by definition $\Omega(Y)=+\infty$, hence $\Omega(Y) \geq \min \{\Omega(X), \Omega(k)\}$.
 
If $E[Y^-]>0$ then 
\[
E[Y^+]
 =\int_0^{\infty}\!P(aX+c>t)\,dt
 =a\!\int_{-c/a}^{\infty}\!P(X>s)\,ds
 \ge aE[X^+]+c\,P(X>0).
\]
Similarly,
\[
E[Y^-]
 =\int_0^{\infty}\!P(-aX-c>t)\,dt
 =a\!\int_{c/a}^{\infty}\!P(-X>u)\,du
 \le aE[X^-].
\]
Hence, since $E[Y^-]>0$
\[
\Omega(Y)=\frac{E[Y^+]}{E[Y^-]}
   \ge \frac{aE[X^+]+c\,P(X>0)}{aE[X^-]}
   \ge \frac{E[X^+]}{E[X^-]}
   =\Omega(X).
\]
Since $\Omega(k)=\infty$ for $k\ge0$, the \quasistar\  inequality holds. We refer to \cite{Righi2024} for the proof of the representations.
\end{proof}

\begin{remark}
    Within this definition, the only normalization compatible with cash-quasiconcavity and monotonicity is $\Omega(0)=+\infty$. 
\end{remark}

\subsection{Ranking metrics generated by risk measures}
  
\subsubsection{Certainty equivalent / Expected loss based ranking metrics}
\label{sub_sub:Exp_lossRM}
A first example is the class of ranking metrics generated by expected loss-based risk measures. 

%A first and simple example for \quasistar \ risk measure is the family of expected losses as $E[f(-X)]$, where 
%$f\colon\R\to\R$ be increasing. 
%Such $\rho$ is a \quasistar \ risk measure.

\begin{definition}
 Let $\left\{f_x\right\}_{x>0}$ be family of increasing and convex, hence continuous, real functions, also increasing in $x$, i.e., $f_x \leq f_y$ for any $x\leq y$. Proposition 3.3 in \cite{Han2024} assures the expected loss $E[f_x(-X)]$ is a risk measure for any $x>0$. Consider a family of risk measures $\{\rho_x\}$ such that $$\rho_x(X)=E[f_x(-X)],$$ then by Theorem \ref{thm:main}.3,  we can generate a new family of ranking metrics as follows: 
 \begin{equation} \label{def:r_EL}
 r^f(X)=\sup\{x>0\colon E[f_x(-X)]\leq 0\}.
  \end{equation}
\end{definition}
  %X is a profit

Typical examples are, for instance, an insurance contract that pays $f_x(-X)$ for an insurable loss $-X$, or the put premium $E[(X-x)^-]$ for some strike price $x$ where $f_x(y)=(y+x)^+$.
%in this way $f_x(-X)= (x-X)^+= (-(X-x))^+ = (X-x)^-$.

\begin{example}
Let $\rho_x(X)=\rho(X)+x$ with $\rho(X)=E\left[f(-X)\right]$, which implies $f_x(y)=f(y)+x$. Then by Corollary \ref{cor:-rho} we obtain $$r^f(X)=\sup \left\{ 0, -E[f(-X)]\right\}.$$
\end{example}
The following lemma shows that
%how various properties of \(r^f\) depend on specific properties of the family \(\{f_x\}\).
under specific assumptions on \(\{f_x\}\), expected loss-based ranking metrics $r^f$ are cash-subadditive.
\begin{lemma}
If $f_{x+c}(y-c)\geq f_x(y)$ for all $x>0.c\geq 0$ we have that $r^f$ is cash-subadditive in $L^\infty$.

If $f_{x+c}(y-c)= f_x(y)$ for all $x>0, \ c\geq 0$ we have that $r^f$ is cash-additive in ${\rm supp}(r)$. 
%\begin{itemize}
% \item[i)]
% $r_{EL}(X)$ is cash-subadditive.
 %\item[ii)] If $f_x(X)$ is positive homogeneous in $X$, then $r_{EL}(X)$ is positive-homogeneous???
%\end{itemize}
\end{lemma}

\begin{proof}
%\begin{itemize}
%\item[i)] 
Assume that $f_{x+c}(y-c)\geq f_x(y)$ for all $x>0.c\geq 0$ holds. 
Then 
$$
{\rho}_x(X):= E[f_x(-X)]
$$
satisfies
$$
{\rho}_{x+c}(X+c)
= E[f_{x+c}(-X -c)] 
\geq E[f_x(-X)]
= {\rho}_x(X)
$$
Hence, cash-subadditivity of $r(X)$ follows from Proposition~\ref{Prop:rho_x}.(i). 

In case $f_{x+c}(y-c)= f_x(y)$ for all $x>0, \ c\geq 0$
cash-additivity follows similarly using Proposition~\ref{Prop:rho_x}.(ii). 
\end{proof}

Expected loss-based ranking metrics admit an equivalent utility-based representation and thus coincide with metrics induced by a family of generalized certainty equivalents. 

\begin{definition}
  Given an increasing family $\{u_x\}_{x>0}$ of strictly  increasing, concave utility functions satisfying $E[\,|u_x(X)|\,]<\infty$ for all $x>0$, we define the family of certainty equivalent associated with $u_x$ as:
\begin{equation}\label{def:cert-eq}
\tilde{\rho}_x(X)\;:=\; -\,u_x^{-1}\!\big(E[u_x(X)]\big).
\end{equation}
The induced expected utility-type ranking metric is then given by:
\[\begin{aligned}
r^{u}(X)\;:=\;\sup\{\,x>0:\ -\,u_x^{-1}\!\big(E[u_x(X)]\big)\le 0\,\}.
\end{aligned}
\]  
\end{definition}
We now show that this construction coincides with the expected utility ranking metric generated by the choice $f_x(y)=-u_x(-y)$ in~\eqref{def:r_EL}.

\begin{proposition}
Let $\{u_x\}_{x>0}$ be an increasing family of strictly increasing and concave utility functions,  such that $u_x(0)=0$ and $E[|u_x(X)|]<\infty$ for all $x>0$. If $f_x:\mathbb{R}\to\mathbb{R}$ is defined by $f_x(y)\;:=\;-u_x(-y)$, then
\[
r^{f}(X)\;=\;r^{u}(X).
\]

\end{proposition}

\begin{proof} By definition of $f_x$ and  $r^f$ follows that

\[\begin{aligned}
r^f(X) & = \sup\left\{ x >0: E[f_x(-X)] \leq 0\right\} \\
& = \sup\left\{ x >0: E[-u_x(X)]\leq 0\right\}\\
& = \sup\left\{ x >0:- u_x^{-1}(E[u_x(X)])\leq 0\right\} =r^u(X)
\end{aligned}\]
\end{proof}

% \[\begin{aligned}
% r^f_{EL}(X) & = \sup\left\{ x >0: \mathbb{E}[f_x(-X)] \leq 0\right\} \\
% & = \sup\left\{ x >0: \mathbb{E}[-u_x(-X)]\leq 0\right\}\\
% & = \sup\left\{ x >0: u^{-1}((\mathbb{E}[-u_x(-X)])\leq 0\right\}\\
% & = \sup\left\{ x >0: -u^{-1}((\mathbb{E}[u_x(-X)])\leq 0\right\}\\
% &= \sup\left\{ x >0: -u^{-1}((\mathbb{E}[u_x(-X)])\leq 0\right\}=r^u_{EL}(X)
% \end{aligned}\]

\begin{example}
If the family of expected-utility-based risk measures is defined by $x$-shifts    $\tilde{\rho}_x(X)=\tilde{\rho}(X)+x$ with $\tilde{\rho}(X)=-u^{-1}E[u(X)]$, then by Corollary \ref{cor:-rho} the ranking metric is given by:
\begin{equation}\label{def:sup-cert-eq}
\begin{aligned}
r^u(X)= \sup \left\{ 0, \,  u^{-1}(E[u(X)])\right\}.
\end{aligned}
\end{equation}
Notice that if $u(y)=-f(-y)$, then$$r^f(X)=u(r^u(X)).$$ 
\end{example}

\subsubsection{$\Lambda$-quantiles based ranking metrics}

As a novel example, we consider ranking metrics generated by $\Lambda$-quantiles, known originally as Lambda value at risk,  \citep{Frittelli2014, BelliniPeri2022}. This risk measure, \(\Lambda\text{VaR}\)(X), is defined for any \(X\) as the negative of the $\Lambda$-quantile $q_{\Lambda}(X)$:
\[
\Lambda\text{VaR}(X) := -q_{\Lambda}(X)= -\inf\left\{y \in \mathbb{R} : \mathbb{P}(X \leq y) > \Lambda(y)\right\},
\]
where \(\Lambda\colon \mathbb{R} \to [0,1]\) is a function that is not identically equal to zero. When \(\Lambda = \alpha\) is constant, we recover the traditional Value-at-Risk at the level $\alpha$ (\(\text{VaR}_\alpha\)).
The \(\Lambda\text{VaR}\) is not quasiconvex but it is a cash-quasiconvex risk measure and becomes cash-subadditive when \(\Lambda\) is decreasing (see~\cite{Han2024}).

\begin{definition}
Given a family of lambda value at risk \(\{\Lambda \mathrm{VaR}_x\}_{x>0}\) increasing in $x$, by Theorem \ref{thm:main}.3, we construct a novel class of ranking metrics by defining
\[
r_{\Lambda \mathrm{VaR}_x}(X) := \sup\left\{x > 0 : \Lambda \mathrm{VaR}_x(X) \leq 0\right\}.
\]
\end{definition}

A class of lambda quantiles ranking metrics can be induced by a family of lambda value at risk measures defined through a decreasing collection of lambda functions \(\{\Lambda^x\}\).  

\begin{example}
Let \(\{\Lambda^x\}\) be a decreasing family of decreasing functions \(\Lambda^x\colon \mathbb{R} \to [0,1]\) not identically equal to zero. We define
\[
\Lambda^x \mathrm{VaR}(X) := -q_{\Lambda^x}(X)= -\inf\left\{y \in \mathbb{R} : \mathbb{P}(X \leq y) > \Lambda^x(y)\right\}.
\]
The induced ranking metric is given by:
\[
r_{\Lambda^x \mathrm{VaR}}(X) := \sup\left\{x > 0 : \Lambda^x \mathrm{VaR}(X) \leq 0\right\}.
\]    
\end{example}

Another class of ranking metrics can be induced by the family of lambda value at risk measures generated by $x$-shifts.  

\begin{example}
Let be \(\{\Lambda \mathrm{VaR}_x\}_{x>0}\) the family of lambda value at risk generated by $x$-shifts \(\Lambda \mathrm{VaR}_x(X) = \Lambda \mathrm{VaR}(X) + x\), for any $x>0$. By Corollary \eqref{cor:-rho}, the induced ranking metric is given by
\begin{equation}\label{def:r_LambdaVaR}
r_{\Lambda\text{VaR}}(X) = \sup\left\{0\, -\Lambda \mathrm{VaR}(X)\right\}= \sup\left\{0, q_\Lambda(X)\right\}.
\end{equation} 
\end{example}
Notice that the $\Lambda$-quantile ranking metric is cash-quasiconcave but not quasiconcave, showing how meaningful performance measures can still emerge when standard quasiconcavity assumptions are relaxed.

In the case where the $\Lambda\mathrm{VaR}$ family is generated by cash shifts the following lemma holds.  

\begin{lemma}\label{L:lambdaVaR_cashsubadd}
   Let us assume $\Lambda_x \mathrm{VaR} (X)=\Lambda \mathrm{VaR}(X)+x$.  If \(\Lambda\) is decreasing in \(y\), then the induced ranking metric $r_{\Lambda VaR}$ is cash-subadditive.
\end{lemma}

\begin{proof}
By Corollary \ref{cor:cash-subad-transfer}, it suffices to show that the underlying risk measure $\Lambda\mathrm{VaR}$ is cash-subadditive. As established in \cite{Han2024}, $\Lambda\mathrm{VaR}$ satisfies cash-subadditivity provided that $\Lambda$ is a decreasing function. 
% To prove cash-subadditivity, we need to show that for any \(c > 0\),
% \[
% \Lambda\text{VaR}(X + c) \geq \Lambda\text{VaR}(X) - c.
% \]
% By definition of \(\Lambda\text{VaR}\), this inequality is equivalent to:
% \[
% \inf\{y \in \mathbb{R} \colon F_{X + c}(y) > \Lambda(y)\} \leq \inf\{y \in \mathbb{R} \colon F_X(y) > \Lambda(y)\} + c.
% \]
% Using the fact that \(F_{X + c}(y) = F_X(y - c)\), we can rewrite the left-hand side as:
% \[
% \inf\{y \in \mathbb{R} \colon F_X(y - c) > \Lambda(y)\}.
% \]
% Similarly, the right-hand side becomes:
% \[
% \inf\{y + c \in \mathbb{R} \colon F_X(y) > \Lambda(y)\} = \inf\{z \in \mathbb{R} \colon F_X(z - c) > \Lambda(z - c)\}.
% \]
% Hence, it is sufficient to show the set inclusion:
% \[
% \{y \in \mathbb{R} \colon F_X(y - c) > \Lambda(y)\} \supseteq \{y \in \mathbb{R} \colon F_X(y - c) > \Lambda(y - c)\}.
% \]
% This inclusion holds whenever \( \Lambda(y - c) \geq \Lambda(y)\), which is guaranteed if \(\Lambda(y)\) is decreasing in \(y\). Thus, the inequality follows, completing the proof.
\end{proof}

\subsection{Ranking metrics generated by bibliometric indices}
\label{sub_biblio_index}
%\section{Bibliometric index based ranking metrics}
An alternative class of ranking measures is given by the so-called scientific research measures (SRM) introduced by \cite{frittelli2016scientific}. This new class has been introduced in the bibliometric literature as a generalization of several bibliometric indices, including the $h$-index, the $h$-alpha-index, the maximum number of citations, the maximum number of publications etc. 

Here, we assume there exists a portfolio with $N$ assets (alternatively, we can think about a fund of portfolios) whose returns are denoted by $X_1, \ldots X_p, \ldots, X_N$, respectively. We denote with $X_{(1)}, \ldots X_{(p)}, \ldots, X_{(N)}$ the same portfolio whose components are ranked in decreasing order. With a slight abuse of notation, we set $X(p):=X_{(p)}$ if $p \in \left \{ 1, \ldots, N \right\}$ and $X(p)=0$ if $p>N$. In what follows, we consider a more general setting where a portfolio is identified with a function from $\mathbb{R}_+$ to $\mathbb{R}_+$, i.e., as a sequence $X=\left\{X(p)\right\}_{p\in \mathbb{R}_+}$. % can be viewed as a random variable from . %\blue{\cancel{$(\mathbb{N},2^{\mathbb{N}})$ to $\mathbb{R}_+$.} [since when we combine we gets real numbers so let's work on reals]} 
In this framework, ranking metrics generated by SRM are maps $r_{\text{SRM}}$ with domain $L^\infty(\mathbb{R}_+) $ (the space of equivalence classes of Lebesgue measurable functions that are essentially bounded, endowed with the essential supremum norm) that provide different ranks based on a preassigned family of performance curves $\{f_x\}$.
\begin{definition}\label{def:SRM}
Let $\{f_x\}_{x>0}$ be a family of performance curves such that $f_x: \mathbb{R}_+ \to \mathbb{R}_+$ for any $x>0$, such that $f_x$ is increasing and left continuous on $x$. Ranking metrics generated by SRM are maps $r_{\text{SRM}} \colon%
 L^{\infty}(\mathbb{R}_+)\rightarrow \lbrack 0,+\infty ]$ defined as follows %\blue{\cancel{L^{\infty}(\mathbb{N},2^{\mathbb{N}})}}
\begin{eqnarray}
r_{\text{SRM}}(X) :=\sup \left\{ x>0 \colon X(p)\geq f_{x}(p) \text{ for all}\ p\in 
\mathbb{R}_+ \ a.s. \right\} .  \label{eq:SRM}
\end{eqnarray} 
\end{definition}
%\blue{\cancel{\mathbb{N}}}
In this context, performance curves indicate the return associated with each ranked asset that is required to achieve a portfolio's performance level $x$. The higher the portfolio's performance, the higher the performance curve $f_{x}$ that it can reach, and the higher the corresponding level $x$. In practice, levels can be taken as integers, i.e. $x \in \mathrm{\mathbb{N}}$. 
%Figure 1 shows two families of performance curves: square type functions corresponding to the $h$-index and power law type functions corresponding to a specific SRM that will be used in our empirical application. Moving from the origin $(0.0)$ of the graph in a northeasterly direction, the research performance level increases, that is, each performance curve $f_q$ is associated to a higher and higher level $q$. 

Different choices of $\{f_x\}_{x>0}$ provide different performance/ranking measures. 
%The following properties are assumed on $\{f_x\}_{x>0}$.
%\begin{assumption}\label{ass:fx}
%    Let $f_x: \mathbb{N} \to \mathbb{R}_+$ for any %$x>0$. 
%    \begin{itemize}
%        \item $f_x$ is increasing in $x$;
%        \item $f_x$ is left continuous on $x$
%    \end{itemize}
%\end{assumption}
For instance, the $h$-index ranking metric is characterized by the performance curve $f_x(p)=x \mathbf{1}_{(0, x]}(p)$, where $\mathbf{1}_{(0, x]}(p)$ takes value 1 when $p \in(0, x]$ and zero otherwise; another trivial example is the max portfolio return that is characterized by $f_x(p)=x \mathbf{1}_{(0,1]}(p)$. Table \ref{tab:research_measures_list} provides further examples of ranking metrics generated by scientific research measures. For a clearer and more insightful interpretation of these new performance measures in the context of finance and insurance, refer to the empirical application presented in Section~\ref{sec: emp_app}.

\begin{table}[phtb]
%\begin{sidewaystable}[phtb]
  \centering
    %\begin{tabular}{l l r r r r}
    \begin{tabular}{c p{5.5cm}p{4cm}p{4cm}}
\multicolumn{1}{l}{Index} &
\multicolumn{1}{c}{Description} &
\multicolumn{1}{l}{Author(s)} &
\multicolumn{1}{l}{$f_x(p)$} \\
\hline
$h$     & A fund has $h$-index ranking metric $h$ if $h$ of its portfolios (or contracts) have at least $h$ percentage return points       &   \cite{hirsch2005index}
& $x \mathbf{1}_{(0,\, x]}(p)$ \\
\\
$h^2$   & A fund has index $h^2$ if $h$ of its portfolios (or contracts) have at least $h^2$ percentage return points   &   \cite{Kosmulski_ISSI2006}     
& $x^2 \mathbf{1}_{(0,\, x]}(p)$ \\
\\
$h_\alpha$ &   A fund has index $h_\alpha$ if $h$ of its portfolios (or contracts) have at least $\alpha h$ percentage return points       &   \cite{Eck_Waltman_JI2008}     
& $\alpha x \mathbf{1}_{(0,\, x]}(p), \alpha > 0$ \\
\\
$w$ &   A fund has index $w$ if $w$ of its portfolios (or contracts) have at least $w, w-1, \ldots, 1$ percentage return points       &   \cite{Woeginger_MSS2008}     
& $(-p+x+1) \mathbf{1}_{(0, \, x]}(p)$ \\
\end{tabular}
  \caption{Bibliometric indices ranking metrics. The table lists popular bibliometric indices, a short description of each index interpreted in the finance context, and the author(s) who introduced the index.
  The $h$-, $h^2$-, $h_\alpha$-, and $w$-index are scientific research measures ranking metrics as defined in~(\ref{eq:SRM}). The last column provides the corresponding performance curves, $f_x(p)$.   }\label{tab:research_measures_list}
\end{table}

By definition, $r_{SRM}$ is monotone increasing and quasiconcave, and this implies $r_{SRM}$ is a ranking metric.

%\textbf{MR: I suggest joining Lemmas 7 and 8 into a single proposition.}
\begin{lemma}
The map $r_{SRM}$ defined in~(\ref{eq:SRM}) is monotone increasing and quasiconcave.
\end{lemma}
\begin{proof}
First, we prove monotonicity. Let $X_1,X_2 \in  L^{\infty}(\mathbb{R}_+)$.  For $X_2 \geq X_1 \text{ for all}\ p\in \mathbb{N}$ we have $X_2(p)\geq X_1(p)\geq f_{x}(p) \text{ for all}\ p\in \mathbb{R}_+$, hence $$\left\{ x>0 \colon X_1(p)\geq f_{x}(p) \text{ for all}\ p\in \mathbb{R}_+\right\} \subseteq \left\{ x>0 \colon X_2(p)\geq f_{x}(p) \text{ for all}\ p\in \mathbb{R}_+\right\}$$This implies $r_{SRM}(X_2) \geq r_{SRM}(X_1)$.

Now, we show that $r_{SRM}$ is quasiconcave. It is sufficient to show that if $r_{SRM}(X_1) \geq x$ and $r_{SRM}(X_2) \geq x$, then $r_{SRM}(\lambda X_1 + (1-\lambda) X_2 ) \geq x$. So, for any $i=1,2$, we have  $\sup \left\{ x>0 \colon X_i(p)\geq f_{x}(p) \text{ for all}\ p\in \mathbb{R}_+ \right\} \geq x$. This means, for any $\epsilon >0$ there exists $x_i$ such that $x_i \geq r_{SRM}(X_i) - \epsilon \geq x - \epsilon$, and $X_i(p)\geq f_{x_i}(p) \text{ for all } \ p\in \mathbb{R}_+$. Since the performance curves $f_x$ are monotone increasing in $x$ by definition, we have, for any $i=1,2$, that $X_i(p)\geq f_{x_i}(p) \geq f_{x -\epsilon}(p) \ \text{for all}\ p\in \mathbb{R}_+$ and for any $\epsilon >0$. This implies
$\lambda X_1 + (1-\lambda) X_2\geq f_{x -\epsilon}(p) \text{ for all }\ p\in \mathbb{R}_+$ and for any $\epsilon >0$. Hence, using the left continuity of $f_x$ we obtain $r_{SRM}(\lambda X_1 + (1-\lambda) X_2 ) \geq x$ and thus $r_{SRM}$ is quasiconcave. 
\end{proof}

In the next lemma we provide a sufficient condition on $\left\{f_x\right\}_{x>0}$ in order to guarantee cash-subadditivity of $r_{SRM}$ at some $k\in \mathbb{R}_+$. This condition works for the restriction of $r_{SRM}$ on 
the positive elements of $L^\infty\subseteq L^\infty(\mathbb{R}_+)$, that we denote by $L^\infty_+$. 

\begin{lemma} \label{lemma:bib-prop}
Let 
$\{f_x\}_{x>0}$ be a family of performance curves as in Definition \ref{def:SRM}.
 Let $k\in \mathbb{R}_+$. If, for any $x>0$, there exists $\ell_{x,k}>0$ such that 
\begin{equation}\label{eq:case_rsm_sub}
    \begin{cases}
        f_{x+k}(p) - k \geq f_x(p), \qquad p \in (0,\ell_{x,k}]\\
    f_{x+k}(p) - k \leq 0=f_x(p), \qquad p > \ell_{x,k},
    \end{cases}
\end{equation} 
then the restriction of $r_{SRM}$ on $L^\infty_+$ is cash-subadditive at $k$.
%\item[iii)] If $f_{x+k} - f_x = k$ for all $k \in \mathbb{R}_{+}$, then $r_{SRM}$ is cash-additive (on ${\rm supp}(r)$).

%\end{itemize}

\end{lemma}

\begin{proof} 
Let $X\in L^\infty_+$. 
We claim that the following inclusion holds:\begin{equation}\label{eq:equalityfx1}\begin{aligned}
     \left\{ y>0: X(p)\geq f_{y+k}(p) -k  \ \forall p>0\right\}\subseteq
     \left\{ y>0: X(p)\geq f_{y}(p) \ \forall p>0  \right\}. 
 \end{aligned}\end{equation}
 In fact, using assumption \eqref{eq:case_rsm_sub}, 
 for any $p \in (0,\ell_{y,k}]$
 we have:
 $$
  X(p) \geq f_{y+k}(p) -k 
  \qquad 
  \Rightarrow
  \qquad 
  X(p) \geq f_y(p).
 $$
 Moreover, since $X\in L^\infty_+$ and $f_y(p)=0$ for any $p> \ell_{y,k}$, the inequality $X(p) \geq 0 = f_y(p)$ holds even in this interval. Thus, we conclude that 
 $X(p) \geq f_y(p)$ for any $p> 0$.
 
 Now we compute $r_{SRM}(X+k)$. We have
\begin{equation*}\begin{aligned}
 r_{SRM}(X+k)&= \sup\left\{x >0: X(p)+k \geq f_x(p), \forall p >0\right\}\\
 &=  \sup\left\{ x >0: X(p) \geq f_x(p)-k\ \forall p >0\right\}\\
  &=  \sup\left\{\left\{x >k: X(p) \geq f_x(p)-k\  \forall p >0\right\} \cup 
  \left\{0<x\leq k: X(p) \geq f_x(p)-k\  \forall p >0\right\} \right\} \\
  &\leq \sup\left\{\left\{x >k: X(p) \geq f_x(p)-k\  \forall p >0\right\} \cup 
 [0,\, k] \right\} \\
 & = \sup\left\{\left\{y+k,y>0: X(p) \geq f_{y+k}(p)-k\  \forall p >0\right\} \cup 
 [0,\, k] \right\} \\
 &= \sup\left\{\left\{y>0: X(p) \geq f_{y+k}(p)-k\  \forall p >0\right\}+k \cup 
 [0,\, k] \right\}\\
 & \leq 
  \sup \left\{y>0: X(p) \geq f_{y+k}(p)-k\  \forall p >0\right\}  +k \\
  &\leq 
  \sup \left\{y>0: X(p) \geq f_{y}(p)\  \forall p >0\right\}  +k =r_{SRM}(X)+k,
 \end{aligned}\end{equation*}
 where in the last inequality we used inclusion \eqref{eq:equalityfx1}.
 This proves cash-subadditivity at $k$.

\end{proof}

% As in the classical bibliometric literature, we define the Hirsch core of $r_{SRM}$ as
% \[
% \mathcal{H}_{r_{SRM}}(X) := \{\, p \in \mathbb{R} \mid f_{r_{SRM}(X)}(p) > 0 \,\}.
% \]
% We denote with $\mathcal{H}_{0} := (0. x] $ the Hirsch core of the $h$-index and with $\mathcal{H}_{k} := \{ (0. x+k] \mid k \in \mathbb{R}_+ \}$ the family of its $k$-shifts.

%\textbf{MR: I have a doubt: if the Hirsch core is defined in terms of $f_r$, then why its shifted and scaled versions do not? Regarding the red part, I think it is OK.}
%, i.e. 
%    $$h(\lambda X) \leq \lambda h(X), \;  \text{ in }  \mathcal{ H}_{h,\lambda}=\{(0. \lambda x] , \; \lambda \in \mathbb{R}_+\}.$$

\begin{lemma}
%and let 
%$
%\mathcal{H}_{h,\lambda} := \{ (0. \lambda x] \mid \lambda %\in \mathbb{R}_+ \}
%$
%denote its $\lambda$-rescalings. 
The following property holds:
\begin{itemize}
\item[i)] The $h$-index is cash-subadditive on $L^\infty_+$.
\item[ii)] The $h_\alpha$-index is cash-subadditive on $L^\infty_+$ for $\alpha > 1$.
\item[iii)] The $h^2$-index is cash-subadditive at any $k\geq 1$ on $L^\infty_+$.
\item[iv)] The $w$-index is cash-subadditive at any $k\in \mathbb{R}_+$.
% \item[ii)] The $h$-, $h^2$- and $h_{\alpha}$-indices are positive-sub-homogeneous under any $\lambda$-rescaling of the Hirsch core, i.e. 
% $$h(\lambda X) \leq \lambda h(X), \;  \text{ in }  \mathcal{ H}_{h,\lambda}=\{(0. \lambda x] , \; \lambda \in \mathbb{R}_+\}.$$
%and positive-homogeneous everywhere. 
%The $w$-index is neither positive super-homogeneous nor positive sub-homogeneous \blue{[check: as it is pos-super-hom only if $\lambda \in (0.1)$, instead if $\lambda > 1$  is pos-sub-hom]} 
\end{itemize}
\end{lemma}

\begin{proof}
The proof proceeds verifying that assumption \eqref{eq:case_rsm_sub} in Lemma \ref{lemma:bib-prop} holds for the bibliometric indices in i)-iv). 
\begin{itemize}
 \item[i)] Let $X \in L^\infty_+$.
 For any $k\in \mathbb{R}_+$ and any $x>0$ we have
 \begin{equation}\label{eq:fx+k}
     f_{x+k}(p)-k=(x+k) \mathbf{1}_{(0,\, x+k]}(p) - 
 k= x \mathbf{1}_{(0,\, x+k]}(p) - 
 k\mathbf{1}_{(x+k,+\infty)}(p). 
 \end{equation}

 By \eqref{eq:fx+k}, for any $p \in (0,\, x+k]$ we have that
 $f_{x+k}(p)-k\geq f_x(p)$. 
 Moreover, whenever $p>x+k$,
 we have $f_{x+k}(p)-k=-k<0=f_x(p)$, 
 since $f_x\equiv 0$ on $(x,+\infty)$.
 This proves condition \eqref{eq:case_rsm_sub} in Lemma \ref{lemma:bib-prop} for the $h$-index with $\ell_{x,k}=x+k$, and, consequently cash-subadditivity. 
 \item[ii)] 
 Since
 % Since \begin{equation}\label{eq:alpha_f}
 %      f_{x+k}(p)- k \geq f_x(p) \ \ \text{for all} \ p \in (0.x+k]. 
 % \end{equation}
 for the increasing family associated to the $h_{\alpha}$-index we have 
 \begin{equation*}
 \begin{aligned}
      f_{x+k}(p)- k&= \alpha (x+k)\mathbf{1}_{(0,\, x+k]}(p)-
     k\\
     & = 
    ( \alpha x + (\alpha-1)k )\mathbf{1}_{(0,\, x+k]}(p)- k\mathbf{1}_{(x+k,+\infty)}(p),
 \end{aligned}   
 \end{equation*}
 we deduce that condition \eqref{eq:case_rsm_sub}  is satisfied for every $\alpha >1$ with $\ell_{x,k}=x+k$.

 \item[iii)] 
 For the increasing family associated to the $h^2$-index we have\begin{equation}\label{eq:fx+kh2}
 \begin{aligned}
     f_{x+k}(p)-k&=(x+k)^2 \mathbf{1}_{(0,\, x+k]}(p) - 
 k\\
 & = (x^2+2kx+k^2-k) \mathbf{1}_{(0,\, x+k]}(p) - 
 k\mathbf{1}_{(x+k,+\infty)}(p). 
 \end{aligned}
 \end{equation}

 We observe that on $(0,\, x+k]$ the term $f_{x+k}(p) -k$ satisfies
 $$
 f_{x+k}(p) -k = (x^2+2kx+k^2-k)\mathbf{1}_{(0,\, x+k]}(p) \geq 
 x^2\mathbf{1}_{(0,\, x+k]}(p) =f_x(p)
 $$
 if  $k^2-k \geq 0$, i.e., if  $k\geq 1$.
 This proves condition \eqref{eq:case_rsm_sub} with $\ell_{x,k}=x+k$ and, consequently, cash-subadditivity for $h^2$ at any $k\geq 1$.

 \item[iv)] 
 We divide the proof in two parts. First, we consider $k\geq 1$. 
 In this case we can prove that
 \begin{equation*}
      \begin{cases}
          f_{x+k}(p) -k \geq f_{x}(p), \text{ for all\ } p \in (0,\, x+1]\\
          f_{x+k}(p)-k\leq 0 \text{ for all\ } p > x+1.
      \end{cases} 
 \end{equation*}
 In fact, we have
 $$
 f_{x+k}(p) -k= 
 \begin{cases}
     -p+x+1= f_x(p), \qquad p \in (0,\, x+1]\\
     -p+x+1\leq 0=f_x(p), \qquad p \in (x+1,x+k]\\
     -k \leq 0= f_x(p), \qquad p>x+k.
 \end{cases}
 $$
 This proves condition \eqref{eq:case_rsm_sub} with $\ell_{x,k}=x+1$ and, consequently, $w$-index is cash-subadditive at any $k\geq 1$. 
 
 Second, we consider $k \in (0,1)$. By calculations similar to those above, we obtain
 \begin{equation*}
      \begin{cases}
          f_{x+k}(p) -k \geq f_{x}(p), \text{ for all\ } p \in (0,\, x+k]\\
          f_{x+k}(p)-k\leq 0=f_x(p), \text{ for all\ } p > x+k.
      \end{cases} 
 \end{equation*}
 This proves that condition \eqref{eq:case_rsm_sub} 
 remains true even for $k\in (0,1)$ with $\ell_{x,k}=x+k$, and, consequently that the $w$-index is cash-subadditive even for $k\in (0,1)$. In conclusion, $w$-index is cash-subadditive at any $k\in \mathbb{R}_+$. 
\end{itemize}
\end{proof}
% \textbf{MR: this Omega ration example could also be deleted.}

%\subsection{Illustrative simulation examples}

%In this section, we present an illustrative example to demonstrate the computation of ranking metrics derived from scientific research measures. The results are summarized in Table xxx. Consider two funds, each comprising five portfolios. The returns of each portfolio are assumed to follow a t-distribution with specified degrees of freedom. At any time $t$, the value of a portfolio is modeled as the expected value of the portfolio's return over 500 simulations, scaled by its weight. Specifically, we define the portfolio value as $\hat{X}_{p}=\omega_p E[X_p]$, $p=1,\dots ,5$, where $\omega_p$ represents the proportion of the total fund invested in the $p$-th portfolio (shown in Column xxx of Table xxx). Then, we rank the vector of expected portfolio $\hat{X}_{(p)}$ in descending order, denoting the ranked values as $\hat{X}_{(1)},\hat{X}_{(2)}, \dots, \hat{X}_{(5)}$, where we set $\hat{X}_{(p)}=0$ for portfolios with $\hat{X}_{(p)}<0$ (Column xxx of Table xxx). To forecast the $h$-index of the fund at time $t$, we rescale this vector by dividing each value by the minimum expected positive value (shown in Column xxx of Table). This allows to have both the axis of the portfolios returns and the number of portfolios on the same scale. The $h$-index of the fund is given by xxx.

\section{Empirical application}\label{sec: emp_app}
Using real-world financial and climate risk datasets, this section illustrates the practical relevance of the proposed ranking metrics through empirical analyses of investment portfolios and European climate risk insurance markets.

\subsection{Financial portfolio rankings}
This empirical analysis compares the rankings of financial investments obtained using both traditional risk-adjusted performance measures and the novel $\Lambda$-quantiles and bibliometric index-based performance metrics. Using daily price data from the NASDAQ and S\&P 500 indices (January 2023 - December 2023, sourced from Bloomberg), we construct 15 equally weighted portfolios, with 27 assets each. Table~\ref{tab:empirical_quantiles} reports the empirical annual mean, annual standard deviation, and the daily quantiles at levels $0.1, 0.5, 0.9$ of the portfolios' returns.

\begin{table}[h]
\centering
\caption{Empirical statistics of the portfolio returns}
\tiny{
\begin{tabular}{r|rrrrrr}
\toprule
% & annual & annual & &  &  \\
Portfolio & annual $\mu$ & annual $\sigma$& $q_{0.1}$ & $q_{0.5}$ & $q_{0.9}$ \\
\midrule
$P_1$  & -0.223 & 0.230 & -0.016 & 0.0000 & 0.017 \\
$P_2$  & 0.086  & 0.208 & -0.016 & 0.0009 & 0.015 \\
$P_3$  & 0.037  & 0.154 & -0.011 & 0.0003 & 0.012 \\
$P_4$  & -0.025 & 0.174 & -0.013 & 0.0006 & 0.011 \\
$P_5$  & 0.019  & 0.119 & -0.009 & 0.0000 & 0.009 \\
$P_6$  & 0.019  & 0.163 & -0.013 & -0.0001 & 0.012 \\
$P_7$  & 0.214  & 0.114 & -0.007 & 0.0004 & 0.010 \\
$P_8$  & 0.220  & 0.080 & -0.004 & 0.0005 & 0.007 \\
$P_9$  & 0.164  & 0.108 & -0.008 & 0.0005 & 0.009 \\
$P_{10}$ & 0.240  & 0.079 & -0.005 & 0.0011 & 0.008 \\
$P_{11}$ & -0.203 & 0.199 & -0.019 & 0.0003 & 0.013 \\
$P_{12}$ & -0.040 & 0.139 & -0.012 & 0.0002 & 0.009 \\
$P_{13}$ & 0.034  & 0.145 & -0.011 & 0.0010 & 0.010 \\
$P_{14}$ & -0.037 & 0.111 & -0.010 & 0.0006 & 0.008 \\
$P_{15}$ & -0.115 & 0.097 & -0.008 & 0.0000 & 0.006 \\
\bottomrule
\end{tabular}
}
\label{tab:empirical_quantiles}
\end{table}

In this analysis, we evaluate historical performance, estimating expectations and conditional expectations using empirical averages based on the available daily return data.

\paragraph{Traditional and quantile-based ranking metrics} RAROC is computed as the ratio of expected return to the historical Conditional Value-at-Risk at 0.05, $\text{CVaR}_{5\%}$, allowing for a direct comparison with the GLR, which instead is normalized using the expectation of the negative part of the returns. Omega is also computed. 

The $\Lambda$-quantile ranking metric, $r_{\Lambda \text{VaR}}$, is computed using formula \eqref{def:r_LambdaVaR}, this means it will return the greater between zero and the negative of the $\Lambda$VaR, which corresponds to the $\Lambda$-quantile itself, $q_\Lambda$. Low confidence levels are generally avoided, as they would often yield zero, lacking meaningful interpretation in the context of performance measures. The $\Lambda$-quantile is implemented by assuming an increasing two-steps $\Lambda$ function, taking values either equal to $\lambda_\text{min}$ or $\lambda_\text{max}$, $\lambda_\text{min},\lambda_\text{max} \in (0,1)$ with the shift occurring at a threshold value $\bar{x} = q_\tau$, where $q_\tau$ is the quantile at level $\tau$ (here, $\tau=\frac{\lambda_{\max}+\lambda_{\min}}{2}$) estimated taking the historical returns of all portfolios. Hence, the lambda quantile ranking metric is equal to the quantile at lower level $q_{\lambda_\text{min}}$, before $\bar{x}$, or to the quantile at higher level $q_{\lambda_\text{max}}$, above $\bar{x}$. In the simplest case, $\Lambda$ is a constant $\lambda$, this performance measure coincides with the quantile at $\lambda$ confidence level. At high-level quantiles, this measure highlights portfolio performances in the upper tail, meaning that portfolios are ranked based primarily on their best outcomes. At a confidence level of 50\%, the measure coincides with the median (if positive), ranking portfolios according to their central tendency.  When considering two confidence levels, as in this analysis, a similar logic applies. For assessing upper-tail performances, the portfolio manager may select high confidence levels (e.g., 85\%–95\%), focusing on returns in that range. In this context, the threshold $\bar{x}$ is set as the 90th percentile of all returns. Conversely, to assess performance slightly above the median, thus avoiding negative values truncated to zero, the manager may select confidence levels just over 50\% (e.g., 55\%–65\%), centered at $\tau = 60\%$. This enables ranking portfolios based on their typical performance.
 
Figure \ref{fig:rankingRAPM} displays the rankings of the 15 portfolios obtained by ordering their metric values in descending order for RAROC, Omega, and the $\Lambda$‑quantiles\footnote{GLR is omitted as it yields rankings nearly identical to RAROC.}. These rankings offer a comparative perspective on how the selection of different performance measures influences portfolio classification. As expected, RAROC and GLR provide almost the same ranking as both are structured in a similar way (expected return divided by a measure of downside risk). Specifically, portfolios $P_8$, $P_{10}$, and $P_7$ consistently appear among the top-ranked. Omega provides additional insights comparing only positive returns, which results in a ranking similar to RAROC and GLR for top-performing portfolios, while also shedding light on positive performances of portfolios with a negative expected returns.
Rankings based on $\Lambda$-quantile performance vary depending on the chosen $\Lambda$ levels. Comparing the rankings derived from the 55\%–65\% $\Lambda$-quantiles with those based on traditional performance measures is particularly insightful, as the former range focuses on performance near the median. \deleted{, offering a more robust and stable classification.} In this case, alongside P10. portfolios P2 and P4 also emerge as strong performers. In contrast, focusing on the right tail, specifically the 85\%–95\% $\Lambda$-quantile, allows us to compare portfolios based on their potential to generate large gains. Under this criterion, the ranking favors portfolios such as $P_2$, $P_1$, and $P_{11}$, highlighting their strong performance in the upper tail of the return distribution.

\begin{figure}[h]
    \centering   
    \includegraphics[width=0.48\textwidth, height=4.5cm]{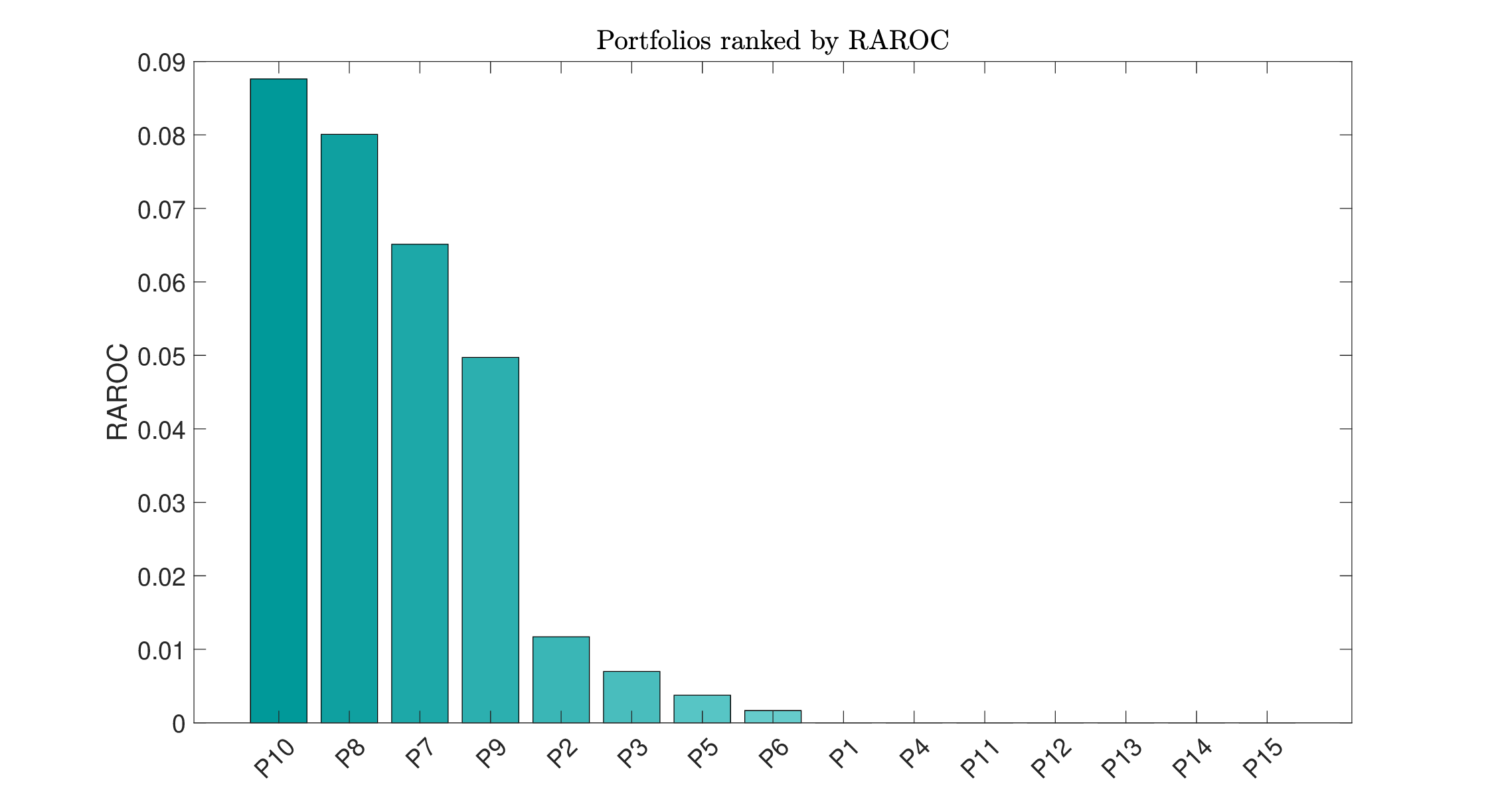}    \includegraphics[width=0.48\textwidth, height=4.5cm]{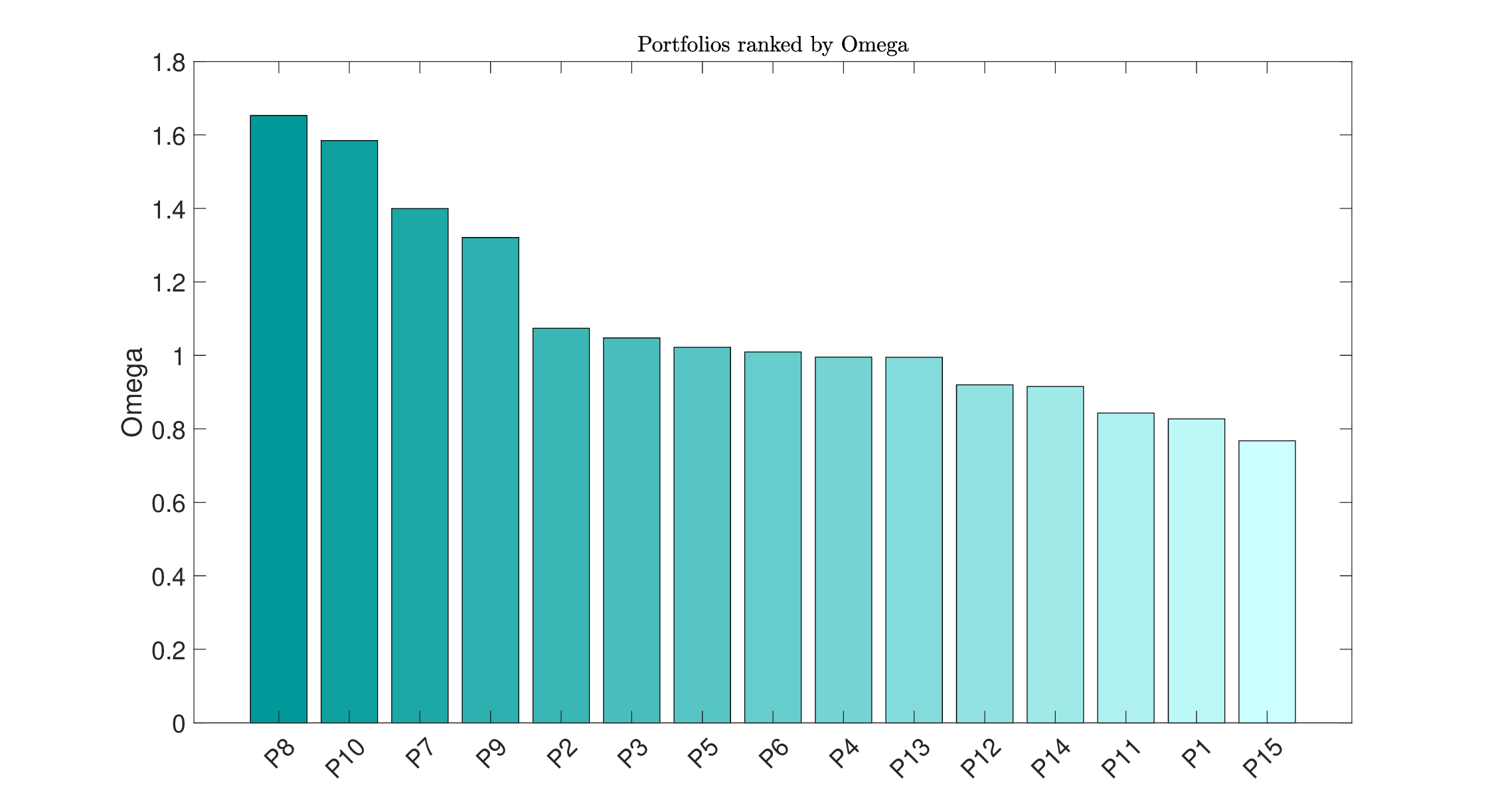}\\
    \includegraphics[width=0.48\textwidth, height=4.5cm]{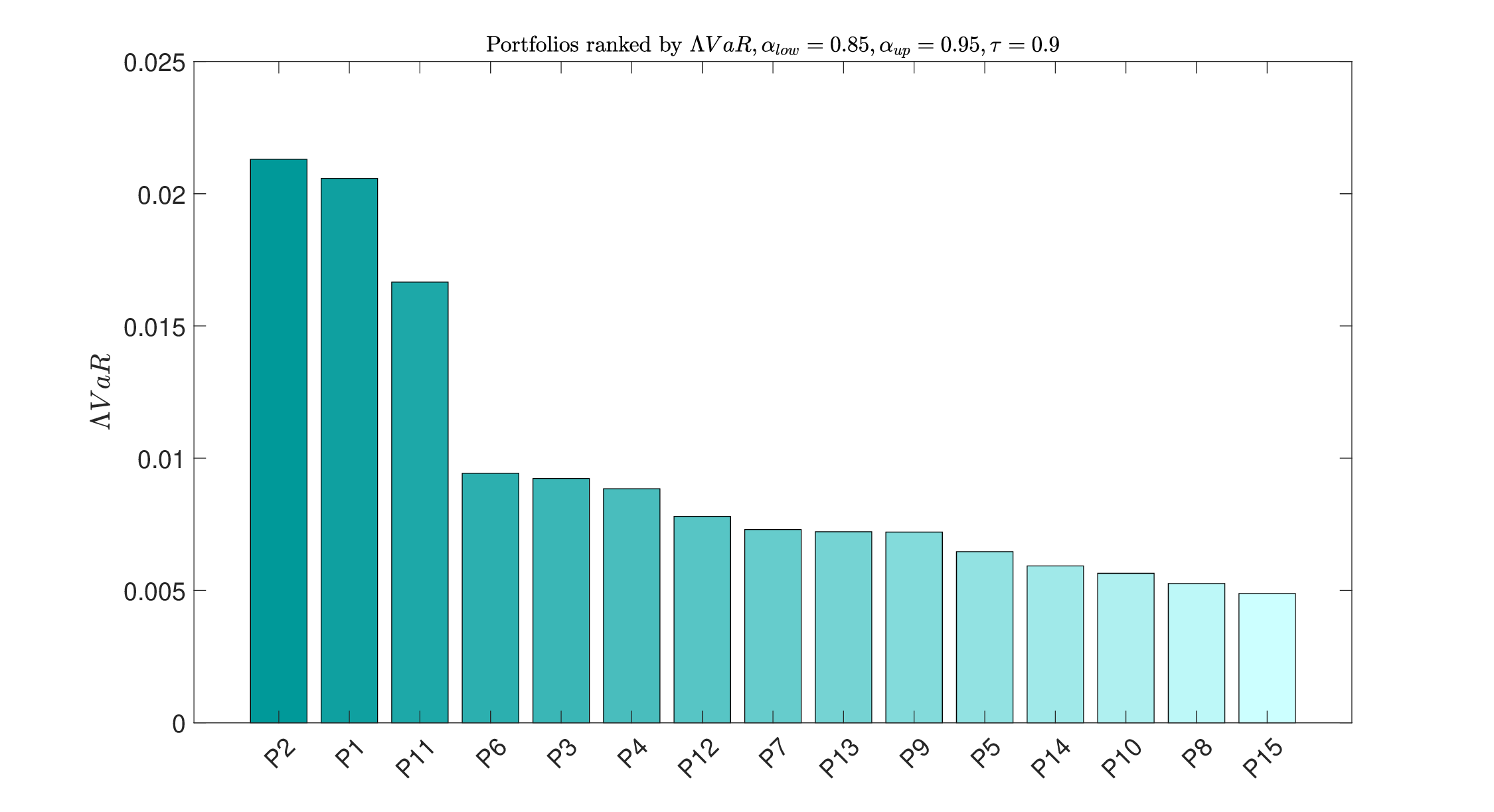}
  \includegraphics[width=0.48\textwidth, height=4.5cm]{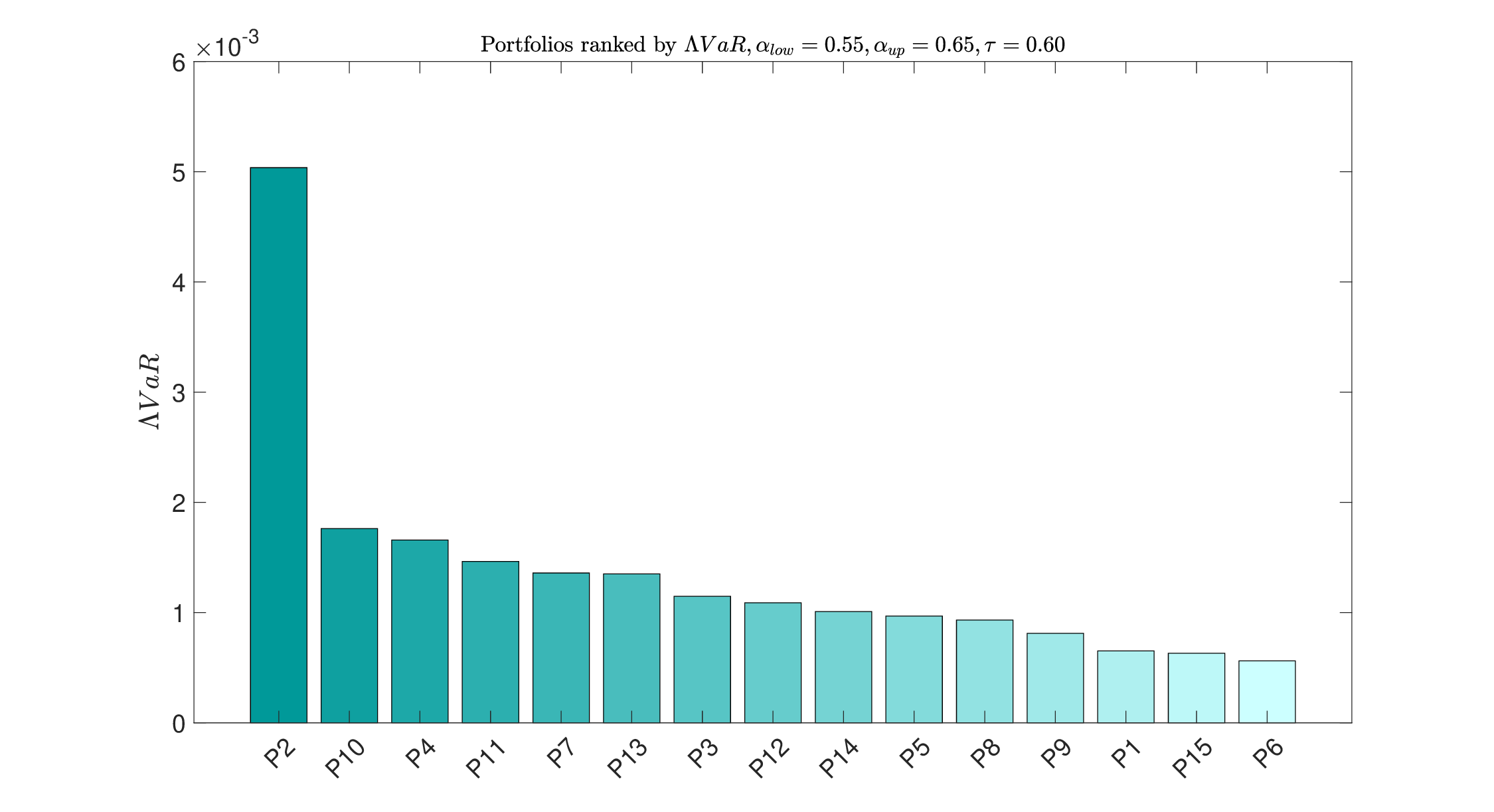}
    \caption{Portfolios' performance rankings using RAROC, Omega and $\Lambda$-quantiles.}
    \label{fig:rankingRAPM}
\end{figure}

\paragraph{Bibliometric index-based ranking metrics} In addition, we assess portfolio rankings using the bibliometric index-based performance measures introduced in Subsection~\ref{sub_biblio_index}; Figure~\ref{fig:rank_biblio} reports the results for all 15 portfolios. 
\begin{figure}[h]
    \centering   
    \includegraphics[width=0.48\textwidth, height=4.5cm]{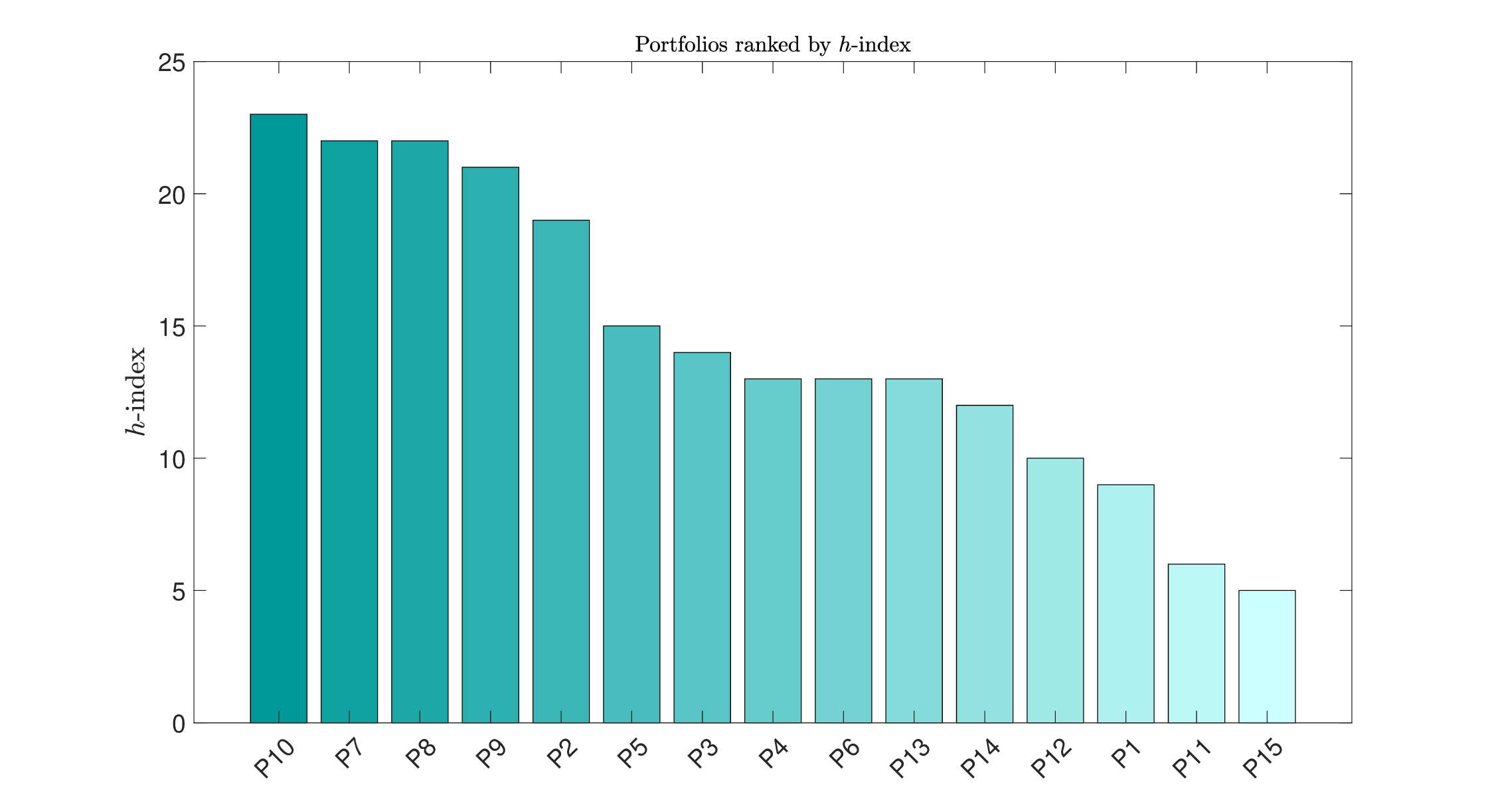}    \includegraphics[width=0.48\textwidth, height=4.5cm]{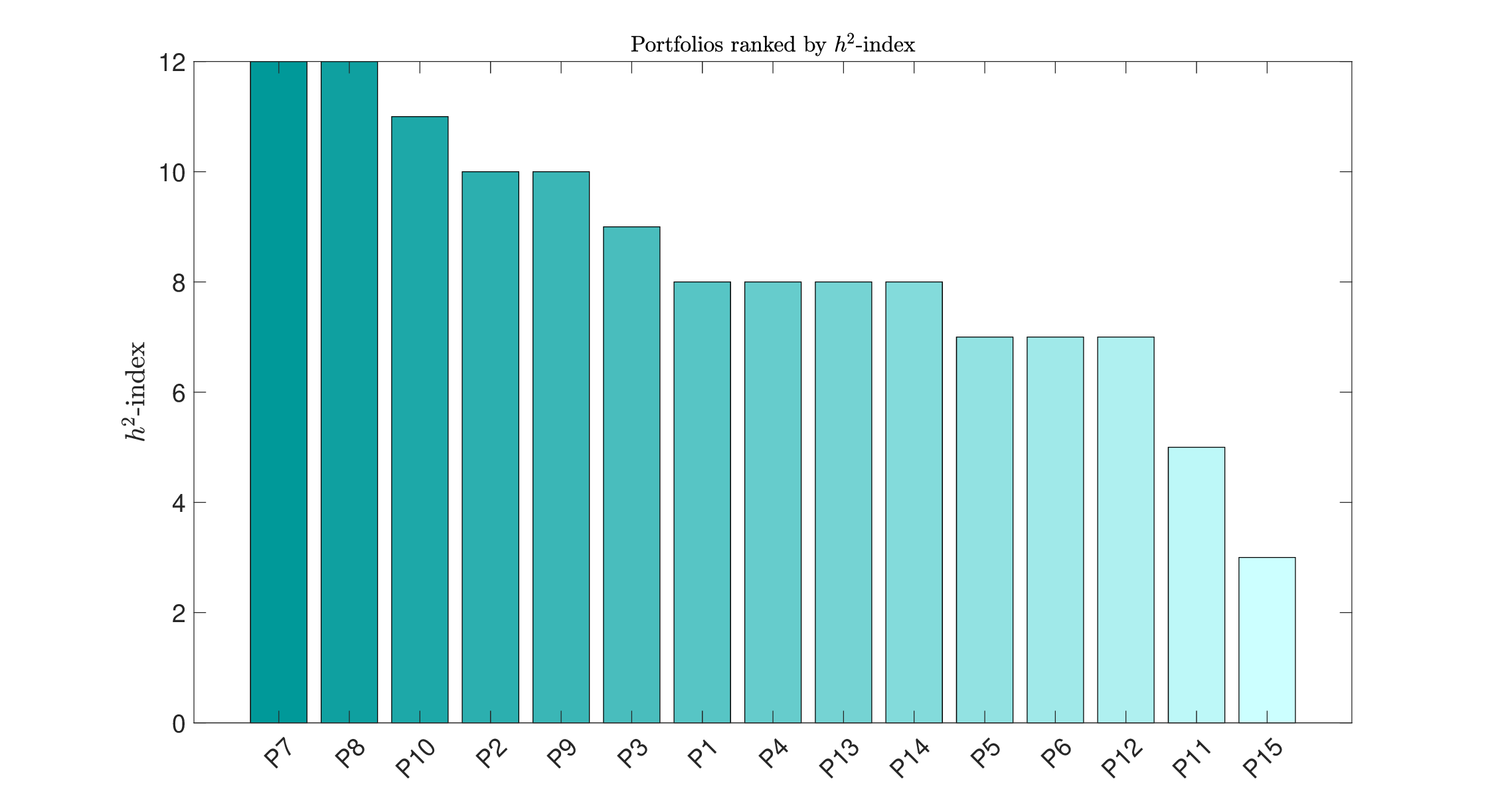}\\
    \includegraphics[width=0.48\textwidth, height=4.5cm]{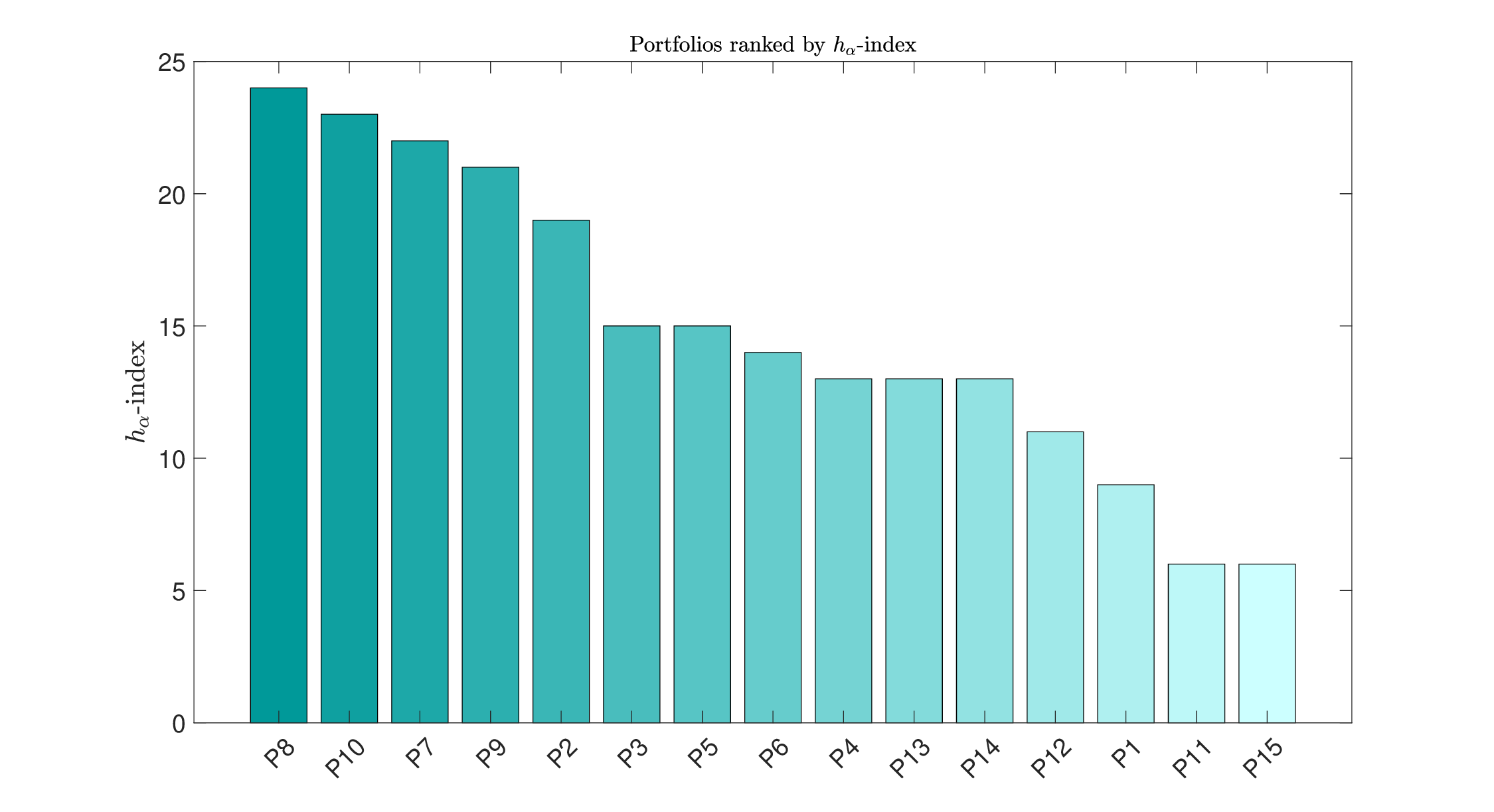}
  \includegraphics[width=0.48\textwidth, height=4.5cm]{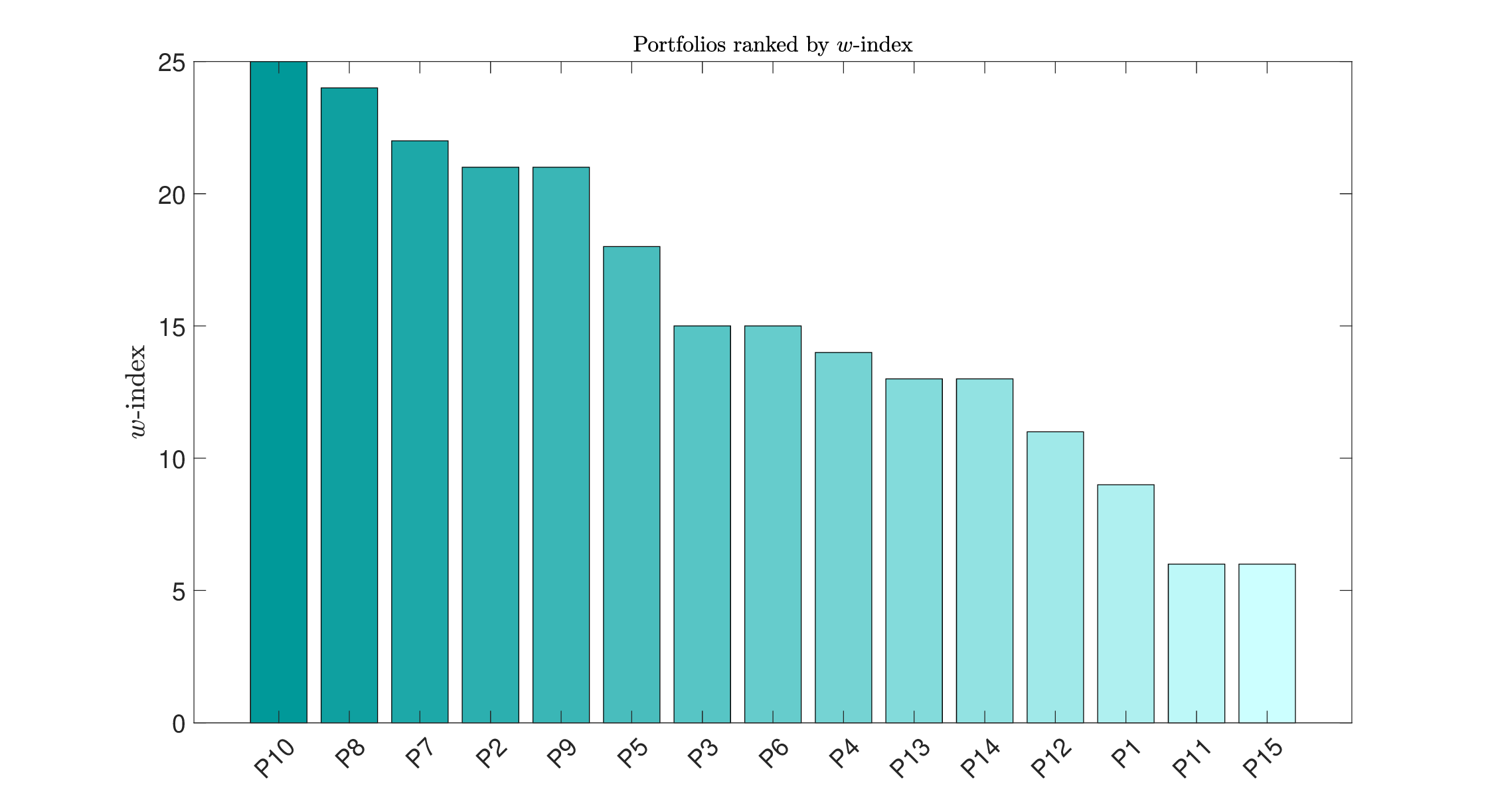}
    \caption{Portfolios' performance rankings using bibliometric indices}
     \label{fig:rank_biblio}
\end{figure}
Specifically, we compute the $h$-index, $h^2$-index, $h_\alpha$-index with $\alpha=0.5$ ($h_{0.5}$) and the $w$-index for each portfolio. Portfolios $P_{10}$, $P_8$, and $P_7$ consistently rank highly across all bibliometric metrics, reflecting a strong alignment with traditional measures such as RAROC and GLR. However, bibliometric rankings provide unique insights into the internal distribution of performances among portfolio’s assets, capturing aspects not fully reflected by conventional performance metrics. For instance, in the case of the $h$-index, a portfolio manager can use this metric to conclude that the portfolio has an $h$-index of $h$, meaning it contains at least $h$ assets whose average returns are each at least $h$ times the benchmark return, and thereby gain insight into the distribution of returns within the portfolio. 

We now provide a detailed description of how the bibliometric indices are constructed. To facilitate comparison with rankings obtained from classical performance measures and facilitate the interpretation, we first construct the $N$-dimensional vector $Y^j = \{Y^j_{p}\}_{p=1,\dots,N}$, which characterises the $j$-th portfolio return structure and serves as the basis for computing each index. For each asset $p$, the expected return $E[Y^j_p]$ is estimated using its \deleted{250-day} historical arithmetic mean, after which the weights $w_p$ (equal in this case, $w_p=\frac{1}{N}$) are applied. Negative expected returns are set to zero, as they provide no informative contribution to the ranking. The weighted asset-level expected returns are then rescaled relative to a benchmark - here, the minimum positive expected return across all portfolios -  and discretised into integer values to facilitate the interpretation. This produces, for each asset $p$:
\begin{equation}\label{eq:Xindices}
X^j_p := \frac{E\left[ w_p Y^j_p \right]^{+}}{\min\limits_{j,p} E\left[ w_p Y^j_p \right]^{+}}.
\end{equation}
The vector $X^j = \{X^j_{(p)}\}_{(p)=1,\dots,N}$ is obtained by sorting the $X^j_p(\cdot)$ in decreasing order. These sorted values are displayed as histogram bars in Figure~\ref{fig:biblio_index_P6}, in the particular case of $X^6$ which illustrates the index construction for portfolio $P_6$. The bars represent the rescaled average returns of the assets, ranked from highest to lowest.

The bibliometric indices are obtained by comparing the histogram bars with families of performance curves characterized by upward-increasing levels. 
Each level represents a performance threshold defined as the minimum average return that each asset must achieve for the portfolio to be assigned that level - expressed in relative terms as the number of times the asset’s return exceeds the benchmark return (here, the minimum return across all assets). These performance curves thus establish a consistent and structured criterion for assessing the distribution of asset returns within the portfolio.

Each index arises from a distinct family of performance curves, imposing different structural conditions on how the portfolio’s return distribution is assessed. For example, the $h$-index uses rectangular curves at height $x$, capturing the smallest integer $h$ such that at least $h$ assets each outperform the benchmark by a factor of $h$ or more. The enlarged view of the top-left subplot in Figure~\ref{fig:biblio_index_P6} clearly illustrates the $h$-index performance curves, specifically highlighting the curves at levels 12, 13, and 14. In particular, for portfolio $P_6$ all bars exceed the level-13 curve but not the level-14 curve, giving an $h$-index of 13. This indicates that at least 13 assets in the portfolio have average returns at least 13 times greater than the benchmark return (taken as the minimum return overall). 
 This corresponds to roughly 85\% of the portfolio constituents.

In contrast, the $h^2$-index (top-right subplot) is constructed comparing the histogram bars with a squared performance curve. $P_{6}$ has an $h^2$-index of 7, indicating that 7 of its assets achieve average returns at least $7^2$ times the benchmark. The $h_\alpha$-index (with $\alpha = 0.5$) equals 14, indicating 14 assets with average returns at least $0.5 \times 14$ times the benchmark. Finally, the $w$-index of 15 means that $P_{6}$ contains 15 assets whose maximum average return performance decreases linearly - starting from 15 times the benchmark for the top-performing asset, 14 times for the next, and so on, down to 1 times the benchmark for the 15th asset.

Additional bibliometric-index-based ranking metrics can be created by modifying the family of performance curves.
A manager can design these curves to specify the required average return of all portfolio constituents - relative to a chosen benchmark - to attain a given performance level.
Exploring such flexible curve specifications offers a promising direction for future research, enabling more tailored assessments of portfolio performance.

\begin{figure}[h!]
    \centering
    \includegraphics[width=0.48\textwidth, height=4.5cm]{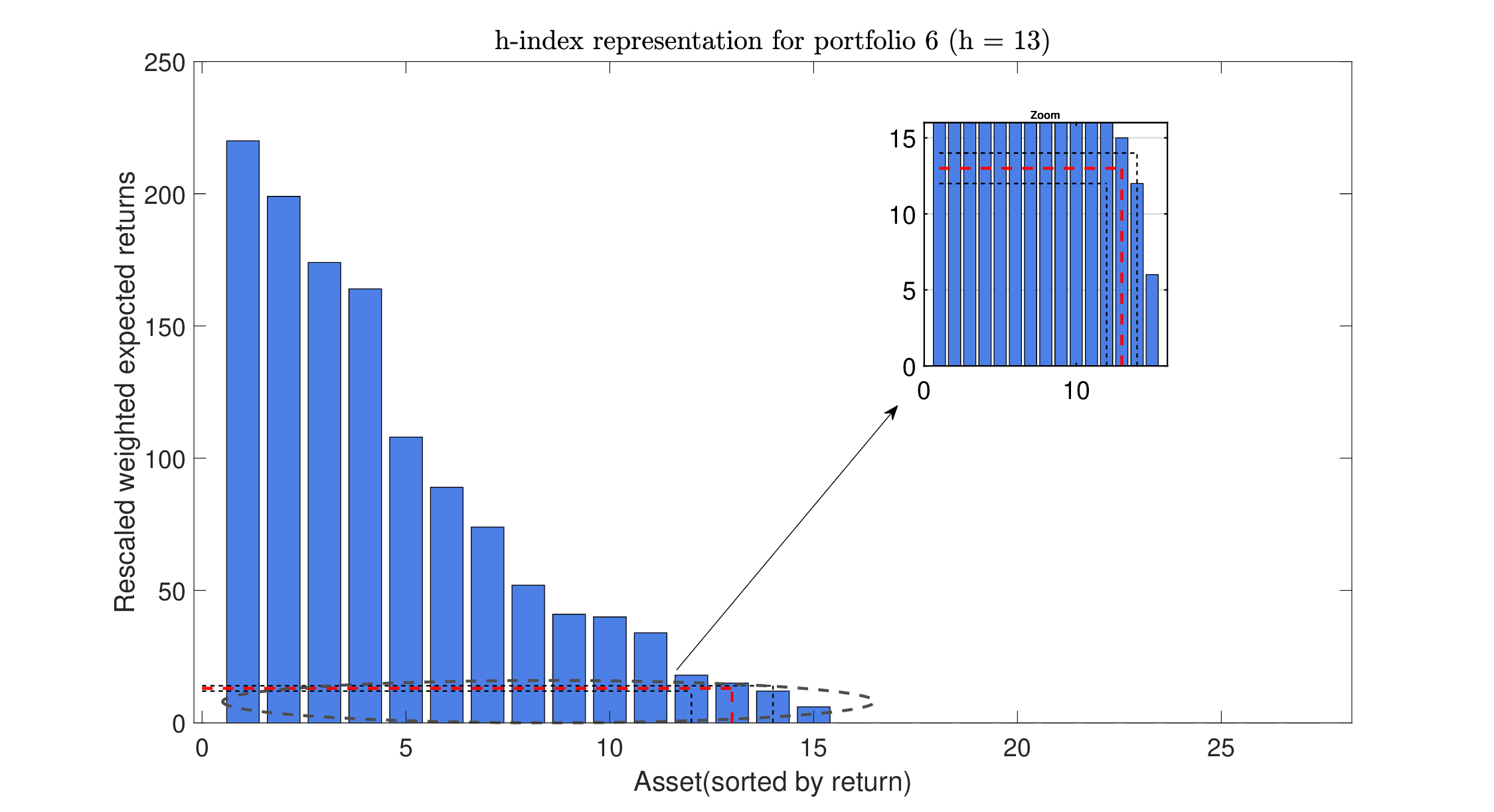}
     \includegraphics[width=0.48\textwidth, height=4.5cm]{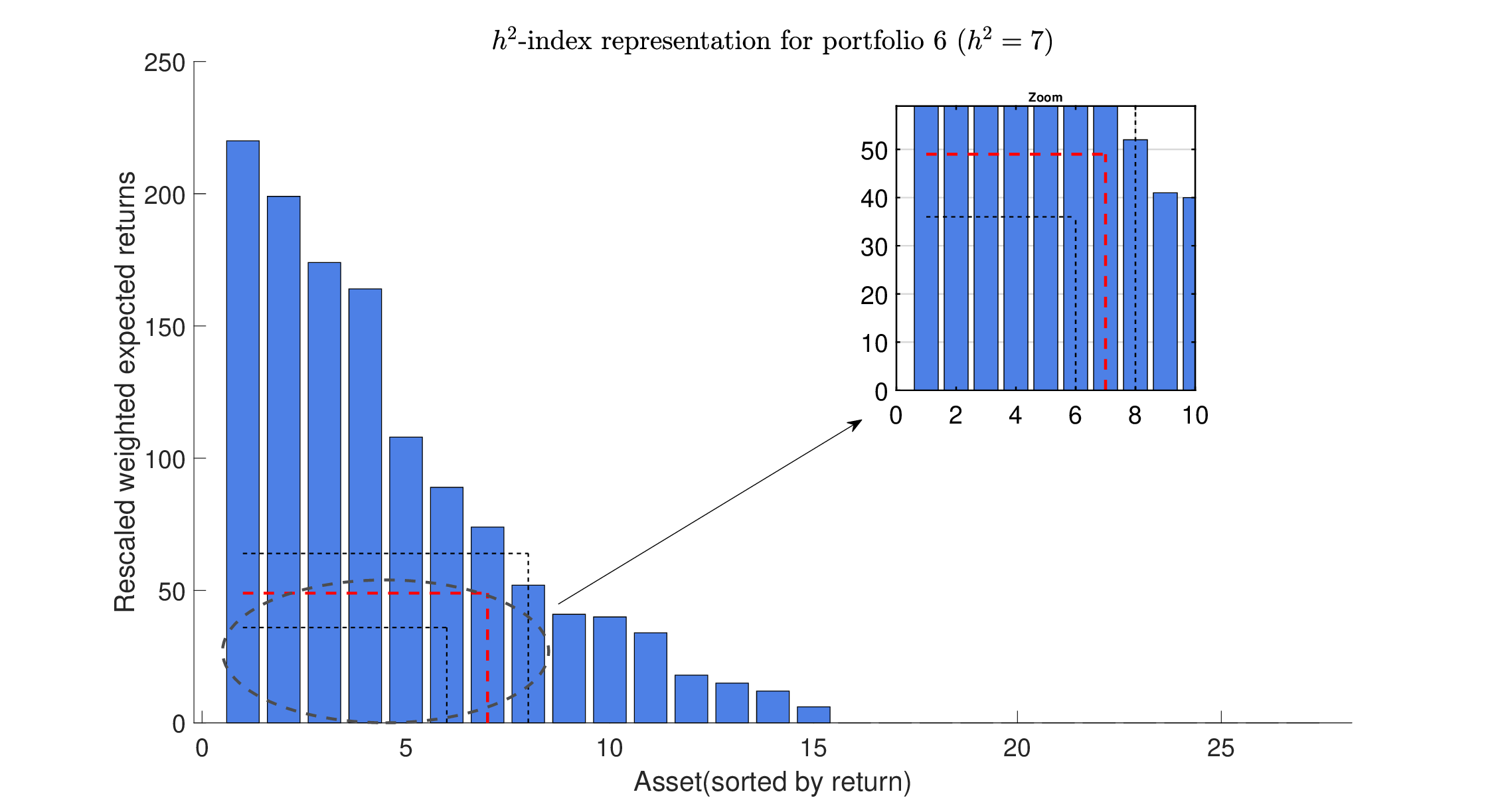}\\       \includegraphics[width=0.48\textwidth, height=4.5cm]{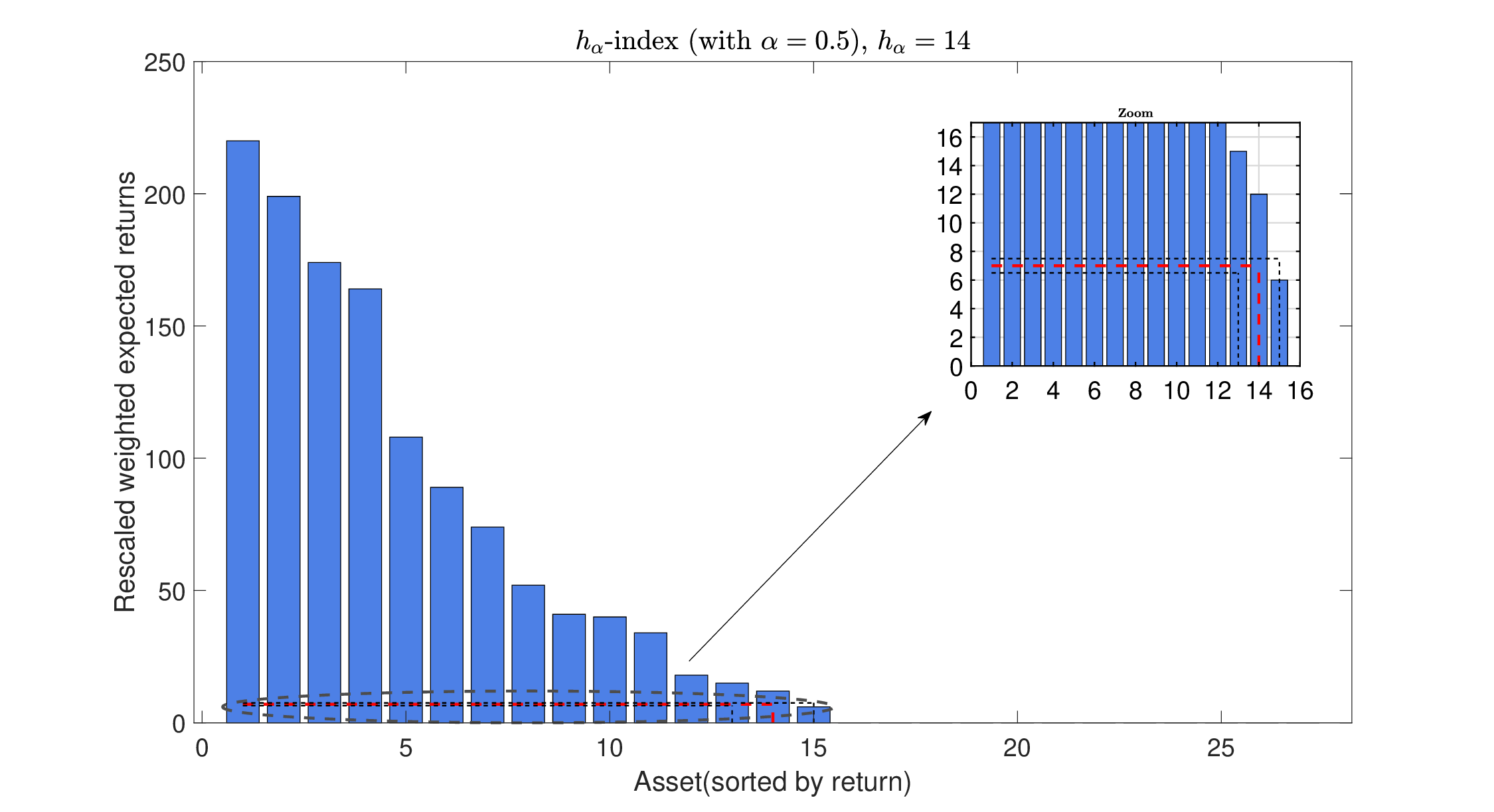}
     \includegraphics[width=0.48\textwidth, height=4.5cm]{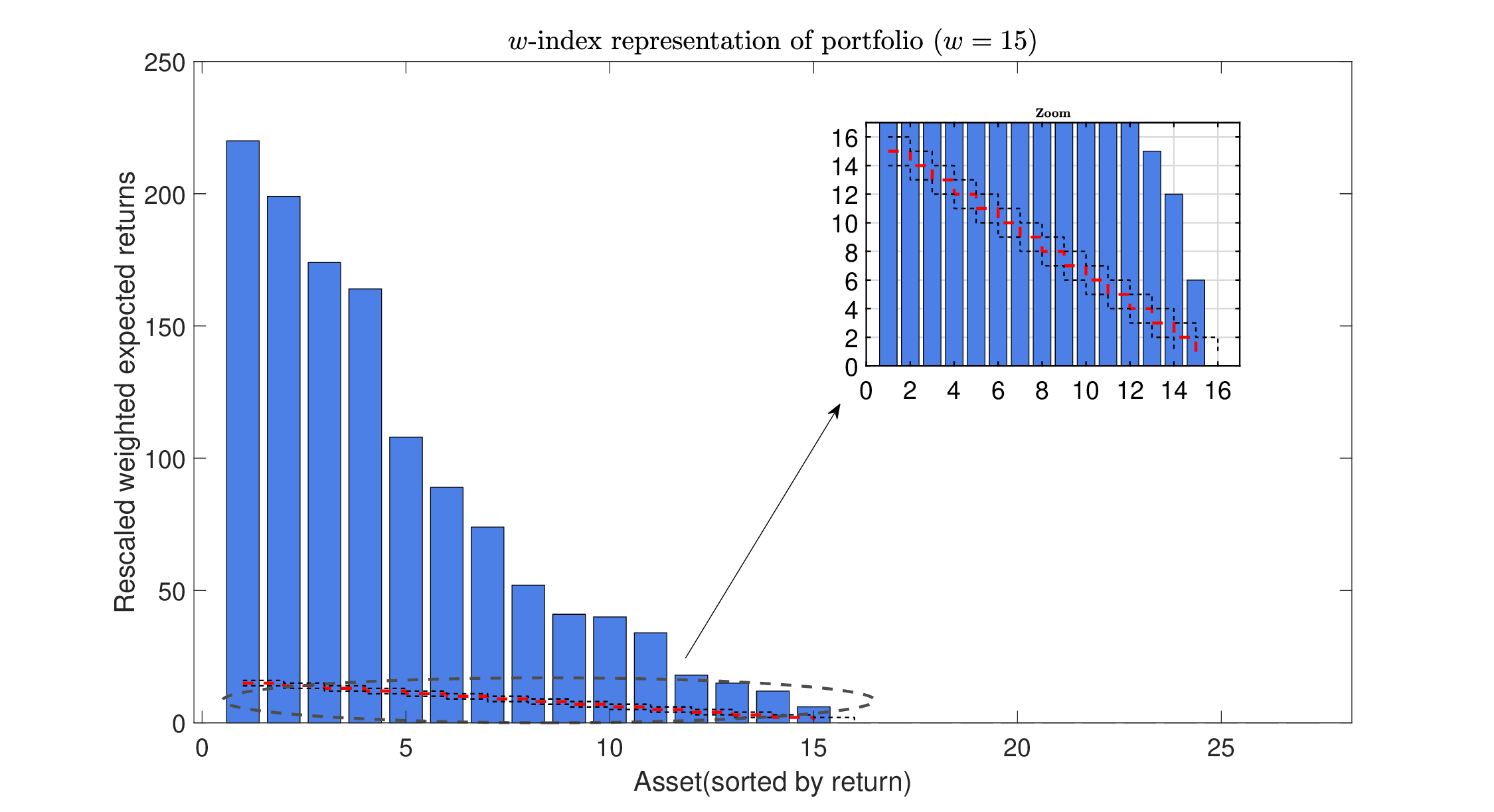}
     \caption{
     Construction of bibliometric indices for portfolio $P_6$. 
In each subfigure, the x-axis lists assets (ordered by expected return) 
and the y-axis shows weighted expected returns rescaled to integer units.}  \label{fig:biblio_index_P6}
\end{figure}

A synthesis of the ranking metric values across all evaluated portfolios is provided in Table \ref{tab:biblioindices}. 

\begin{table}[h!]
\centering
\tiny{
\begin{tabular}{lccccccccc}
\toprule
Portfolio & RAROC & GLR &  $\Omega$ & $(0.85-0.95)\Lambda$VaR &  $(0.55-0.65)\Lambda$VaR & $h$-index & $h^2$-index & $h_{0.5}$-index & w-index \\
\midrule
$P_1$  & 0.0000      & 0.0000     & 0.8271 & 0.0206 & 0.0007 & 9  & 8  & 9  & 9  \\
$P_2$  & 0.0117 & 0.0739 & 1,0739 & 0.0213 & 0.0050 & 19 & 10 & 19 & 21 \\
$P_3$  & 0.0070 & 0.0474 & 1,0474 & 0.0092 & 0.0011 & 14 & 9  & 15 & 15 \\
$P_4$  & 0.0000 & 0.0000 & 0.9954 & 0.0088 & 0.0017 & 13 & 8  & 13 & 14 \\
$P_5$  & 0.0038 & 0.0218 & 1,0218 & 0.0065 & 0.0010 & 15 & 7  & 15 & 18 \\
$P_6$  & 0.0017 & 0.0093 & 1,0093 & 0.0094 & 0.0006 & 13 & 7  & 14 & 15 \\
$P_7$  & 0.0651 & 0.3995 & 1,3995 & 0.0073 & 0.0014 & 22 & 12 & 22 & 22 \\
$P_8$  & 0.0801 & 0.6528 & 1,6528 & 0.0053 & 0.0009 & 22 & 12 & 24 & 24 \\
$P_9$  & 0.0497 & 0.3208 & 1,3208 & 0.0072 & 0.0008 & 21 & 10 & 21 & 21 \\
$P_{10}$ & 0.0876 & 0.5848 & 1,5848 & 0.0057 & 0.0018 & 23 & 11 & 23 & 25 \\
$P_{11}$ & 0.0000      & 0.0000      & 0.8428 & 0.0167 & 0.0015 & 6  & 5  & 6  & 6  \\
$P_{12}$ & 0.0000 & 0.0000      & 0.9198 & 0.0078 & 0.0011 & 10 & 7  & 11 & 11 \\
\color{red}{$P_{13}$} & 0.0000      & 0.0000      & 0.9948 & 0.0072 & 0.0014 & 13 & 8  & 13 & 13 \\
$P_{14}$ & 0.0000      & 0.0000      & 0.9153 & 0.0059 & 0.0010 & 12 & 8  & 13 & 13 \\
$P_{15}$ & 0.0000      & 0.0000      & 0.7676 & 0.0049 & 0.0006 & 5  & 3  & 6  & 6  \\
\bottomrule
\end{tabular}
}
\caption{Summary of rankings metric values across all portfolios.}
\label{tab:biblioindices}
\end{table}

To conclude, bibliometric indices provide a richer, structure-aware perspective compared to conventional performance ratios, revealing whether strong performance is broadly supported by multiple solid contributors or driven by a few dominant assets - an insight essential for assessing portfolio robustness and performance diversification beyond traditional mean-variance allocation.

\subsection{Climate risk resilience rankings for European zones}
This empirical analysis applies the proposed ranking metrics to assess climate risk resilience across European regions. We use annual data from the European Environment Agency on economic losses resulting from weather- and climate-related events, covering 29 European countries over the period 1980–2023 (44 years). Specifically, the dataset records losses $L_{c,t} \in \mathbb{R}_+$ for each country $c$ and year $t$, capturing the financial impact of floods, storms, droughts, and heatwaves.
For regional analysis, countries are grouped into four geographic zones -- Eastern, Northern, Southern, and Western Europe -- as detailed in Table~\ref{tab:regions}.
We compare these regions’ exposure to climate-related damages over time by applying various ranking metrics to their aggregated loss series. 
 \begin{table}[H]
\centering
\tiny{
\begin{tabular}{llll}
\toprule
\textbf{Eastern Europe} & \textbf{Northern Europe} & \textbf{Southern Europe} & \textbf{Western Europe} \\
\midrule
Bulgaria & Denmark     & Greece     & Belgium       \\ 
Czechia  & Estonia     & Spain      & Germany       \\ 
Hungary  & Ireland     & Croatia    & France        \\ 
Poland   & Latvia      & Italy      & Luxembourg    \\ 
Romania  & Lithuania   & Cyprus     & Netherlands   \\ 
Slovakia & Finland     & Malta      & Austria       \\
         & Sweden      & Portugal   & Liechtenstein \\ 
         & Iceland     & Slovenia   & Switzerland   \\ 
         & Norway      & Türkiye    &               \\ 
\bottomrule
\end{tabular}
}
\caption{Country groups by region}
\label{tab:regions}
\end{table}  

Simply using the losses would not allow for a coherent interpretation of the ranking metrics, since lower losses indicate better outcomes. To preserve the interpretability of the ranking metrics, we first need to transform the dataset, as the losses are reported as positive amounts. However, just assigning losses a negative sign would make ranking metrics degenerate. Instead, we mean-center losses by subtracting the overall average across countries and years $\bar{L}$, so that each observation captures its excess loss relative to the European historical mean. We then change the sign to align higher scores with better outcomes, defining for each year $t$
\[
Y_{c,t} \;=\; -\bigl(L_{c,t}-\bar L\bigr),
\qquad
\bar L \;=\; \frac{1}{TN}\sum_{c=1}^{N}\sum_{t=1}^{T} L_{c,t}.
\]
Thus, larger values of \(Y_{c,t}\) indicate better-than-average climate risk resilience, meaning
losses below the European historical mean. We denote by \(Y^z_{c,t}\) the smaller-than-mean losses for the countries in zone $z$ at time $t$.

Aggregate climate risk resilience at the regional level is obtained by summing country-level scores within each geographic zone:
\begin{equation}\label{eq:Xzone}
Y^z_{t} \;=\; \sum_{c \in \mathcal{C}(z)} Y^z_{c,t},
\end{equation}
where \(\mathcal{C}(z)\) denotes the set of countries in zone \(z\).
\paragraph{Traditional and quantile-based ranking metrics} 
Rankings of climate-risk resilience across the four European zones, using RAROC, GLR, Omega, and $\Lambda$-quantiles, are reported in
Figure~\ref{fig:climate_certainty_RAPM}. The metrics are computed as in the previous section, but applied to each zone’s smaller-than-mean loss series \(Y^z_{t}\). 
\begin{figure}[H]
    \centering
    \includegraphics[width=0.48\textwidth, height=4.5cm]{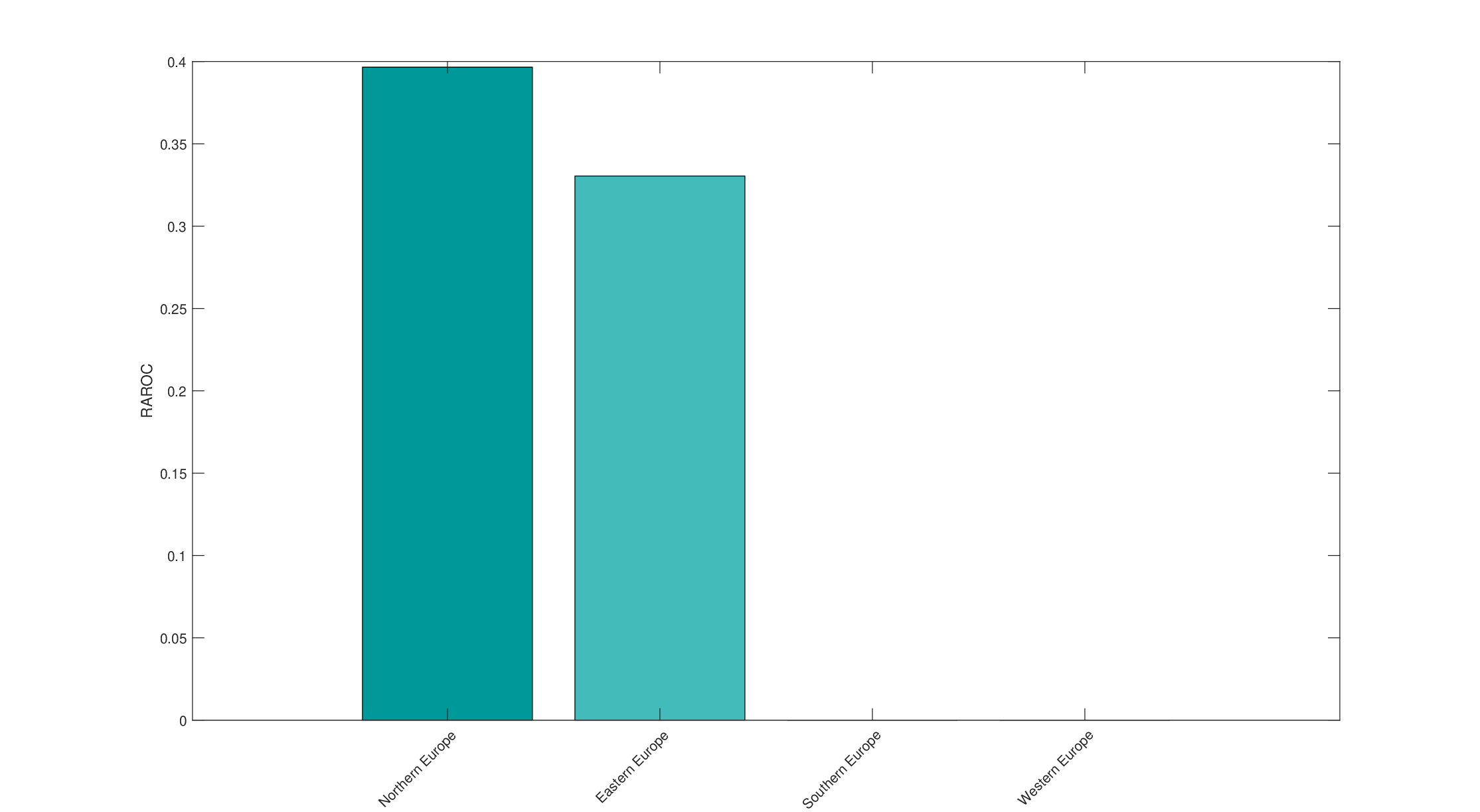}
     \includegraphics[width=0.48\textwidth, height=4.5cm]{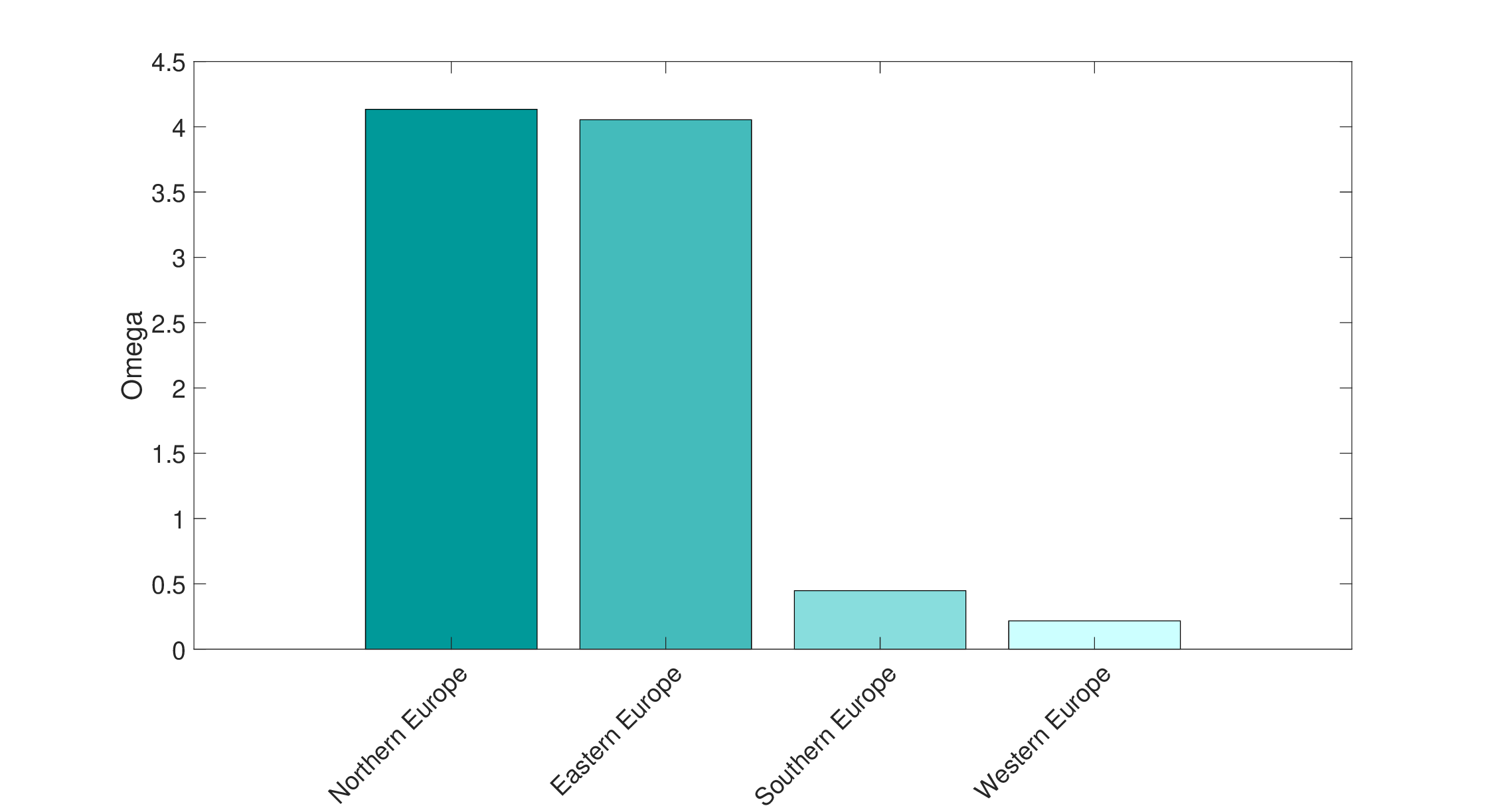}\\       \includegraphics[width=0.48\textwidth, height=4.5cm]{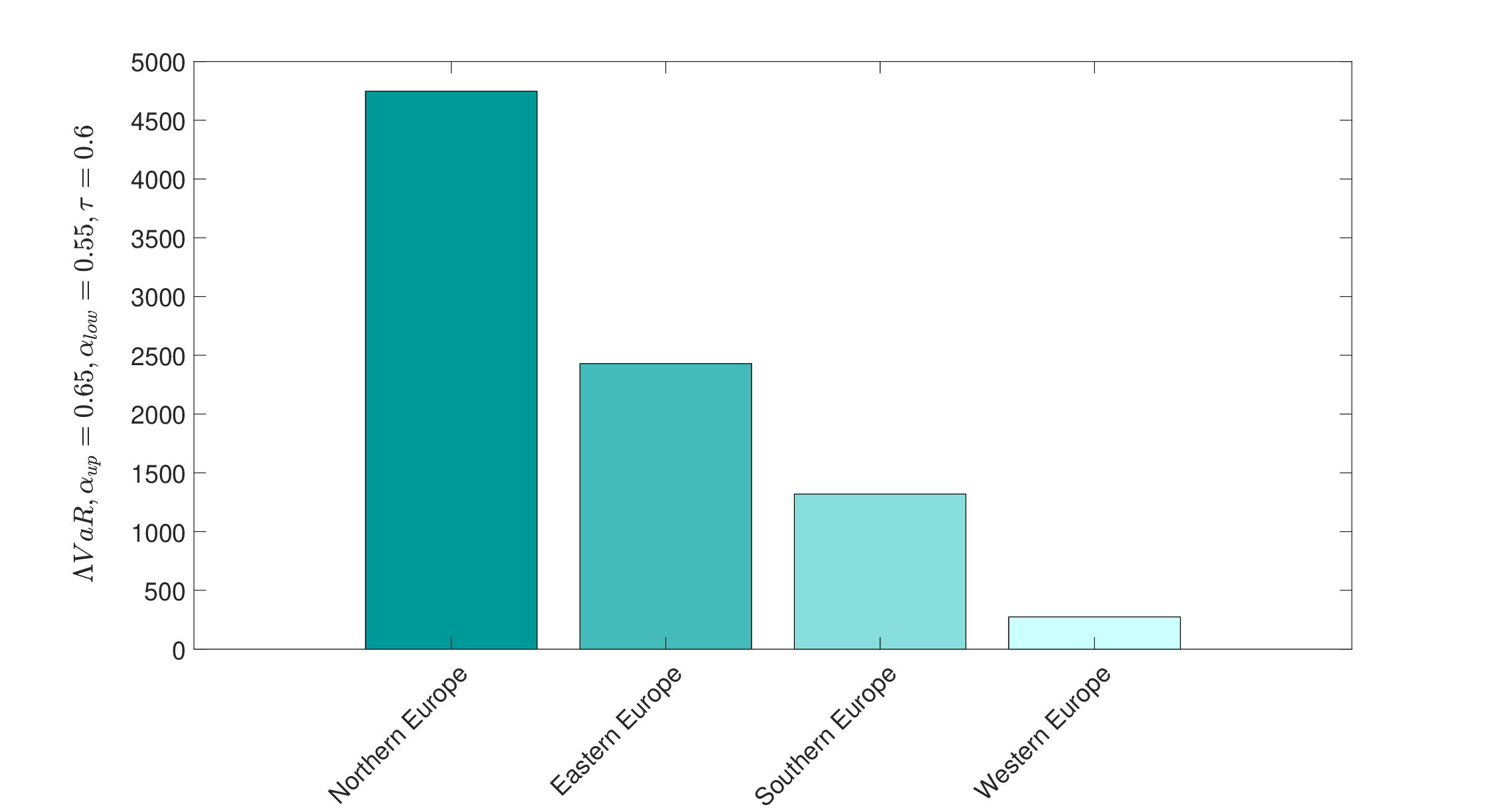}
     \includegraphics[width=0.48\textwidth, height=4.5cm]{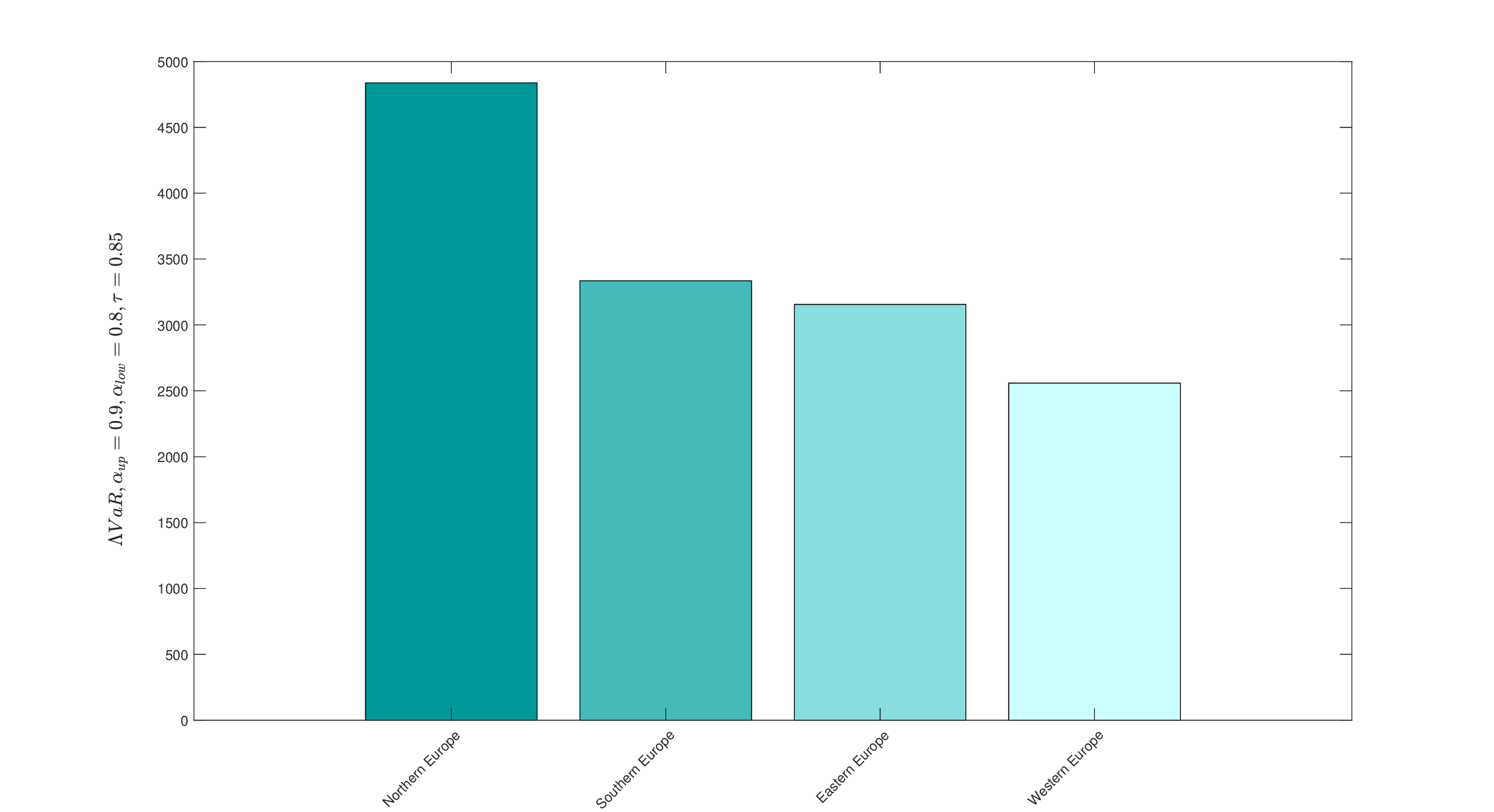}
     \caption{Ranking of European zones climate-related-losses based on RAROC, Omega and $\Lambda$-quantiles.}
	\label{fig:climate_certainty_RAPM}
\end{figure}

\textsc{RAROC}, \textsc{GLR}, $\Omega$, and the $55\% \text{--} 65\% \text{-} \Lambda$-quantiles consistently rank Northern and Eastern Europe as the most resilient zones. For Northern Europe, a \textsc{RAROC} of $0.40$ ($\mathbb{E}[Y] / \text{CVaR}_{\alpha}(Y)$) indicates 40\% smaller-than-mean losses per units of tail-risk exposure. This performance is substantiated by a \textsc{GLR} $\approx 3$, where net expected resilience, accounting for both positive and negative realizations, is three times the magnitude of the expected tail-vulnerability ($\mathbb{E}[Y^-]$). Similarly, an $\Omega \approx 4$ indicates that the expected value of strictly resilient outcomes ($\mathbb{E}[Y^+]$) is four times greater than that of excess-loss events ($\mathbb{E}[Y^-]$). The identity $\Omega - \text{GLR} = 1$ validates the distributional consistency of these results, suggesting that frequent below-average loss years provide a structural buffer that cushions the impact of extreme climate events.

The $55\% \text{--} 65\% \text{-} \Lambda$-quantiles further confirm Northern and Eastern Europe as the regions with the highest near-median resilience. Specifically, Northern Europe reports smaller-than-mean losses of approximately $4,700$ million euros, nearly double the $2,400$ million euros observed for Eastern Europe. In other words, in $55\% \text{--} 65\%$ of the years under observation (i.e., 24--28 years out of 44), Northern Europe maintained losses below the European mean by a margin of approximately $4,700$ million euros, whereas Eastern Europe's margin was $2,400$ million euros.

Northern Europe also emerges as the top performer at the $85\% \text{--} 95\% \text{-} \Lambda$ quantile level, maintaining a resilience margin of approximately $4,500$ million euros. Interestingly, at this higher quantile level, Southern Europe becomes the second-best performer with a margin of $3,500$ million euros, suggesting that while Southern Europe may struggle near the median, it possesses significant tail-resilience in its best-performing years. Conversely, these $\Lambda$-quantiles highlight the structural vulnerabilities of the least-performing regions: Southern and Western Europe at the $55\% \text{--} 65\%$ level, and Western and Eastern Europe at the $85\% \text{--} 95\%$ level.

\paragraph{Bibliometric-indices based ranking metrics}

We next apply a set of bibliometric-inspired ranking metrics, specifically the $h$-index, $h^2$-index, $h_\alpha$-index, and $w$-index, to evaluate regional resilience. Following the approach outlined in the previous section, each zone’s profile is represented by a vector $X^z = \{X^z_{(c)}\}_{c \in \mathcal{C}(z)}$, where components correspond to the sorted values, in decreasing order, of the normalized weighted expected resilience: $X^z_{c}:=\frac{\mathbb{E}[ w_cY^z_c]^+}{\min_{z,c} \mathbb{E}[ w_cY^z_c]^+ }$.
Here, $X^z_c$ measures the expected smaller-than-mean loss of each country relative to the European minimum benchmark. In this calculation, all zones are assigned equal weights, and the expectation is estimated using the $Y^z_c$ historical average\footnote{We acknowledge that differences in the number of countries per region may influence the attainable index values; this could be mitigated by introducing weights proportional to the country count in each zone.}.

Figure~\ref{fig:climate_certainty_h} and Table~\ref{tab:biblio_indices_clim} present the rankings across all four regions. Northern Europe consistently emerges as the top-performing region, ranking highest across all bibliometric indices. This confirms the results obtained from traditional risk-adjusted measures and $\Lambda$-quantiles, while providing deeper insight into the distribution of climate-risk resilience across countries within the zone. Specifically, Northern Europe’s $h$-index of $6$ indicates that at least six countries in the region maintain average resilience levels (smaller-than-average losses) at least six times greater than the continental minimum benchmark.

The remaining indices offer nuanced perspectives on the structural composition of this regional performance. The $h^2$-index emphasizes the contribution of ``top-tier" resilient countries, highlighting those that most significantly outperform the European average, whereas the $h_\alpha$-index ($\alpha=0.5$) provides a more inclusive measure by weighting the contribution of countries with more moderate resilience levels. Furthermore, the $w$-index accounts for the linear distribution of resilience across the zone, rewarding regions with a broad, consistent base of below-average losses. Collectively, these indices suggest that Northern Europe's dominance is not merely the result of a few outliers, but rather a robust, region-wide capacity to maintain damages below the continental mean.

\begin{figure}[H]
    \centering
    \includegraphics[width=0.48\textwidth, height=4.5cm]{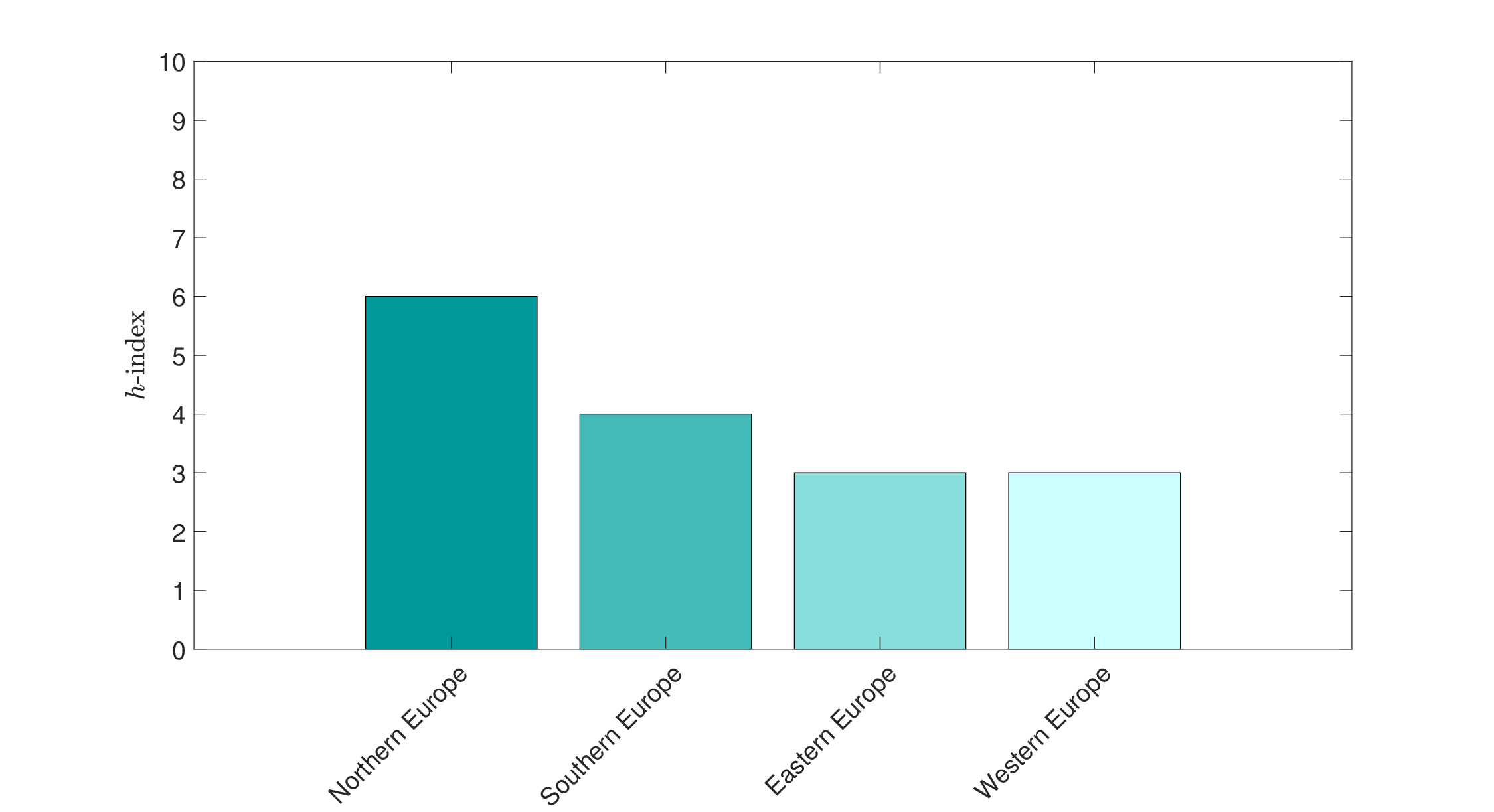}
     \includegraphics[width=0.48\textwidth, height=4.5cm]{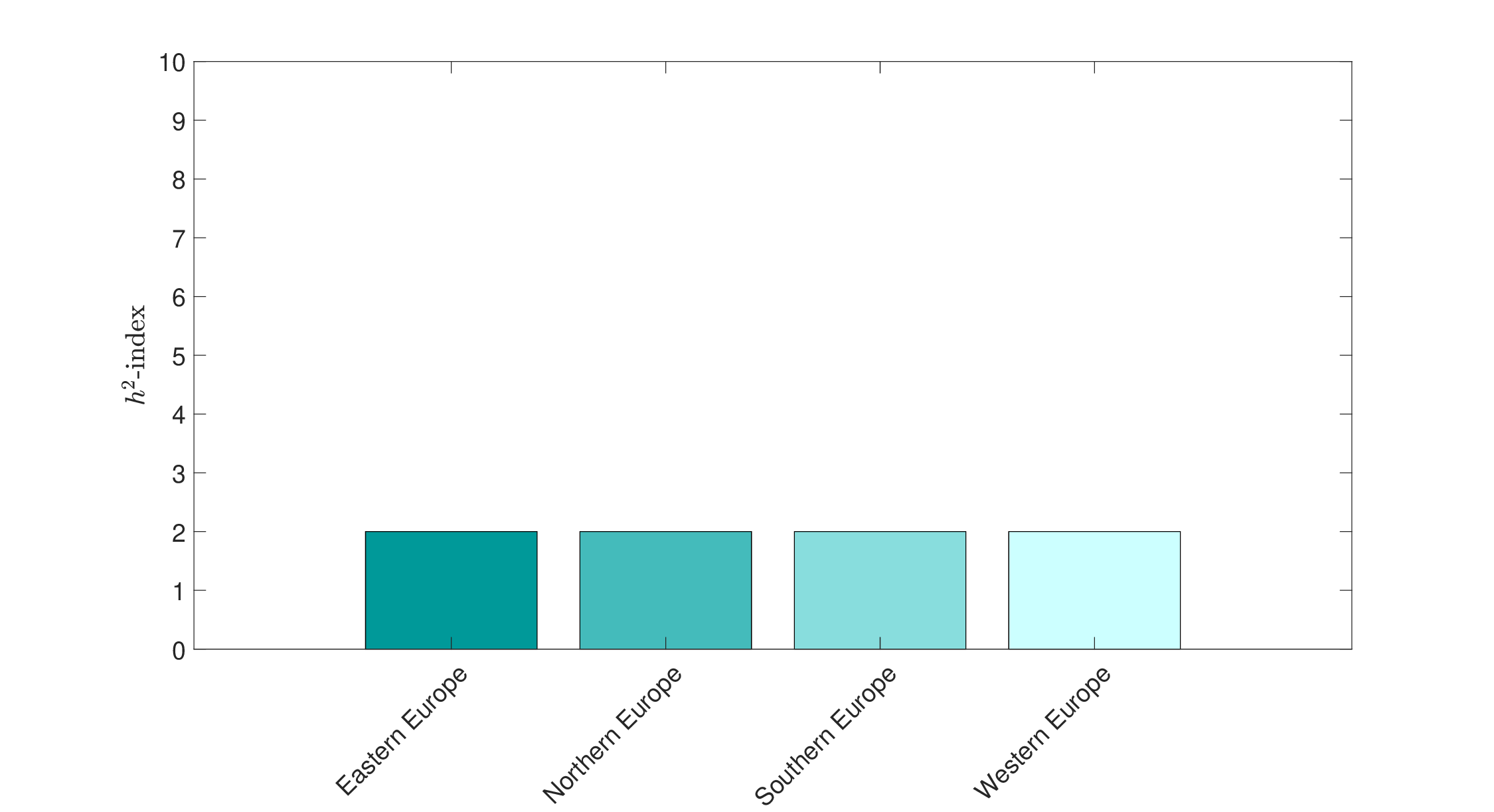}\\       \includegraphics[width=0.48\textwidth, height=4.5cm]{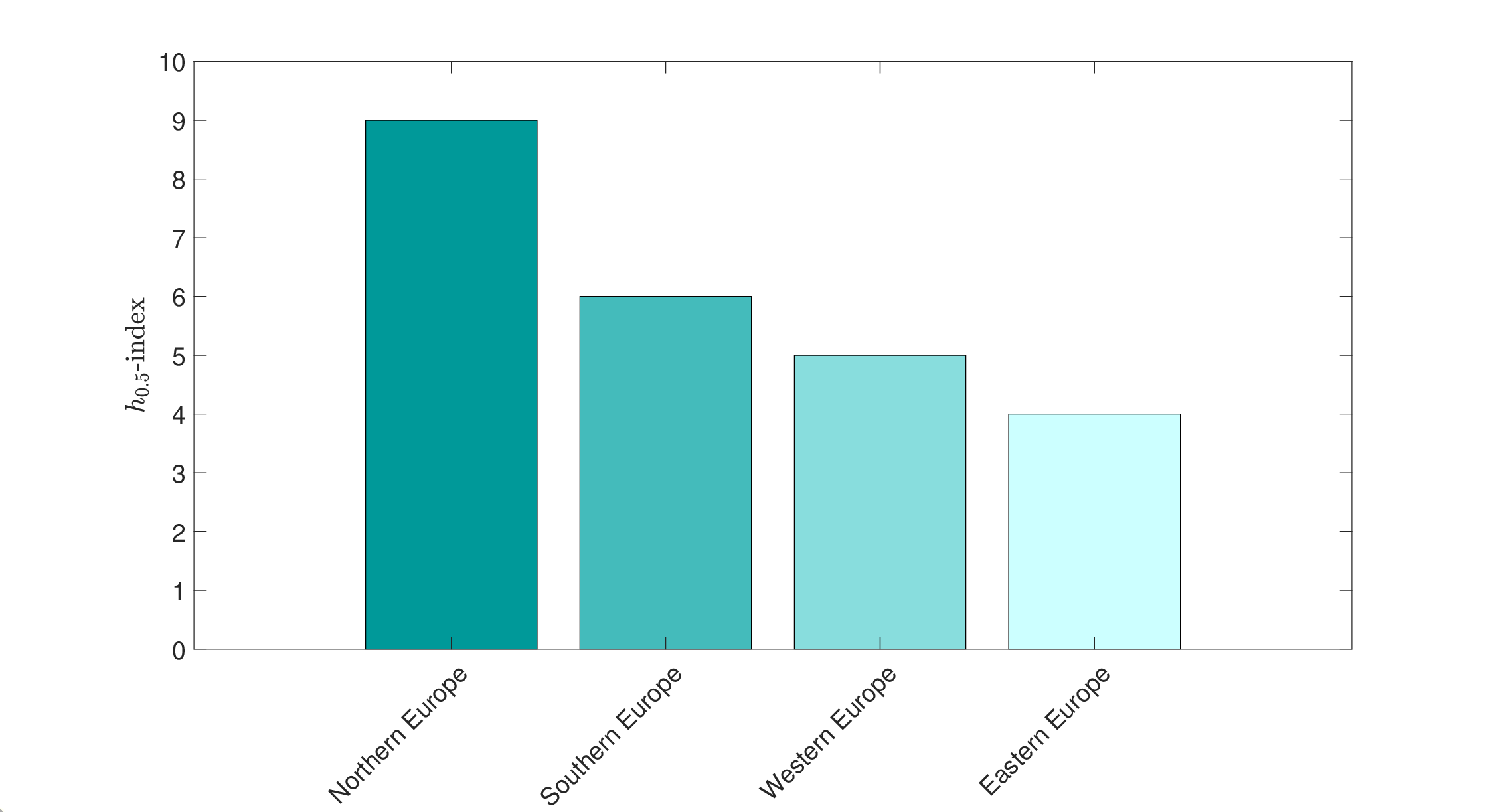}
     \includegraphics[width=0.48\textwidth, height=4.5cm]{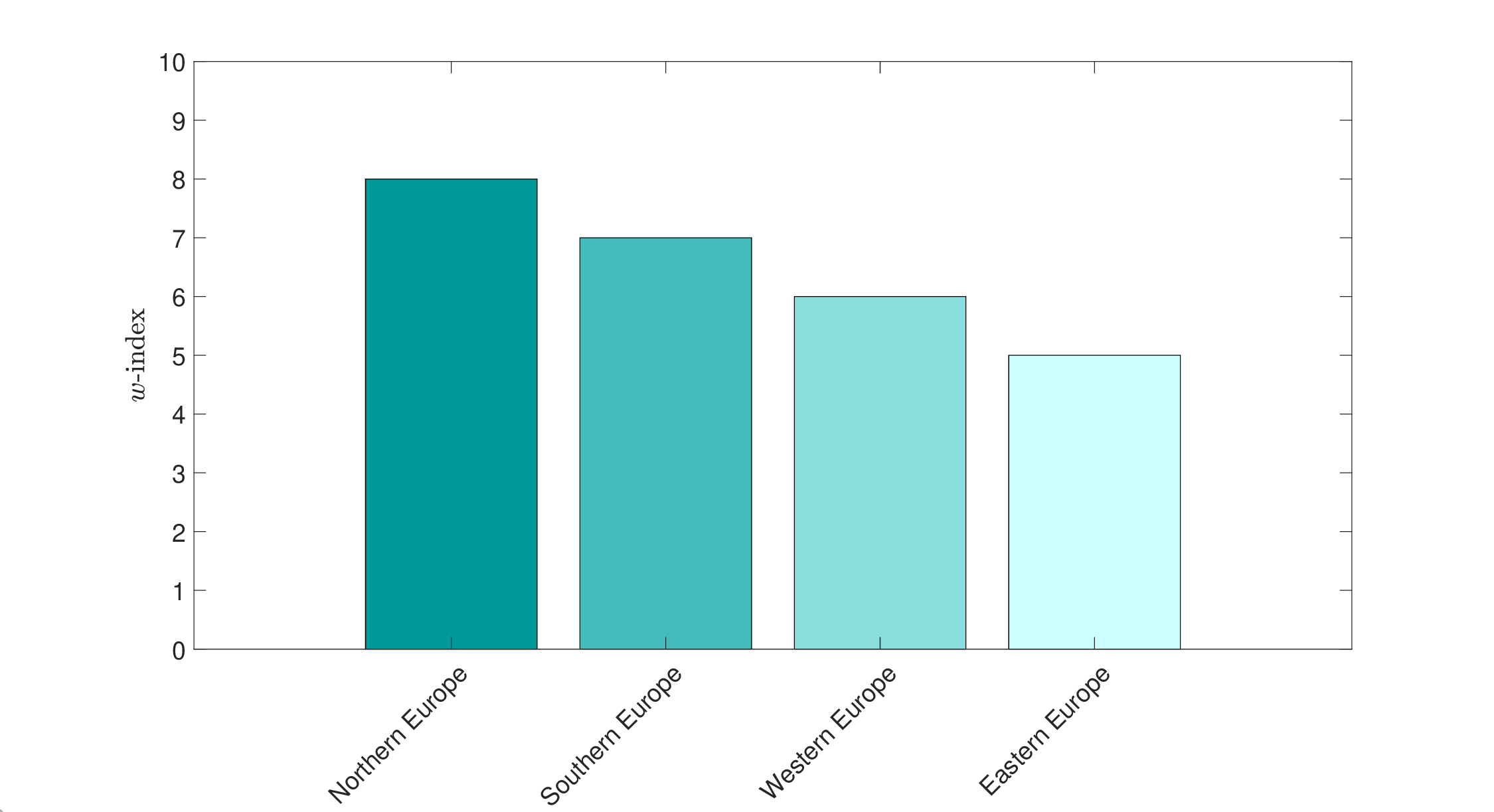}
     \caption{Ranking of European zones climate-related-losses based on bibliometric indices.}
	\label{fig:climate_certainty_h}
\end{figure}

\paragraph{Ranking using certainty equivalent} 
Finally, we implement the ranking metric $r^u$ introduced in Section~\ref{sub_sub:Exp_lossRM} by applying the certainty equivalent formula \eqref{def:sup-cert-eq} to the zonal smaller-than-mean losses $Y^z$ in equation \eqref{eq:Xzone}. For this purpose, we employ a piecewise-linear transformation function $u(y)$, tailored to reflect the socio-economic impact of different magnitudes of smaller-than-mean losses:
\[
u(y) = 
\begin{cases}
	y, & \text{if } y > \theta, \\
	\theta + (1+m)(y - \theta), & \text{if } y \leq \theta,
\end{cases}
\]
where $\theta$ is set at the 75th percentile of the smaller-than-mean losses. Under this specification, resilience levels falling below this threshold are penalized by a factor $m=0.1$, reflecting a decision-theoretic preference for regions that consistently exceed the benchmark by a significant margin (i.e. higher smaller-than-mean losses). This transformation allows us to recover a certainty equivalent value that represents a risk-adjusted level of resilience.

Figure~\ref{fig:climate_certainty} presents the ranking of the four zones based on these evaluations. As expected, Northern Europe emerges as the top-performing zone, exhibiting the highest certainty equivalent. This result indicates that Northern Europe not only achieves the largest average margin of climate-risk resilience above the European mean but also maintains this superiority most consistently when lower resilience outcomes are penalized. These findings reinforce the conclusions drawn from the other ranking metrics, providing a robust decision-theoretic validation of Northern Europe’s leading position in climate resilience.

\begin{figure}[H]
	\centering
	\includegraphics[width=0.8\textwidth]{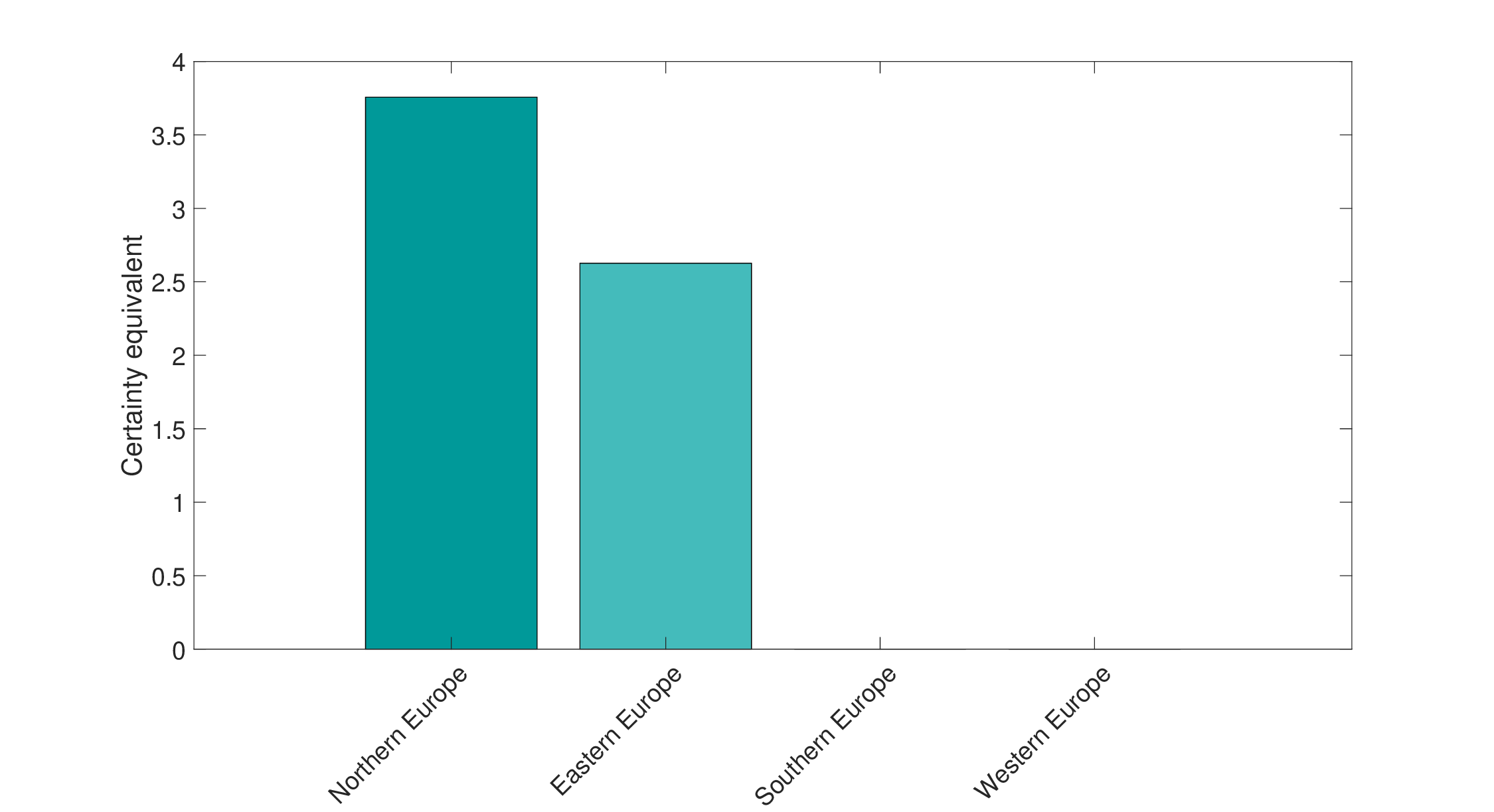}
	\caption{Ranking of zones smaller-than-mean losses based on certainty equivalent (values in thousands)} 	\label{fig:climate_certainty}
\end{figure}

Table \ref{tab:biblio_indices_clim} summarizes all the metrics used for ranking European zones based on climate risk exposure. 
\begin{table}[h!]
\centering
\small{
\begin{tabular}{lcccccccccc}
\toprule
{European } & RAROC & GLR & $\Omega$ & CE & $(0.8-0.9)$ &  $(0.55-0.65)$ & h & $h^2$ & $h_{0.5}$ & w\\
{ zones} &  & & &  &  $\Lambda$VaR &  $\Lambda$VaR & index & index & index & index   \\
\midrule\\

Eastern  & 0.33 & 3.05 & 4.053 & 2.62
& 3155.70 & 2429.70 & 3 & 2 & 4 & 5\\
Northern & 0.40 & 3.11 & 4.134 & 3.76 & 4837.05 & 4747.05 & 6 & 2 &9 & 8 \\
Southern & 0.00 & 0.00 & 0.448 & 0.00 & 3335.05 & 1319.05 & 4 & 2 & 6 & 7 \\
Western  & 0.00 & 0.00 & 0.216 & 0.00 & 2558.60 &  274.60 & 3 & 2 & 5 & 6 \\
\bottomrule
\end{tabular}
}
\caption{Summary of rankings metric values across all European zones. The values of the certainty equivalent (CE) are reported in thousands Euros}
\label{tab:biblio_indices_clim}
\end{table}

\subsection{Portfolio selection using ranking metrics}
In this section, we conduct performance‑based portfolio optimization\footnote{The portfolio optimization problem was solved using a genetic algorithm (GA) implemented in MATLAB's Global Optimization Toolbox. The optimization was performed using a population of 200 candidate solutions and a maximum of 500 generations, allowing for an extensive exploration of the feasible solution space.} using classical ratios (Section~\ref{sec:Rank_Metr_exam}) and the new $\Lambda$‑quantiles and bibliometric‑inspired metrics (Subsection~\ref{sub_biblio_index}). Each optimization problem maximizes a specific ranking metric: RAROC, GLR, Omega, $\Lambda$‑quantiles, the $h$‑index, $h^2$‑index, $h_\alpha$‑index, and $w$‑index. The RAROC's CVaR and the two-levels $\Lambda$‑quantiles are computed using a parametric Student’s‑t specification.
%For completeness, we also perform an in-sample and out-of-sample analysis to examine how the different portfolios are allocated and how they perform. This analysis is intended solely for demonstration purposes.

We aim to construct a portfolio of 14 assets, that are all constituents of the NASDAQ index. Short selling is not allowed. Daily returns are retrieved from Bloomberg for the period Fabruary 17, 2022, to Fabruary 02, 2023.
Descriptive statistics - exhibiting deviations from normality - are reported in Table~\ref{tab:stat} for the in-sample period.  
\begin{table}[h]
		\centering
	
		\label{tab:descriptive_stats}
        \begin{adjustbox}{width=0.7\textwidth}
		\begin{tabular}{lcccccc}
			\toprule
			\textbf{Ticker} & \textbf{annual $\mu$} & \textbf{annual $\sigma$} & \textbf{Skewness} & \textbf{Kurtosis} & \textbf{p-value} & \textbf{J-B Test} \\
			\midrule
OKE UN Equity  & 0.05610 & 0.54053 & -3.51744 & 55.02362 & 0.001 & 122639.60 \\
TRGP UN Equity & 0.16666 & 0.69805 & -5.73717 & 95.65334 & 0.001 & 387875.42 \\
HES UN Equity  & 0.28624 & 0.52940 & -1.33776 & 26.47235 & 0.001 & 24835.88 \\
PSX UN Equity  & 0.03209 & 0.43804 & -0.08203 & 10.18874 & 0.001 & 2300.87 \\
DG FP Equity   & 0.08901 & 0.32533 & -0.59715 & 19.94834 & 0.001 & 12845.92 \\
TTE FP Equity  & 0.04659 & 0.33086 & -0.70651 & 17.85473 & 0.001 & 9908.35 \\
DSY FP Equity  & 0.14483 & 0.30539 & 0.32109  & 9.75550  & 0.001 & 2049.19 \\
AI FP Equity   & 0.11312 & 0.21967 & -0.60074 & 10.99649 & 0.001 & 2909.74 \\
KER FP Equity  & 0.08767 & 0.32660 & -0.26381 & 7.13213  & 0.001 & 772.20  \\
RNO FP Equity  & -0.07849 & 0.49470 & -0.44762 & 10.76494 & 0.001 & 2718.76 \\
MSFT UQ Equity & 0.21853 & 0.31134 & -0.31669 & 10.87781 & 0.001 & 2779.51 \\
CRM UN Equity  & 0.05200 & 0.39157 & 0.23653  & 13.78702 & 0.001 & 5187.97 \\
INTC UQ Equity & -0.10226 & 0.38313 & -0.78455 & 15.38322 & 0.001 & 6933.38 \\
VZ UN Equity   & -0.07101 & 0.19648 & -0.12089 & 8.60786  & 0.001 & 1402.04 \\
			\bottomrule
		\end{tabular}
        \end{adjustbox}
        	\caption{Descriptive Statistics of portfolio components}
            \label{tab:stat}
	\end{table}

The in-sample data is used to obtain the optimal portfolios. In particular, we generate $1000$ random initial portfolio weight vectors, %using the Dirichlet distribution, 
ensuring that each vector lies in the unit simplex. For each initial guess $w_0$, we solve a separate portfolio optimization problem for each performance metric using a numerical global optimization routine. For each metric, the objective function is maximized over the feasible set using a derivative-free optimization algorithm. Among all solutions, we select the weight vector achieving the highest objective value; if multiple optimal solutions are found, their average is used.

Figure~\ref{fig:asset_allocation} displays the asset allocations for the optimized portfolios across the different ranking metrics. The stacked bar chart highlights how each metric leads to distinct portfolio compositions.
\begin{figure}[h]
    \centering
    \includegraphics[width=0.7\textwidth]{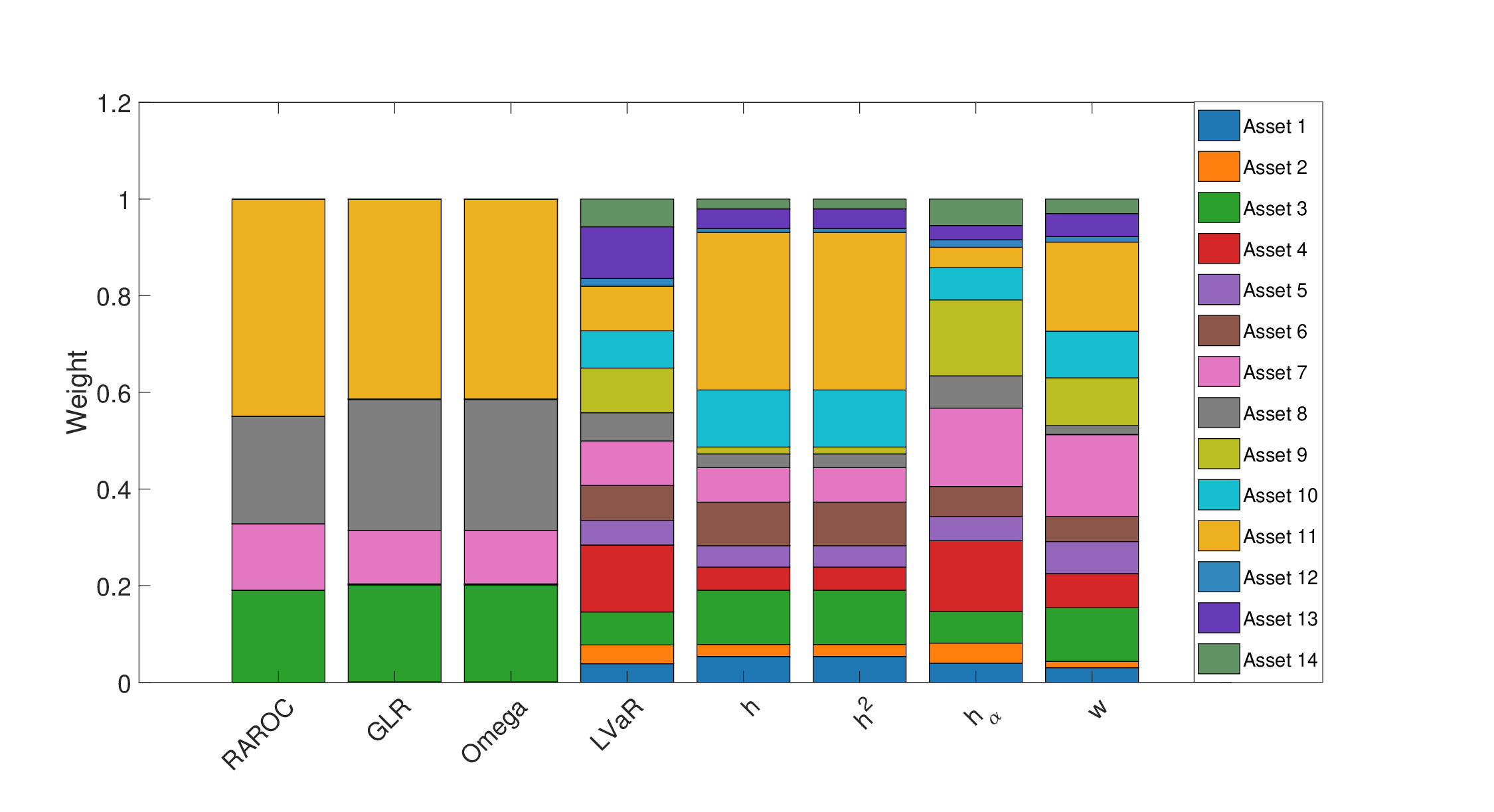}
    \caption{Portfolio allocation using different ranking metrics}
    \label{fig:asset_allocation}
\end{figure}

Figure~\ref{fig:asset_allocation} illustrates that the optimal portfolios derived by maximizing bibliometric indices tend to be more diversified than those obtained using traditional financial ratios. Each ranking metric guides portfolio construction toward specific objectives, with bibliometric-inspired measures aiming for a specific distribution of portfolio returns across assets. The 90\%-99\%-$\Lambda$ quantile optimization strategy builds a portfolio which maximize the return at the 90\%-99\% quantile level, i.e. the portfolio composition with the highest performance at 90\%-99\% confidence level. The $h$-index optimization strategy aims to construct a portfolio in which the largest possible number of assets $q$ each achieve returns of at least $q$ times the benchmark, capturing both the magnitude and breadth of performance across the asset set. The $h^2$-index, by contrast, favors portfolios where a small subset of top-performing assets concentrates the majority of the portfolio return, thus emphasizing strong performers. The $h_\alpha$ index with $\alpha = 0.5$ targets a middle ground by identifying the largest number of assets $q$ for which each asset achieves at least $0.5 \cdot q$ in average returns. This effectively enforces a minimum return floor across the portfolio, offering a balance between robustness and diversification. As expected, the asset allocation produced by the $h_\alpha$ index differs from those derived using the $h$ and $h^2$ indices due to its distinct return threshold logic. The $w$-index, on the other hand, aims to construct a portfolio with a linearly decreasing return distribution-starting from the highest return and tapering down across assets-thereby encouraging a smooth performance gradient across constituents. Finally, by modifying the performance curves $f_x$ in formula \eqref{eq:SRM}, one can design a variety of performance measures that shape the distribution of returns across portfolio assets in targeted ways. This flexibility allows for the construction of portfolios that align more closely with specific investment goals or risk preferences.

%While this behavior holds for the dataset analyzed here, a more comprehensive empirical evaluation across different market environments and datasets would be needed to assess the robustness of these findings. We leave this direction for future research.

%Once the optimal in-sample weights are computed, we assess their effectiveness using a rolling-window approach in a fixed out-of-sample evaluation. That is, we apply the in-sample optimal weights to a subsequent investment period-February 1, 2023 to January 31, 2024. Each strategy's portfolio is initialized with an investment of 100 euro, and its value evolves daily according to realized out-of-sample asset returns. \blue{[Dove mostriamo i risultati di questo?]}

%\blue{One could think about strategies maximixing RAROC while controlling for the $h$-index for instance.}

%\blue{[Lila metti una tabella che display performance indices and the other perf measures for each portfolio]}

\section{Conclusion}
This paper establishes a unified axiomatic foundation for ranking metrics, extending the classical theory of risk-adjusted performance measures and acceptability indices \citep{ChernyMadan2009, DrapeauKupper2013, Righi2024}. The framework replaces the traditional reward-per-unit-of-risk interpretation with a more general notion of performance rank, formalised through monotonicity and the newly introduced cash-quasiconcavity (or quasiconcavity at constants) as concave counterpart of the quasi-starshapedness by \cite{Han2024}. This property constitutes a minimal requirement ensuring that mixing a risky position with a deterministic payoff cannot worsen its ranking and captures the fundamental diversification intuition underlying most reward–risk measures.

We further elucidate its link with classical properties, showing that for maps normalised at zero, concavity -- and hence quasi-concavity -- at zero is equivalent to $[0,1]$-super-linearity, while cash-additivity bridges these local properties to cash-concavity and cash-quasiconcavity. Within this structure, any ranking metric admits dual representations via cash-convex acceptance sets or cash-quasiconvex risk measures, providing a broad framework encompassing traditional ratios (RAROC, GLR, Omega) as well as new, monetary and non-monetary performance measures. Beyond known measures, this framework introduces new families of performance metrics derived from expected-loss/ certainty-equivalent, and $\Lambda$-quantile risk measures, and, notably, a novel class of bibliometric-based ranking metrics inspired by the 
$h$-index and its extensions. These bibliometric metrics provide an entirely new approach to portfolio performance evaluation compared to  traditional ratio-based assessments by focusing on the distributional characteristics of portfolio constituents rather than their risk adjusted mean performance. Empirical analyses on portfolio performance and climate-risk resilience demonstrate the framework’s practical relevance and reveal that alternative ranking metrics can induce materially distinct rankings, implicating the importance of aligning ranking rules with specific performance objectives.

%Overall, the results highlight that ranking metrics constitute a flexible and unifying tool for comparative performance evaluation — capable of encompassing both monetary and non-monetary performance indicators within a single theoretical scheme. 
Future research could explore bibliometric implementations using different portfolio statistics and studying the relationship between 
$h$-index-based and quantile-based metrics, design portfolio optimisation strategies that maximise traditional ratios while controlling for bibliometric indices, and, extending the framework to dynamic, robust, or multivariate settings to further enhance its applicability to ESG-oriented decision environments.

\bibliographystyle{chicago}
\bibliography{MyrefHMPR}

@article{Artzner1999,
  title={Coherent measures of risk},
  author={Artzner, Philippe and Delbaen, Freddy and Eber, Jean-Marc and Heath, David},
  journal={Mathematical Finance},
  volume={9},
  number={3},
  pages={203--228},
  year={1999},
  publisher={Wiley Online Library}
}

@article{Bernardo2000,
  title={Gain, loss, and asset pricing},
  author={Bernardo, Antonio E and Ledoit, Olivier},
  journal={Journal of Political Economy},
  volume={108},
  number={1},
  pages={144--172},
  year={2000},
  publisher={The University of Chicago Press}
}

@article{Castagnoli2022,
  title={Star-shaped risk measures},
  author={Castagnoli, Erio and Cattelan, Giacomo and Maccheroni, Fabio and Tebaldi, Claudio and Wang, Ruodu},
  journal={Operations Research},
  volume={70},
  number={5},
  pages={2637--2654},
  year={2022},
  publisher={INFORMS}
}

@article{ChernyMadan2009,
  title={New measures for performance evaluation},
  author={Cherny, Alexander and Madan, Dilip},
  journal={The Review of Financial Studies},
  volume={22},
  number={7},
  pages={2571--2606},
  year={2009},
  publisher={Society for Financial Studies}
}

@article{DrapeauKupper2013,
  author    = {Drapeau, Samuel and Kupper, Michael},
  title     = {Risk preferences and their robust representation},
  journal   = {Mathematics of Operations Research},
  year      = {2013},
  volume    = {38},
  number    = {1},
  pages     = {28--62},
  doi       = {10.1287/moor.1120.0564}
}

@article{Frittelli2014,
  title={Risk measures on P(R) and value at risk with probability/loss function},
  author={Frittelli, Marco and Maggis, Marco and Peri, Ilaria},
  journal={Mathematical Finance},
  volume={24},
  number={3},
  pages={442--463},
  year={2014},
  publisher={Wiley Online Library}
}

@article{Frittelli2014SRM,
  title={Scientific research measures},
  author={Frittelli, Marco and Mancini, Loriano and Peri, Ilaria},
  journal={Journal of the Association for Information Science and Technology},
  volume={67},
  number={12},
  pages={3051--3063},
  year={2016},
  publisher={Wiley Online Library}
}

@article{FollmerSchied2002,
  title={Convex measures of risk and trading constraints},
  author={F{\"o}llmer, Hans and Schied, Alexander},
  journal={Finance and Stochastics},
  volume={6},
  pages={429--447},
  year={2002},
  publisher={Springer}
}

@misc{Han2024,
      title={Cash-subadditive risk measures without quasi-convexity}, 
      author={Xia Han and Qiuqi Wang and Ruodu Wang and Jianming Xia},
      year={2024},
      eprint={2110.12198},
      archivePrefix={arXiv},
      primaryClass={q-fin.RM},
      url={https://arxiv.org/abs/2110.12198}, 
}

@phdthesis{PeriThesis,
  author    = {Ilaria Peri},
  title     = {Quasi-convex Risk Measures and Acceptability Indices: Theory and Applications},
  school    = {Università degli Studi di Milano-Bicocca},
  year      = {2012},
  type      = {Ph.D. Thesis},
  address   = {Milan, Italy}
}

@article{Righi2024,
  title={Star-shaped acceptability indexes},
  author={Righi, Marcelo Brutti},
  journal={Insurance: Mathematics and Economics},
  volume={117},
  pages={170--181},
  year={2024},
  publisher={Elsevier}
}

@article{frittelli2002putting,
  title={Putting order in risk measures},
  author={Frittelli, Marco and Gianin, Emanuela Rosazza},
  journal={Journal of Banking \& Finance},
  volume={26},
  number={7},
  pages={1473--1486},
  year={2002},
  publisher={Elsevier}
}

@article{hirsch2005index,
  title={An index to quantify an individual's scientific research output},
  author={Hirsch, Jorge E},
  journal={Proceedings of the National academy of Sciences},
  volume={102},
  number={46},
  pages={16569--16572},
  year={2005},
  publisher={National Academy of Sciences}
}

@article{Kosmulski_ISSI2006,
  title={A new Hirsch-type index saves time and works equally well as the original h-index},
  author={Kosmulski, Marek and others},
  journal={ISSI newsletter},
  volume={2},
  number={3},
  pages={4--6},
  year={2006}
}

@article{Woeginger_MSS2008,
  title={An axiomatic characterization of the Hirsch-index},
  author={Woeginger, Gerhard J},
  journal={Mathematical Social Sciences},
  volume={56},
  number={2},
  pages={224--232},
  year={2008},
  publisher={Elsevier}
}

@article{Eck_Waltman_JI2008,
  title={Generalizing the h-and g-indices},
  author={van Eck, Nees Jan and Waltman, Ludo},
  journal={Journal of Informetrics},
  volume={2},
  number={4},
  pages={263--271},
  year={2008},
  publisher={Elsevier}
}

@article{frittelli2016scientific,
  title={Scientific research measures},
  author={Frittelli, Marco and Mancini, Loriano and Peri, Ilaria},
  journal={Journal of the Association for Information Science and Technology},
  volume={67},
  number={12},
  pages={3051--3063},
  year={2016},
  publisher={Wiley Online Library}
}

@article{frittelli2002R,
  title={Putting order in risk measures},
  author={Frittelli, Marco and Gianin, Emanuela Rosazza},
  journal={Journal of Banking \& Finance},
  volume={26},
  number={7},
  pages={1473--1486},
  year={2002},
  publisher={Elsevier}
}

@article{RosazzaSgarra2013,
  title={Acceptability indexes via-expectations: an application to liquidity risk},
  author={Rosazza Gianin, Emanuela and Sgarra, Carlo},
  journal={Mathematics and financial economics},
  volume={7},
  number={4},
  pages={457--475},
  year={2013},
  publisher={Springer}
}

@article{BelliniPeri2022,
  title={An axiomatization of $\Lambda$-quantiles},
  author={Bellini, Fabio and Peri, Ilaria},
  journal={SIAM Journal on Financial Mathematics},
  volume={13},
  number={1},
  pages={SC26--SC38},
  year={2022},
  publisher={SIAM}
}

@article{hansen2020starshaped,
  title={Starshaped sets},
  author={Hansen, Guillermo and Herburt, Irmina and Martini, Horst and Moszy{\'n}ska, Maria},
  journal={Aequationes mathematicae},
  volume={94},
  number={6},
  pages={1001--1092},
  year={2020},
  publisher={Springer}
}

@article{keating2002introduction,
  title={An introduction to omega},
  author={Keating, Con and Shadwick, William F},
  journal={AIMA newsletter},
  year={2002}
}

@article{zaik1996raroc,
  title={RAROC at Bank of America: from theory to practice},
  author={Zaik, Edward and Walter, John and Retting, Gabriela and James, Christopher},
  journal={Journal of applied corporate finance},
  volume={9},
  number={2},
  pages={83--93},
  year={1996},
  publisher={Wiley Online Library}
}

\end{document}